\documentclass[11pt]{article}
\usepackage[utf8]{inputenc}
\usepackage[
backend=bibtex,
url=false,isbn=false,doi=false,
date=year,
hyperref=auto,
style=alphabetic,
natbib=true,
sorting=nty,
firstinits=true,
maxnames=2, maxbibnames=10,
]{biblatex}

\bibliography{./references.bib} 

\usepackage{hyperref}
\hypersetup{
    bookmarks=true,         
    unicode=false,          
    pdftoolbar=true,        
    pdfmenubar=true,        
    pdffitwindow=false,     
    pdfstartview={FitH},    
    pdfauthor={Reinhard Heckel},     
    pdfsubject={Subject},   
    pdfcreator={Reinhard Heckel},   
    pdfproducer={Producer}, 
    pdfnewwindow=true,      
    colorlinks=true,       
    linkcolor=red,          
    citecolor=green,        
    filecolor=magenta,      
    urlcolor=cyan           
}

\usepackage[utf8]{inputenc} 
\usepackage[T1]{fontenc}    
\usepackage{hyperref}       
\usepackage{url}            
\usepackage{booktabs}       
\usepackage{amsfonts}       
\usepackage{nicefrac}       
\usepackage{microtype}      
\usepackage[ruled,vlined]{algorithm2e} 
\usepackage{rays_defs_18}

\usepackage{graphicx}
\usepackage{caption}
\usepackage{subcaption}
\usepackage{amssymb}
\usepackage{bbm}
\usepackage[normalem]{ulem}
\usepackage{enumitem}
\usepackage{float}

\usepackage{amsmath}
\usepackage{amsthm}
\theoremstyle{plain}

\newtheorem{proposition}{Proposition}[section]

\theoremstyle{remark}

\usepackage{blkarray}
\usepackage{mwe,subcaption,tikz}
\tikzset{boximg/.style={remember picture,red,thick,draw,inner sep=0pt,outer sep=0pt}}
\usepackage{pgfplots}
\usetikzlibrary{pgfplots.statistics}
\usepackage{pgfplotstable}
\usepgfplotslibrary{colormaps}
\usetikzlibrary{quotes,angles}
\usepackage{textpos}
\usepackage{pgfplotstable}
\usepackage{array}
\usepackage{authblk}
\usepackage{multirow}
\usepgfplotslibrary{fillbetween}
\usepackage{enumitem}

\usepackage{threeparttable}
\usepackage{comment}
\usepackage{booktabs}
\usepackage{colortbl} 
\usepackage{xcolor} 
\usepackage{xfrac}
\usepackage{caption} 
\usetikzlibrary{matrix, tikzmark, decorations.pathreplacing, calc}
\pgfplotsset{compat=1.14}
\usepgfplotslibrary{groupplots}
\usepgfplotslibrary{dateplot}
\usepackage[export]{adjustbox}

\DeclareMathAlphabet{\mathcal}{OMS}{cmsy}{m}{n}
\SetMathAlphabet{\mathcal}{bold}{OMS}{cmsy}{b}{n}
\providecommand{\abs}[1]{\lvert#1\rvert}

\long\def\comment#1{}
\newcommand{\blue}[1]{\textcolor{blue}{#1}}

\newcommand{\rh}[1]{{\bf{{\blue{{RH --- #1}}}}}}

\usetikzlibrary{backgrounds,
    calc,
    fadings,
    shadows,
    shapes.arrows,
    shapes.symbols
}

\usepackage{xparse}
\edef\RestoreEndlinechar{%
    \endlinechar=\the\endlinechar\relax
}
\usepackage{smartdiagram}
\RestoreEndlinechar

\tikzset{bubble node/.append style={
        draw=none, opacity=0.5
    }
}

\usepackage{adjustbox}
\usepackage{booktabs}
\usepackage{relsize,stackengine}

\usepackage{ragged2e,tabularx}
\newcolumntype{L}[1]{>{\raggedright\let\newline\\\arraybackslash\hspace{0pt}}p{#1}}

\newcommand{\thickhline}{%
    \noalign {\ifnum 0=`}\fi \hrule height 1.5pt
    \futurelet \reserved@a \@xhline
}

\usetikzlibrary{3d,decorations.text,shapes.arrows,positioning,fit,backgrounds}
\tikzset{pics/fake box/.style args={
#1 with dimensions #2 and #3 and #4}{
code={
\draw[ultra thin,fill=#1]  (0,0,0) coordinate(-front-bottom-left) to
++ (0,#3,0) coordinate(-front-top-right) --++
(#2,0,0) coordinate(-front-top-right) --++ (0,-#3,0) 
coordinate(-front-bottom-right) -- cycle;
\draw[ultra thin,fill=#1] (0,#3,0)  --++ 
 (0,0,#4) coordinate(-back-top-left) --++ (#2,0,0) 
 coordinate(-back-top-right) --++ (0,0,-#4)  -- cycle;
\draw[ultra thin,fill=#1!80!black] (#2,0,0) --++ (0,0,#4) coordinate(-back-bottom-right)
--++ (0,#3,0) --++ (0,0,-#4) -- cycle;
}
}}
\tikzset{pics/empty fake box/.style args={
#1 with dimensions #2 and #3 and #4}{
code={
\draw[ultra thin,fill=#1]  (0,0,0) coordinate(-front-bottom-left) to
++ (0,#3,0) coordinate(-front-top-right) --++
(#2,0,0) coordinate(-front-top-right) --++ (0,-#3,0) 
coordinate(-front-bottom-right) -- cycle;
\draw[ultra thin,fill=#1] (0,#3,0)  --++ 
 (0,0,#4) coordinate(-back-top-left) --++ (#2,0,0) 
 coordinate(-back-top-right) --++ (0,0,-#4)  -- cycle;
\draw[ultra thin,fill=#1!80!black] (#2,0,0) --++ (0,0,#4) coordinate(-back-bottom-right)
--++ (0,#3,0) --++ (0,0,-#4) -- cycle;
}
}}
\definecolor{electriclavender}{rgb}{0.96, 0.73, 1.0}
\definecolor{grannysmithapple}{rgb}{0.66, 0.89, 0.63}
\definecolor{aliceblue}{rgb}{0.94, 0.97, 1.0}

\definecolor{steelblue}{rgb}{0.27, 0.51, 0.71}
\definecolor{amber}{rgb}{1.0, 0.49, 0.0}
\definecolor{caribbeangreen}{rgb}{0.0, 0.8, 0.6}
\definecolor{ash}{rgb}{0.7, 0.75, 0.71}
\definecolor{brightpink}{rgb}{1.0, 0.0, 0.5}

\newcommand{\sblue}[1]{\textcolor{steelblue!75!black}{#1}}
\newcommand{\mzdcm}[1]{{\bf{{\sblue{{MZD --- #1}}}}}}

\usepackage{bm}
\newcommand\vtheta{{\bm \theta}}
\newcommand\vvartheta{{\bm \vartheta}}
\newcommand\vbeta{{\bm \beta}}

\title{\textbf{Test-Time Training Can Close the Natural Distribution Shift Performance Gap in Deep Learning Based Compressed Sensing}}
\author{}
\date{}

\begin{document}

\begin{center}

{\bf{\LARGE{
Test-Time Training Can Close the Natural Distribution Shift Performance Gap in Deep Learning Based Compressed Sensing
}}}

\vspace*{.2in}

{\large{
\begin{tabular}{cccc}
Mohammad Zalbagi Darestani$^{\ast}$,
Jiayu Liu$^{\dagger}$, 
and 
Reinhard Heckel$^{\dagger,\ast}$
\end{tabular}
}}

\vspace*{.05in}

\begin{tabular}{c}
$^\ast$Dept. of Electrical and Computer Engineering, Rice University\\
$^\dagger$Dept. of Electrical and Computer Engineering, Technical University of Munich
\end{tabular}

\vspace*{.1in}

\today

\vspace*{.1in}

\end{center}

\begin{abstract}
    Deep learning based image reconstruction methods outperform traditional methods. 
    However, neural networks suffer from a performance drop when applied to images from a different distribution than the training images. 
    For example, a model trained for reconstructing knees in accelerated magnetic resonance imaging (MRI) does not  reconstruct brains well, even though the same network trained on brains reconstructs brains perfectly well. 
    Thus there is a distribution shift performance gap for a given neural network, defined as the difference in performance when training on a distribution $P$ and training on another distribution $Q$, and evaluating both models on $Q$. 
    In this work, we propose a domain adaptation method for deep learning based compressive sensing that relies on self-supervision during training paired with test-time training at inference. We show that for four natural distribution shifts, this method essentially closes the distribution shift performance gap for state-of-the-art architectures for accelerated MRI. 
\end{abstract}

\section{Introduction}\label{sec:intro}
Deep learning methods enable fast and accurate image reconstruction and outperform traditional methods on a variety of imaging tasks~\citep{dong2014learning,jin2017deep,zhang2017beyond,sriram2020end,rivenson2018phase,jalal2021robust}. 
Performance is typically measured as in-distribution performance: 
A dataset is split into test and training sets, and a method trained on the training set is evaluated on the test set.

In practice, however, the  train and test distributions are usually different: For example, we train a network on data from one hospital, and apply the network to data from a different hospital. Or we train on data acquired with one scanner type and acquisition mode, and apply it to a different scanner type or acquisition mode.

Deep learning imaging methods perform significantly worse under such distribution shifts. For accelerated MRI, a medical imaging technique, deep learning methods incur a significant accuracy drop when shifting from one distribution to another, as shown for three natural distribution shifts by~\citet{darestani2021measuring} and for SNR changes by~\citet{knoll2019assessment}. 

The part of this accuracy drop that can be overcome in principle can be measured by the \emph{distribution shift performance gap}:
Suppose $P$ and $Q$ denote the train and test distributions. We define the distribution shift performance gap as the reconstruction accuracy (measured by a standard metric, e.g., the SSIM score) of training on $Q$ and testing on $Q$ minus the reconstruction accuracy of training on $P$ and testing on $Q$..  


In this paper, we propose a novel domain adaptation method for deep learning based compressive sensing, and show that it overcomes the gap caused by four natural distribution shifts in accelerated MRI. Our approach consists of two parts: (1) including self-supervision during the supervised training stage of deep learning models, and (2) performing self-supervised test-time training for each new test sample at inference.

    We show that our method works with two well-known network architectures: the baseline U-Net~\citep{ronneberger2015u} and the state-of-the-art end-to-end variational network~\citep{sriram2020end}. 
    We evaluate robustness under four natural distribution shifts, illustrated in Figure~\ref{fig:shifts}:
    (1) anatomy shift, 
    (2) dataset shift where the training set comes from a different hospital than the test set, 
    (3) modality shift where the acquisition mode changes, and 
    (4) acceleration shift where the acceleration factor changes. 
    For the U-Net, our method closes the distribution shift performance gap by 98.6, 87.4, 96.3\%, 97.4\% (see Table~\ref{tab:gap}) which in each case yields a significant increase in image quality. 
    

\begin{table}[t]
\setlength{\tabcolsep}{15pt}
\centering
\begin{tabular}{ccccc}
  & anatomy shift & dataset shift & modality shift & \hspace{-17pt} acceleration shift 
  \\[5pt]
  & \footnotesize knee & \footnotesize fastMRI & \footnotesize AXT2 & \hspace{-30pt} \footnotesize 4x \\
  \rule{0pt}{3ex}
  $P$ &
  \scalebox{1}[-1]{\includegraphics[height=0.15\textwidth,valign=m]{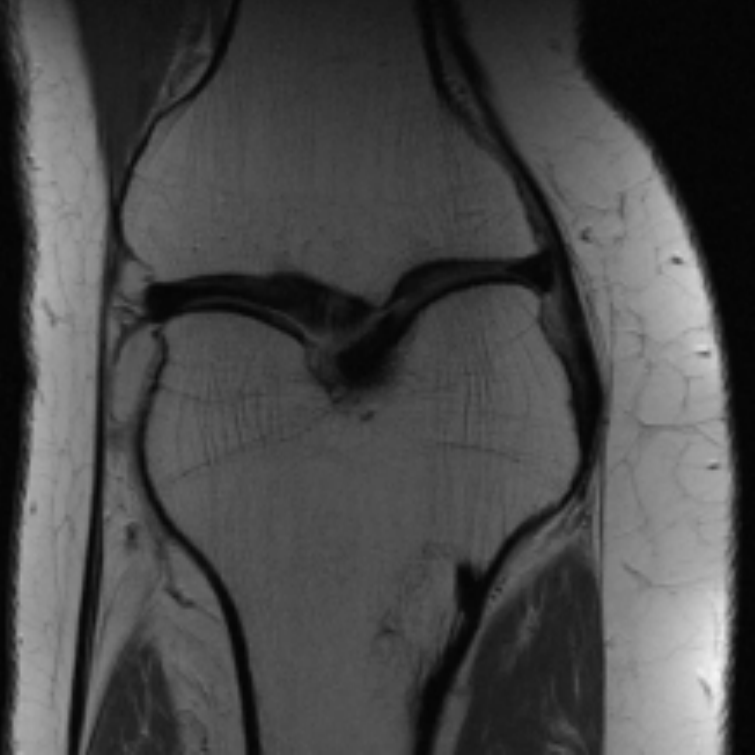}} &
  \scalebox{1}[-1]{\includegraphics[height=0.15\textwidth,valign=m]{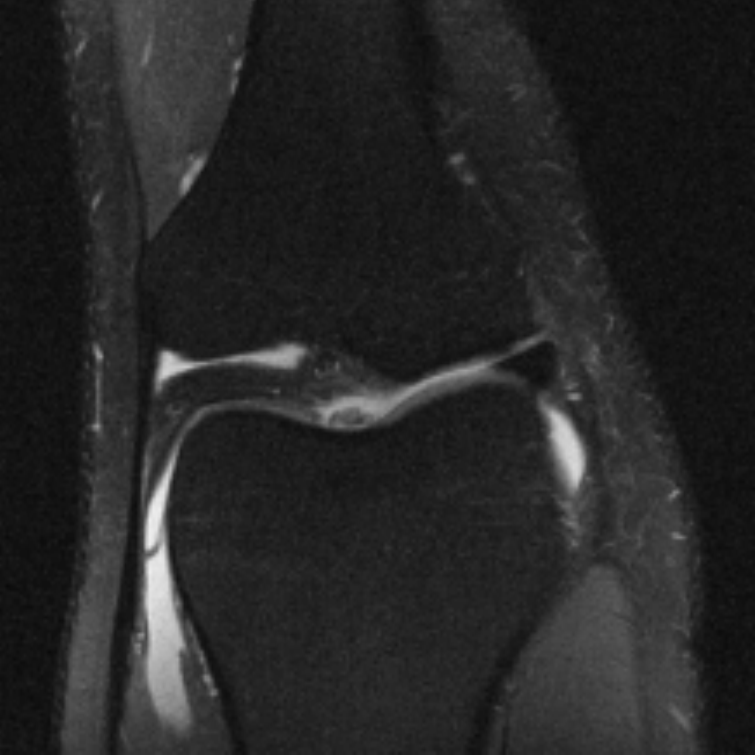}} &
  \scalebox{1}[-1]{\includegraphics[height=0.15\textwidth,valign=m]{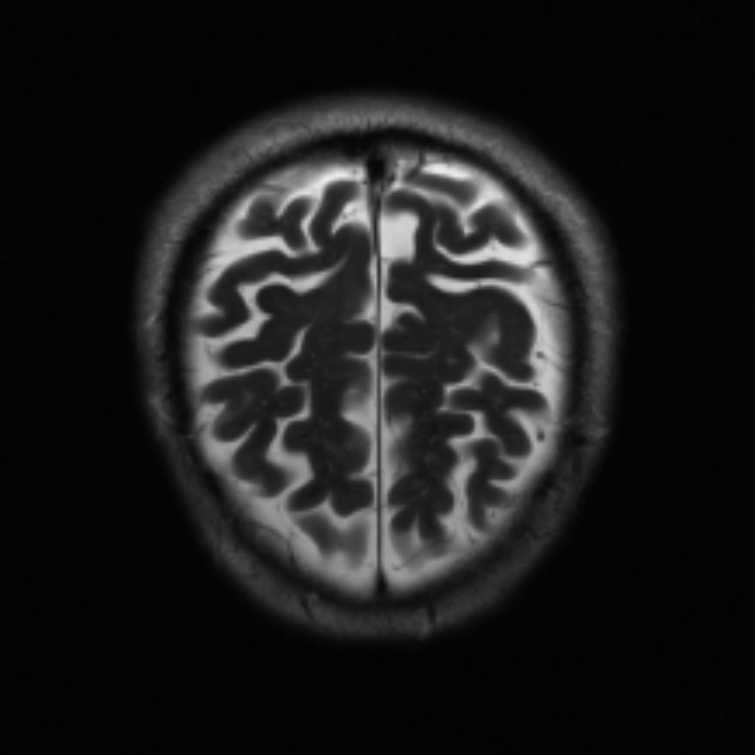}} & \hspace{-30pt}
  \scalebox{1}[-1]{\includegraphics[height=0.15\textwidth,valign=m]{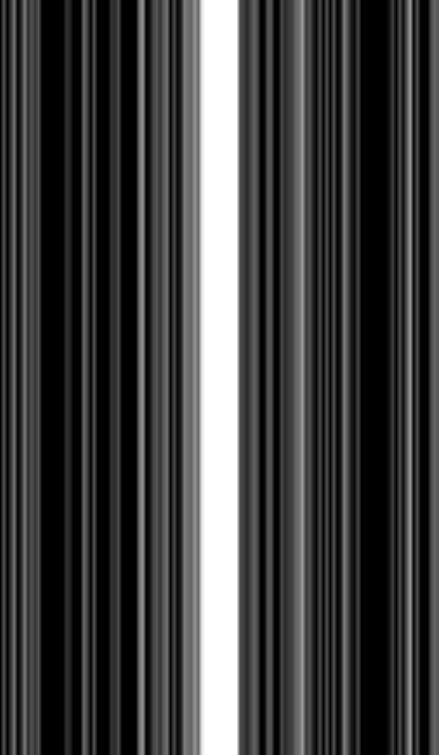}}\\
  \rule{0pt}{3ex}
  & \footnotesize brain & \footnotesize Stanford & \footnotesize AXT1PRE & \hspace{-30pt} \footnotesize 2x \\
  \rule{0pt}{3ex}
  $Q$ &
  \scalebox{1}[-1]{\includegraphics[height=0.15\textwidth,valign=m]{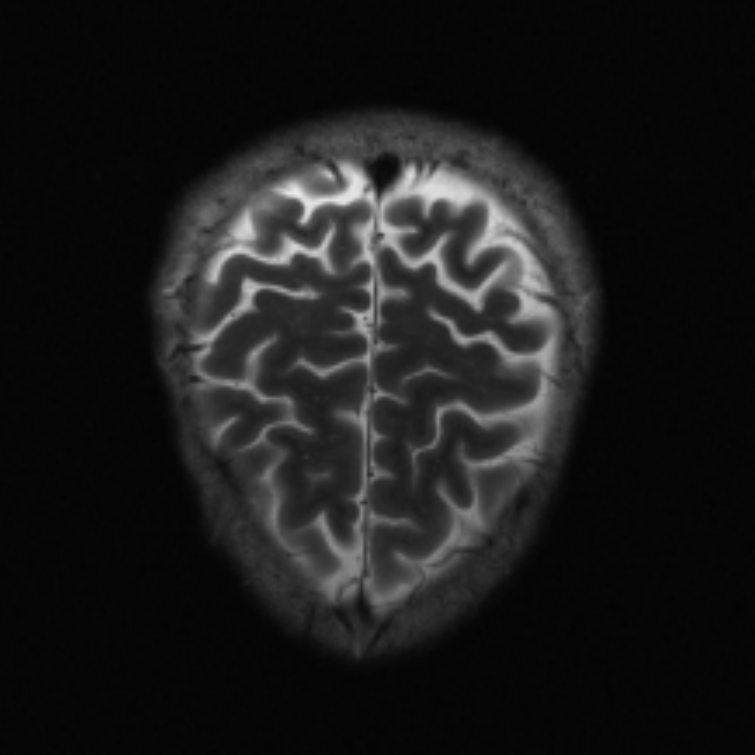}} &
  \scalebox{1}[-1]{\includegraphics[height=0.15\textwidth,valign=m]{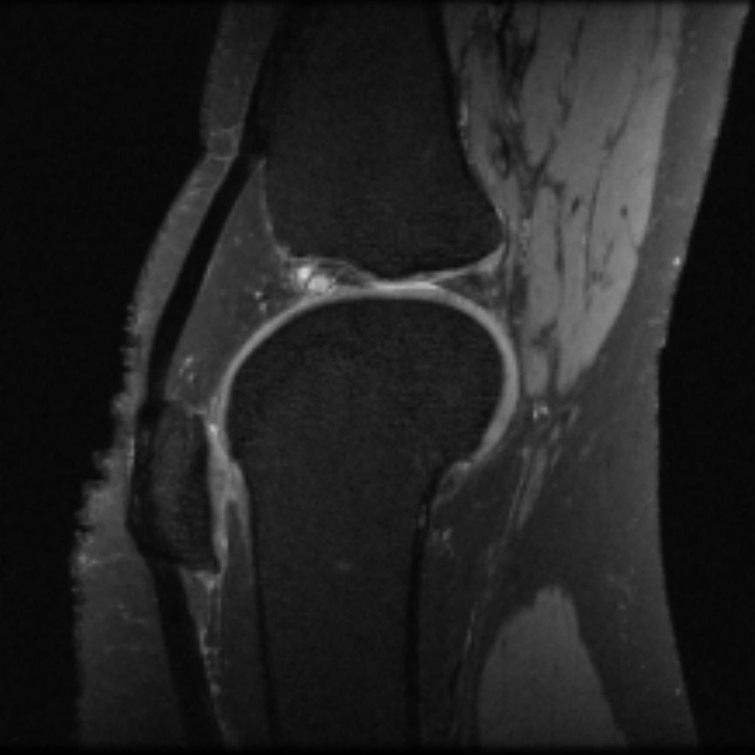}} &
  \scalebox{1}[-1]{\includegraphics[height=0.15\textwidth,valign=m]{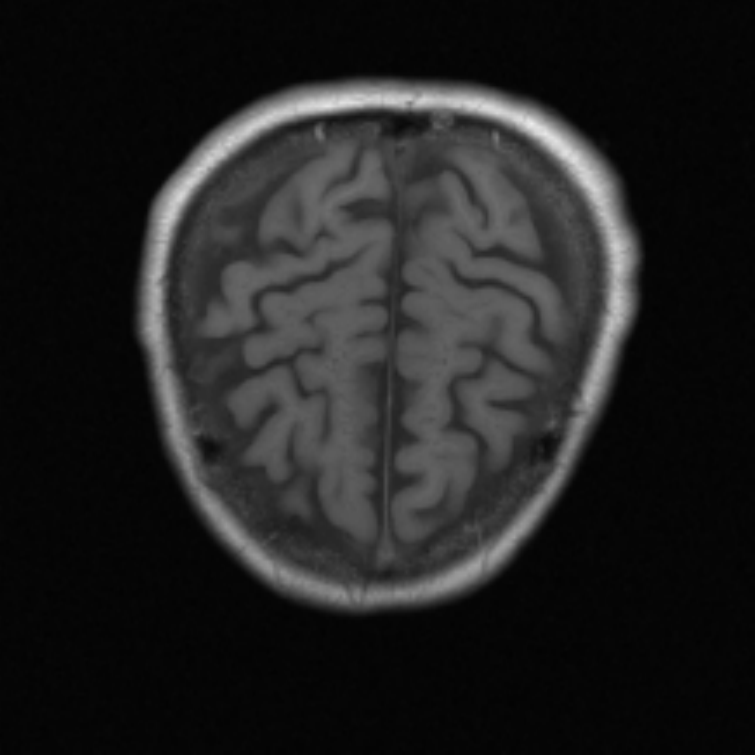}} & \hspace{-30pt}
  \scalebox{1}[-1]{\includegraphics[height=0.15\textwidth,valign=m]{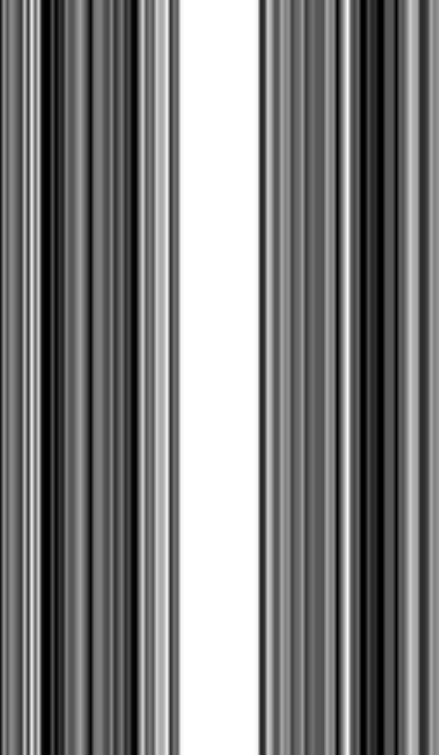}}\\
\end{tabular}
\captionof{figure}{
We study our domain adaptation method under four natural distribution shifts: anatomy, dataset, modality, and acceleration shifts. $P$ and $Q$ are the train and test domains.
}
\vspace{-10pt}
\label{fig:shifts}
\end{table}

\begin{table*}[t]
\centering
\begin{adjustbox}{width=0.95\textwidth}
\begin{tabular}{l|c|c|c|c|c}
\toprule 
\multicolumn{1}{c}{setup} & \multicolumn{1}{c}{P: knee}  & \multicolumn{1}{c}{P: fastMRI}  & \multicolumn{1}{c}{P: AXT2} & \multicolumn{1}{c}{P: 4x} & \multicolumn{1}{c}{P: fastMRI}  \\
\multicolumn{1}{c}{} & \multicolumn{1}{c}{Q: brain} & \multicolumn{1}{c}{Q: Stanford} & \multicolumn{1}{c}{Q: AXT1PRE} & \multicolumn{1}{c}{P: 2x} & \multicolumn{1}{c}{Q: adv-filt fastMRI} \\
\hline
    train on Q test on Q                & 0.9187 & 0.7164 & 0.9026 & 0.9004 & 0.6865 \\ 
    train on P test on Q                & 0.8521 & 0.6830 & 0.8506 & 0.8385 & 0.6861 \\ \hline 
    distribution shift performance gap  & 0.0666 & 0.0334 & 0.0520 & 0.0619 & 0.0004 \\ 
\thickhline
    train on Q test on Q + TTT          & 0.9234 & 0.7268 & 0.9086 & 0.9192 & 0.6827  \\
    train on P test on Q + TTT          & 0.9225 & 0.7226 & 0.9067 & 0.9176 & 0.6806  \\ \hline 
    distribution shift performance gap  & 0.0009 & 0.0042 & 0.0019 & 0.0016 & 0.0021  \\
    \thickhline
    fraction of gap closed by TTT & 98.6\% & 87.4\% & 96.3\% & 97.4\% & -- \\
\bottomrule
\end{tabular}
\end{adjustbox}
\caption{\textbf{Using self-supervision with test-time training (TTT) closes 99\%, 87\%, 96\%, and 97\% of the distribution shift performance gap for anatomy, dataset, modality, and acceleration distribution shifts.} The first three rows are SSIM scores for U-Net when trained in a supervised manner. The second three rows are SSIM scores for U-Net when self-supervision is included during training and then TTT is applied at the inference. 
The adversarially-filtered shift is an example where the distribution shift performance gap is close to zero (more than an order of magnitude smaller than for the other shifts), and thus TTT is not having an impact here (and other methods are also not expected to have an impact here). 
}
\vspace{-10pt}
\label{tab:gap}
\end{table*}

\subsection{Prior work}
A series of recent influential works in image classification has shown that image classifiers often incur a significant performance drop under natural distribution shifts~\citep{recht2019imagenet,taori2020measuring,hendrycks2020many,koh2021wilds,miller2021accuracy}.
For image reconstruction, \citet{darestani2021measuring} demonstrated that reconstruction methods for accelerated MRI (even un-trained methods such as $\ell_1$-minimization tuned on the train distribution) also suffer a significant accuracy drop when shifting from one distribution to another. 
\citep{knoll2019assessment,johnson2021evaluation} also observed a performance drop under distribution shifts in MRI. 
Consequently, several works, mainly in image classification, made efforts to overcome the distribution shift performance gap:

\paragraph{Robust optimization.} Distributionally-robust optimization learns a model by minimizing a loss with a robustness notion 
\citep{duchi2019variance,duchi2020distributionally,duchi2021learning}. 
Robust optimization methods yield a robustness gain on synthetic distribution shifts, and can yield a small gain on some natural shifts~\citet{koh2021wilds}, 
but it is unclear whether they yield significant gains on natural shifts. 

\paragraph{Data-driven interventions.} Training on larger datasets~\citep{mahajan2018exploring,yalniz2019billion} and data augmentations~\citep{devries2017improved,geirhos2018imagenet,zhang2018mixup,engstrom2019exploring,hendrycks2019augmix,yun2019cutmix,xie2020adversarial} are popular data-driven robustness interventions. 
For classification, \citet{taori2020measuring} found training with more data to marginally improve model robustness to natural distribution shifts. \citet{fabian_data_2021} found a small improvement by using data augmentations, and \citet{desai2021vortex} recently also proposed a data augmentation scheme that yields robustness gains for accelerated MRI. 

\paragraph{Domain adaptation.} 
Fine-tuning-based domain adaptation pre-trains classifiers on an auxiliary distribution $R$, then fine-tunes on a train distribution $P$, and evaluates on a test distribution $Q$ to measure the out-of-distribution (OOD) generalization performance~\citep{sharif2014cnn,donahue2014decaf,kornblith2019better}. 
Zero-shot learning methods, pre-train classifiers on $R$ and perform zero-shot inference on $Q$~\citep{radford2021learning}.
A third group of domain adaptation methods, most closely related to ours, train classifiers on $P$ 
and perform per-instance test-time training (TTT) at inference~\citep{sun2020test,liu2021ttt++,wang2021tent}.  
However, \citet{miller2021accuracy} has shown that for several natural distribution shifts, zero-shot methods offer only marginal robustness improvements and the other domain adaptation methods do not offer any improvement.

Domain adaptation methods have also been proposed for improving robustness in imaging. 
For image denoising, \citet{mohan2021adaptive} proposed GainTuning, which performs TTT at inference only on a few scaling factors of a neural network. We studied a version of this method tailored to MRI, and found no improvement in robustness for our problems.
For accelerated MRI reconstruction, \citet{liu2021universal} proposed a method 
that assumes access to examples from the target domain, which we do not have here.

Finally, \citet{yaman2021zero}  proposed a zero-shot learning method (ZS-SSL) for accelerated MRI and demonstrated that it can be used as a TTT method as well. Specifically, \citet{yaman2021zero} applied ZS-SSL to a pre-trained (on one anatomy in a fully supervised manner) model and observed that it improves model robustness under anatomy shift. Our TTT approach differs from ZS-SSL in that ZS-SSL relies on creating a synthesized dataset by repeatedly splitting the given under-sampled measurement into training and validation measurements, whereas our approach is to include a self-supervised loss in pre-training and then performing TTT with respect to that self-supervised loss. 
See the supplement for further comparison.

\section{Problem setup}
\label{sec:prob_statement}

We consider the problem of reconstructing an image from undersampled measurements. We focus on accelerated multi-coil magnetic resonance imaging (MRI), but our method also applies to other compressive sensing image reconstruction problems, for example to computed tomography. 
For such imaging problems, deep learning methods perform best. 
In our setup, the network is trained on one distribution (e.g., knees) and is tested on another distribution (e.g., brains).

\subsection{Compressive sensing}

Our goal is to reconstruct an image $\vx^* \in \mathbb{C}^N$ from undersampled measurements
\begin{align}
\label{eq:fwmap}
\vy = \mA \vx^\ast + \text{noise} \in \mathbb{C}^M,
\end{align}
where the number of measurements, $M$, is typically lower than the dimension of the image, $N$. We are given the measurement matrix $\mA$. 
We focus on accelerated MRI, in which the measurements, often called $k$-space measurements, are obtained as
\[
\vy_i = \mM \mF \mS_i \vx^* + \text{noise} \in \mathbb{C}^{M_c}, 
\quad i=1,\ldots, n_c.
\]
Here, $n_c$ denotes the number of radiofrequency coils, $\mS_i$ is a complex-valued position-dependent coil sensitivity map, that is applied through element-wise multiplication to the image $\vx^*$, $\mF$ is the 2D discrete Fourier transform, and $\mM$ is a mask (a diagonal matrix with ones and zeros on its diagonal) that implements under-sampling of $k$-space data.
The measurements $\vy_i$ and matrices can be organized so that the measurement model has the form~\eqref{eq:fwmap}. 


The MRI datasets we work with (see Section~\ref{sec:exp}) consist of
pairs of measurements and corresponding reference image $\{(\vx_j,\vy_j)\}$. 
The datasets are constructed from fully-sampled MRI data (i.e., taken with an identity mask $\mM = \mI$). 
The reference images are obtained by reconstructing the coil images from each full coil measurement as $\vx_i = \inv{\mF} \vy_i$ and then combining them via the root-sum-of-squares (RSS) algorithm to a single image:
$
    \vx = \sqrt{\sum_{i=1}^{n_c} \abs{\vx_i}^2}.
$
Here, $\abs{\cdot}$ and $\sqrt{\cdot}$ denote element-wise absolute value and squared root operations. 
The under-sampled $k$-space measurements (for acceleration) 
are obtained by applying a standard 1D random mask (random vertical lines in the frequency domain), which is the default in the fastMRI challenge. We consider 4x acceleration throughout the paper, the acceleration factor considered in the fastMRI challenge~\citep{knoll2020advancing,muckley2020state}. 

\subsection{Image reconstruction with neural networks}\label{sec:methods}
We study our domain adaptation method for two neural networks, a standard baseline method (U-net) and the state-of-the-art reconstruction method (VarNet). 

U-Net~\citep{ronneberger2015u} is a convolutional network which for MRI is trained end-to-end to map a least-squares reconstruction obtained from a measurement $\vy$ to a clean image $\vx$ by minimizing the loss
$
\mathcal{L}(\vy,\vx,\vtheta) = \norm[1]{ \vx- f_\vtheta(\pinv{\mA}\vy)}
$
over a training dataset $\{(\vx_1,\vy_1),\ldots,(\vx_n,\vy_n)\}$ \cite{jin2017deep}. 
Here, $f_\vtheta$ is a U-Net with parameters $\vtheta$.  

VarNet is a variational network that gives state-of-the-art performance~\citep{sriram2020end}. 
Similar to U-Net, VarNet is trained to map an under-sampled measurement to a clean image, but contrary to U-net, it contains data consistency blocks. 


\subsection{Problem statement: Overcoming the distribution shift performance gap}

Neural networks for MRI are often trained on the fastMRI training set~\citep{zbontar2018fastmri} and evaluated on the fastMRI test set. This measures in-distribution performance since the sets are constructed by collecting data and splitting it in train and test sets. 

This evaluation mode is not reflective of performance in practice, where we typically train on data acquired in one setup (anatomy, scanner type, acquisition mode, etc.) and apply the network in a different setup 
on data from another anatomy, scanner, or acquisition mode. All those changes introduce distribution shifts. 

Under anatomy shifts (training on knee and testing on brain) and dataset shifts (training on NYU data~\citep{zbontar2018fastmri} and testing on Stanford data~\citep{epperson2013creation}), different methods lose a similar and significant amount in image quality~\citet{darestani2021measuring}. 
This loss comprises two parts: one is due to variations in difficulty of the datasets; e.g., images with finer details are harder to reconstruct and result in lower scores. This part cannot be overcome by better algorithms or even by having access to the test distribution. The second part is due to a missfit of algorithm and test distribution, this can in principle be overcome if we had access to the test distribution. We call this second gap the \textbf{distribution shift performance gap}. 

The goal of this work is to close the distribution shift performance gap without having access to the test distribution. We assume that we can train a network on a training distribution $P$, but at inference, we are only given a measurement $\vy$ from a test distribution $Q$, without any other information about the test distribution.


\section{The distribution shift performance gap for four natural distribution shifts}

In this section, we introduce the four natural distribution shifts we consider and measure the corresponding distribution shift performance gap. 
We also include a fifth distribution shift for which the distribution shift performance gap is close to zero, 
even though we observe a performance drop. 

We then show that the models we consider have the ability to close the performance gap, by training a network on data from both distributions. This show that the networks we consider can perform well on both distributions simultaneously. Of course, in practice we cannot train on the test distribution, since we do not have access to examples from the test distribution.


\subsection{Natural distribution shifts considered}
We consider four natural distribution shifts illustrated in Figure~\ref{fig:shifts}, and one artificial distribution shift.

\paragraph{Anatomy shift.} 
We consider an anatomy shift from training on fastMRI knee images to testing on fastMRI brain images. 

\paragraph{Dataset shift.} 
We consider a dataset shift from training on fastMRI knee images (collected at NYU) to testing on Stanford knee images~\citep{epperson2013creation}. The main differences are: (1) The Stanford data is constructed via volumetric 3D recording (fastMRI scans are 2D), (2) The Stanford set contains samples with a lower frequency resolution than fastMRI, and (3) the slice thickness is 5 times smaller in the Stanford set.

\paragraph{Modality shift.} 
A modality shift occurs when the acquisition mode of training images is different than the one for test images. A modality shift is subtle, in that it occurs within an anatomy (say brain) and only the contrast of the images changes. 
We consider a modality shift from AXT2 to AXT1PRE images (see Figure~\ref{fig:shifts} for an example, and the supplement for the change in pixel intensity distribution).

\paragraph{Acceleration shift.} 
We consider an acceleration shift from training on 4x accelerated measurements to testing on 2x accelerated measurements. 


For each of these natural distribution shifts, we compute the distribution shift performance gap for a given model as the gap in SSIM (an image comparison metric) when the training domain changes. Table~\ref{tab:gap} shows the distribution shift performance gap for U-Net given each of the distribution shifts explained above. 

We note that there are distribution shifts for which the distribution shift performance gap is essentially zero, and thus no robustness intervention can be helpful. 
An example is an {\bf adversarially-filtered shift} which occurs when the target distribution contains hard-to-reconstruct samples from the original domain. Specifically, models trained on the fastMRI dataset achieve a significantly lower score on fastMRI-A, which contains hard-to-reconstruct samples~\citep{darestani2021measuring}. 
However, models achieve a low score on fastMRI-A simply because those samples contain finer details.  
This can be seen in Table~\ref{tab:gap} where the best performance obtainable by models is 0.6865 and when models are trained on fastMRI instead of fastMRI-A, they still achieve that performance on fastMRI-A. This example shows that a performance degradation on the target domain does not necessarily translate into a distribution shift performance gap as defined above.


\subsection{Networks can close the distribution shift performance gap}


Before describing our domain adaptation method, we investigate whether models are capable of achieving close-to-zero performance gap. To this end, we assume access to data from distribution $Q$ and train a U-Net on the mixture distribution of $P$ and $Q$, for each distribution shift introduced in the previous section, and evaluate the models on the two distributions. 
The results are shown in Figure~\ref{fig:trained_on_mixture}, which depicts performance as a function of the mixture coefficient, which indicates the proportion of data from distribution $Q$. 

The distribution shift performance gap is captured by the difference in the vertical direction between points with mixture coefficient $1.0$ and $0.0$. For the natural distribution shifts we study,
there is a significant distribution shift performance gap, while for adversarially-filtered shift, the gap is relatively small.

We observe an approximately vertical section intersected with an approximately horizontal section at a relatively sharp angle. This shows that, when more data from the target distribution $Q$ is added to the training set, performance on $Q$ increases while performance on $P$ does not degrade. 
As a consequence, the performance gap is roughly closed when the model is trained on the mixture distribution with the mixture coefficient at the intersection point. 

\begin{table}[t]
\setlength{\tabcolsep}{-2pt}
\centering
\begin{adjustbox}{width=1\textwidth}
\begin{tabular}{ccccc}
  \footnotesize Anatomy shift & \footnotesize Dataset shift & \footnotesize Modality shift & \footnotesize Acceleration shift & \footnotesize  Adversarially-filtered shift \\
  \begin{tikzpicture}
\begin{axis}[
    xmin = 0.72, xmax = 0.885,
    ymin = 0.84, ymax = 0.96,
    xtick distance = 0.05,
    ytick distance = 0.025,
    yticklabel style={xshift=1.5pt},
    every tick label/.append style={font=\tiny},
    grid = both,
    minor tick num = 1,
    major grid style = {lightgray!25},
    minor grid style = {lightgray!25},
    width = 0.27\textwidth, 
    height = 0.2\textwidth, 
    xlabel = {SSIM on $P$}, 
    ylabel = {SSIM on $Q$}, 
    label style={font=\scriptsize},
    legend cell align = {left},
    legend pos = south west,
]
\addplot+[
    mark size=1.5pt, only marks, steelblue,mark options={solid,steelblue,fill=steelblue}, 
    point meta = explicit symbolic,
    nodes near coords,
    visualization depends on = {\thisrow{theta} \as \theta},
    visualization depends on = {\thisrow{r} \as \r},
    every node near coord/.append style = {
        xshift = {\r * cos(\theta)},
        yshift = {\r * sin(\theta)},
        font = \scriptsize
    },
    error bars/.cd, x dir=both, x explicit, y dir=both, y explicit,] 
    table [x = {data_p_mean}, 
        y = {data_q_mean}, 
        x error expr = {\thisrow{data_p_std}*2}, 
        y error expr = {\thisrow{data_q_std}*2},
        meta = {meta}]{
        data_p_mean data_q_mean data_p_std data_q_std meta theta r
0.84303 0.85891 0.00233 0.00660 0.0 -40 10
0.84374 0.90119 0.00186 0.00268 0.05 -50 13
0.84350 0.90614 0.00287 0.00322 0.1 -40 10
0.84016 0.91096 0.00368 0.00284 0.25 -10 10
0.83255 0.91251 0.00599 0.00429 0.5 -100 12
0.83036 0.91863 0.00101 0.00050 0.75 5 7
0.81650 0.91860 0.00211 0.00089 0.9 60 1
0.80233 0.91822 0.00443 0.00097 0.95 180 5
0.73329 0.91971 0.00293 0.00051 1.0 90 1      
    }; 
\addplot +[mark=none,steelblue] table[x=x,y=y]{./files/curve_anatomy_shift.txt};
\end{axis}
\end{tikzpicture} &
\begin{tikzpicture}
\begin{axis}[
    xmin = 0.70, xmax = 0.735,
    ymin = 0.68, ymax = 0.745,
    xtick distance = 0.01,
    ytick distance = 0.01,
    yticklabel style={xshift=1.5pt},
    every tick label/.append style={font=\tiny},
    grid = both,
    minor tick num = 1,
    major grid style = {lightgray!25},
    minor grid style = {lightgray!25},
    width = 0.27\textwidth, 
    height = 0.2\textwidth, 
    xlabel = {SSIM on $P$}, 
    ylabel = {}, 
    label style={font=\scriptsize},
    legend cell align = {left},
    legend pos = south west,
]
\addplot+[
    mark size=1.5pt, only marks, steelblue, mark options={solid,steelblue,fill=steelblue}, point meta = explicit symbolic,
    nodes near coords,
    visualization depends on = {\thisrow{theta} \as \theta},
    visualization depends on = {\thisrow{r} \as \r},
    every node near coord/.append style = {
        xshift = {\r * cos(\theta)},
        yshift = {\r * sin(\theta)},
        font = \scriptsize
    },error bars/.cd, x dir=both, x explicit, y dir=both, y explicit,] 
    table [x = {data_p_mean}, 
        y = {data_q_mean}, 
        x error expr = {\thisrow{data_p_std}*2}, 
        y error expr = {\thisrow{data_q_std}*2},
        meta = {meta}]{
        data_p_mean data_q_mean data_p_std data_q_std meta theta r
        0.72878 0.69247 0.00194 0.00396 0.0 -10 9
        0.72651 0.72122 0.00084 0.00212 0.05 -50 19
        0.72766 0.72399 0.00212 0.00292 0.1 -45 15
        0.72657 0.72592 0.00103 0.00305 0.25 0 10
        0.72266 0.72402 0.00335 0.00221 0.5 30 7
        0.72245 0.72560 0.00134 0.00089 0.75 180 6
        0.71331 0.72244 0.00260 0.00134 0.9 20 3
        0.71093 0.72093 0.00314 0.00111 0.95 180 4
        0.70518 0.72307 0.00179 0.00137 1.0 190 5
    };
\addplot +[mark=none,steelblue] table[x=x,y=y]{./files/curve_dataset_shift.txt};
\end{axis}
\end{tikzpicture} &
\begin{tikzpicture}
\begin{axis}[
    xmin = 0.86, xmax = 0.94,
    ymin = 0.83, ymax = 0.93,
    xtick distance = 0.02,
    ytick distance = 0.02,
    yticklabel style={xshift=1.5pt},
    every tick label/.append style={font=\tiny},
    grid = both,
    minor tick num = 1,
    major grid style = {lightgray!25},
    minor grid style = {lightgray!25},
    width = 0.27\textwidth, 
    height = 0.2\textwidth, 
    xlabel = {SSIM on $P$}, 
    label style={font=\scriptsize},
    legend cell align = {left},
    legend pos = south west,
    point meta = explicit symbolic,
    nodes near coords,
    visualization depends on = {\thisrow{theta} \as \theta},
    visualization depends on = {\thisrow{r} \as \r},
    every node near coord/.append style = {
        xshift = {\r * cos(\theta)},
        yshift = {\r * sin(\theta)},
        font = \scriptsize
    },
]
\addplot+[
    mark size=1.5pt, steelblue, mark options={solid,steelblue,fill=steelblue}, error bars/.cd, x dir=both, x explicit, y dir=both, y explicit,] 
    table [x = {data_p_mean}, 
        y = {data_q_mean}, 
        x error expr = {\thisrow{data_p_std}*2}, 
        y error expr = {\thisrow{data_q_std}*2},
        meta = {meta}]{
        data_p_mean data_q_mean data_p_std data_q_std meta theta r
0.91881 0.84967 0.00080 0.00231 0.0 -50 12
0.91780 0.87514 0.00090 0.00260 0.05 -45 14
0.91813 0.88528 0.00126 0.00265 0.1 -45 12
0.91631 0.89400 0.00083 0.00113 0.25 -30 14
0.91199 0.89701 0.00042 0.00182 0.5 -15 12
0.90543 0.90089 0.00164 0.00112 0.75 0 9
0.89638 0.90141 0.00046 0.00044 0.9 10 6
0.89109 0.90284 0.00044 0.00092 0.95 180 2
0.87380 0.90211 0.00312 0.00129 1.0 90 1
    };
\end{axis}
\end{tikzpicture} 
&
\begin{tikzpicture}
\begin{axis}[
    xmin = 0.75, xmax = 0.87,
    ymin = 0.82, ymax = 0.93,
    xtick distance = 0.05,
    ytick distance = 0.02,
    yticklabel style={xshift=1.5pt},
    every tick label/.append style={font=\tiny},
    grid = both,
    minor tick num = 1,
    major grid style = {lightgray!25},
    minor grid style = {lightgray!25},
    width = 0.27\textwidth, 
    height = 0.2\textwidth, 
    xlabel = {SSIM on $P$}, 
    label style={font=\scriptsize},
    legend cell align = {left},
    legend pos = south west,
]
\addplot+[
    mark size=1.5pt, only marks, steelblue, mark options={solid,steelblue,fill=steelblue}, point meta = explicit symbolic,
    nodes near coords,
    visualization depends on = {\thisrow{theta} \as \theta},
    visualization depends on = {\thisrow{r} \as \r},
    every node near coord/.append style = {
        xshift = {\r * cos(\theta)},
        yshift = {\r * sin(\theta)},
        font = \scriptsize
    },error bars/.cd, x dir=both, x explicit, y dir=both, y explicit,] 
    table [x = {data_p_mean}, 
        y = {data_q_mean}, 
        x error expr = {\thisrow{data_p_std}*2}, 
        y error expr = {\thisrow{data_q_std}*2},
        meta = {meta}]{
        data_p_mean data_q_mean data_p_std data_q_std meta theta r
        0.76956 0.90240 0.00036 0.00016  1.0 150 2
        0.80356 0.90280 0.00262 0.00065 0.95 180 10 
        0.81546 0.90367 0.00260 0.00086 0.9 180 4 
        0.82433 0.90296 0.00054 0.00016 0.75 0 5
        0.83086 0.90300 0.00202 0.00106 0.5 -20 13
        0.82980 0.89713 0.00157 0.00091 0.25 -30 15
        0.83390 0.89413 0.00045 0.00134 0.1 -50 15 
        0.83257 0.88240 0.00066 0.00139 0.05 -50 17
        0.83366 0.83740 0.00057 0.00396 0.0 -20 10 
    };
\addplot +[mark=none,steelblue] table[x=x,y=y]{./files/curveacc.txt};
\end{axis}
\end{tikzpicture} & 
\begin{tikzpicture}
\begin{axis}[
    xmin = 0.77,
    xmax = 0.815,
    ymin = 0.63,
    ymax = 0.725,
    yticklabel style={xshift=1.5pt},
    every tick label/.append style={font=\tiny},
    yticklabel style={/pgf/number format/fixed,/pgf/number format/precision=3},
    xticklabel style={/pgf/number format/fixed,/pgf/number format/precision=3},
    grid = both,
    minor tick num = 1,
    major grid style = {lightgray!25},
    minor grid style = {lightgray!25},
    width = 0.27\textwidth, 
    height = 0.2\textwidth, 
    xlabel = {SSIM on $P$}, 
    ylabel = {}, 
    label style={font=\scriptsize},
    legend cell align = {left},
    legend pos = south west,
    point meta = explicit symbolic,
    nodes near coords,
    visualization depends on = {\thisrow{theta} \as \theta},
    visualization depends on = {\thisrow{r} \as \r},
    every node near coord/.append style = {
        xshift = {\r * cos(\theta)},
        yshift = {\r * sin(\theta)},
        font = \scriptsize
    },
]
\addplot+[
    mark size=1.5pt, only marks, steelblue, mark options={solid,steelblue,fill=steelblue}, error bars/.cd, x dir=both, x explicit, y dir=both, y explicit,
    ] 
    table [
    x = {data_p_mean}, 
        y = {data_q_mean}, 
        x error expr = {\thisrow{data_p_std}*2}, 
        y error expr = {\thisrow{data_q_std}*2},
        meta = {meta}]{
        data_p_mean data_q_mean data_p_std data_q_std meta theta r
        0.7956 0.6859 0.00194 0.00094  1.0 205 15
        0.7975 0.6866 0.00108 0.00134 0.95 180 9 
        0.7979 0.6869 0.00135 0.00107 0.9 130 8 
        0.7996 0.6876 0.00200 0.00184 0.75 50 8
        0.7973 0.6853 0.00080 0.00070 0.5 -120 13
        0.7981 0.6854 0.00149 0.00105 0.25 -90 19
        0.7999 0.6863 0.00196 0.00166 0.1 -10 12 
        0.7997 0.6860 0.00045 0.00077 0.05 -30 18
        0.7993 0.6854 0.00144 0.00118 0.0 -70 14  
    };
\end{axis}
\end{tikzpicture}
\end{tabular}
\end{adjustbox}

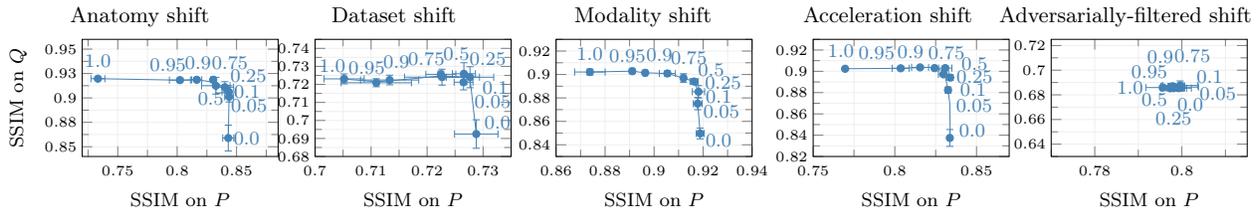
\captionof{figure}{Performance of U-Net trained on the mixture distributions of source and target domains: 
Numbers are the mixture coefficients; for example 0.25 means that $25\%$ of the training data comes from distribution $Q$.
The performance gap is significant for the four natural distribution shifts, but essentially non-existent for the adversarially-filtered shift.
The model can perform well on both distributions $P$ and $Q$ simultaneously irrespective of the distribution shifts. Error bars denote $95\%$ confidence intervals. 
}
\label{fig:trained_on_mixture}
\end{table}

\section{Method: Incorporating self-supervised training and then performing test-time training at inference}\label{sec:methodology}

In this section, we describe our domain adaptation method for compressive sensing which incorporates a self-supervised loss into the training of a deep learning model, and performs test-time training (TTT) during inference.
Let $f_\vtheta$ be a neural network mapping a coarse reconstruction from a measurement $\vy$ to a clean image (e.g., a U-Net or VarNet). The training and inference stages are as follows.

\paragraph{Training stage:}
Given a training set consisting of (ground-truth-image, measurement) pairs $\{(\vx_1,\vy_1),\ldots, (\vx_n,\vy_n)\}$, we learn a model by minimizing the loss function
\begin{align}
    \mc L(\vtheta) &=
    \frac{1}{n}
    \sum_{i=1}^n
    \Bigg(
    \underbrace{\frac{\norm[1]{\vx_i-f_\vtheta(\mA^\dagger \vy_i)}}{\norm[1]{\vx_i}}}_{\mathcal{L}_{\text{sup}}} 
    +
    \underbrace{\frac{\norm[1]{\vy_i - 
    \mA 
    f_\vtheta(\mA^\dagger \vy_i)}}{\norm[1]{\vy_i}}}_{\mathcal{L}_{\text{self}}}
    \Bigg).
    \label{eq:joint-loss}
\end{align}
Here, $\mathcal{L}_{\text{sup}}$ is the supervised part of the loss function which ensures that the output of the network is close to the ground truth image. 
For our setup, the $\ell_1$-norm works well as a loss but other choices, such as the SSIM loss, also work. 

The un-supervised loss we propose is simply enforcing data-consistency, and might seem like an odd choice, given that existing TTT schemes from the classification literature typically choose an auxiliary task, and not a fitting term as a self-supervised loss. We also experimented with auxiliary reconstruction tasks, such as a self-supervised denoising task (see appendix), without success. We further discuss the motivation for the un-supervised loss below. 

\paragraph{Inference stage:}
At inference, we are given a (typically out-of-distribution) measurement $\vy$, and we estimate an image as follows. 
We optimize the network parameters $\vtheta$ with respect to the loss $\mathcal{L}_{\text{self}}$ for the given under-sampled test measurement. We refer to this step as TTT. To prevent TTT from overfitting to the measurement, we early-stop the optimization based on a self-validation loss computed over a fraction of the under-sampled measurement. 
Specifically, let $\vy$ be the under-sampled test measurement with $M$ frequency-domain measurements. 
We split $\vy$ into a $\vy_{\text{train}}$ and $\vy_{\text{val}}$, which contains a random fraction $q = M/k$ of all measurements in $\vy$. 
We then perform TTT on $\vy_{\text{train}}$ and monitor the error between $\vy_{\text{val}}$ and the predicted frequencies by the network for self-validation to stop TTT early when the self-validation error starts to rise. 

\paragraph{Motivation for the self-supervised loss:}
Our self-supervised loss is low if the network generates an image that is consistent with the measurement. 
This loss might seem an unusual choice, given that TTT methods from the literature usually use a loss based on an auxiliary task. 
However, we observe footprints of this loss in works on cycle consistency losses in image translation 
~\citep{zhu2017unpaired} and image classification~\citep{hoffman2018cycada}. 
For image-to-image translation, \citet{zhu2017unpaired} proposed cycle consistency to learn a consistent mapping from one image domain $X$ to another image domain $Y$. 
\citet{oh2020unpaired} used a cycle consistency loss to train a GAN for image reconstruction. 

Our self-supervised loss uses the network as an image model to enforce consistency between the output and the under-sampled measurement. 
Therefore, we expect such self-supervision to work well for architectures that are \emph{good} image models. Both U-Net and VarNet are Convolutional Neural Networks (CNNs), and CNNs are good image priors even without any training~\citet{ulyanov2018deep}.
Un-trained CNNs are such good image priors so that they can perform image reconstruction without any training~\citep{ulyanov2018deep,heckel2019deep,arora2020untrained,heckel2020compressive,wang2020phase,bostan2020deep,darestani2020can}. 


\section{Experiments}\label{sec:exp}

In this section, we show that self-supervised training with test-time training (TTT) closes the distribution shift performance gap for the four natural distribution shifts considered.  

For each experiment, we have an original distribution $P$ and a target distribution $Q$. Performance is measured in terms of the structural similarity index measure (SSIM). 
The distribution shift performance gap when training is fully supervised is defined as 
\begin{align*}
\text{gap}_\text{before}=\text{SSIM(train on $Q$ test on $Q$)} - \text{SSIM(train on $P$ test on $Q$)}, 
\end{align*}
and is defined as 
\begin{align*}
\text{gap}_\text{after}= 
\text{SSIM(train on $Q$ test on $Q$ + TTT)} - \text{SSIM(train on $P$ test on $Q$ + TTT)}
\end{align*}
when we include self-supervision during training and then apply TTT.

If the robustness intervention is successful, then $\text{gap}_\text{after}$ is significantly smaller than $\text{gap}_\text{before}$. 

Throughout the experiments, we work with a U-Net with 8 layers and 64 channels as the width factor and a VarNet with 12 cascades and 18 channels as the width factor. For training, we run the Adam optimizer~\citep{kingma2014adam} with learning rate set to 1e-5 for U-Net (and 1e-4 for VarNet). 

The train and test datasets we used for each experiment are described below for each distribution shift. 
We observed that a training set of 300-400 slices gives a good performance on the validation set. For larger training sets, we observed no significant performance improvement (less than 0.5\% SSIM score).

\begin{table*}[th]
\centering
\begin{adjustbox}{width=1\textwidth}
\begin{tabular}{l|l|c|c|c|c|c|c|c|c}
\toprule 
\multicolumn{1}{c|}{} & \multicolumn{1}{c|}{} & \multicolumn{2}{c|}{\bf anatomy shift} & \multicolumn{2}{c|}{\bf dataset shift} & \multicolumn{2}{c|}{\bf modality shift} & \multicolumn{2}{c}{\bf acceleration shift} \\
\multicolumn{1}{c|}{} & \multicolumn{1}{c|}{} & \multicolumn{2}{c|}{\bf P:knee \; Q:brain} & \multicolumn{2}{c|}{\bf P:fastMRI \; Q:Stanford} & \multicolumn{2}{c|}{\bf P:AXT2 \; Q:AXT1PRE} & \multicolumn{2}{c}{\bf P:4x \; Q:2x} \\\cline{3-10}
\multicolumn{1}{c|}{training scheme} & \multicolumn{1}{c|}{setup} & \multicolumn{1}{c|}{U-Net} & \multicolumn{1}{c|}{VarNet} & \multicolumn{1}{c|}{U-Net} & \multicolumn{1}{c|}{VarNet}  & \multicolumn{1}{c|}{U-Net} & \multicolumn{1}{c|}{VarNet} & \multicolumn{1}{c|}{U-Net} & \multicolumn{1}{c}{VarNet}\\
\multicolumn{1}{c|}{} & \multicolumn{1}{c|}{} & \multicolumn{1}{c|}{SSIM} & \multicolumn{1}{c|}{SSIM} & \multicolumn{1}{c|}{SSIM} & \multicolumn{1}{c|}{SSIM} & \multicolumn{1}{c|}{SSIM} & \multicolumn{1}{c|}{SSIM} & \multicolumn{1}{c|}{SSIM} & \multicolumn{1}{c}{SSIM}  \\
\hline
    \multirow{2}{*}{\begin{tabular}{@{}c@{}} self-supervision \\ included \end{tabular}} & train on P test on Q + TTT  & 0.9225 & 0.9315 & 0.7226 & 0.7247 & 0.9067 & 0.9084 & 0.9176 & 0.9203 \\ 
    & train on Q test on Q + TTT    & \textbf{0.9234} & 0.9322 &  \textbf{0.7268} & 0.7295 & \textbf{0.9086} & \textbf{0.9110} & \textbf{0.9192} & 0.9207 \\ 
    \thickhline
    \multirow{3}{*}{supervised} & train on Q test on Q      & 0.9187  & \textbf{0.9344}  & 0.7164  &  \textbf{0.7442} & 0.9026 & 0.9105  & 0.9004  & \textbf{0.9224}\\
    & train on P test on Q                                  & 0.8521 & 0.8717 & 0.6830 & 0.7062 & 0.8506 & 0.8744 & 0.8385 & 0.8744 \\
    & train on P test on Q + TTT                             & 0.8648 & 0.8533 & 0.6130 & 0.6026 & 0.8695 & 0.8102 & 0.8090 & 0.8324 \\
\thickhline
    \multicolumn{2}{l|}{fraction of gap closed by TTT} & \multicolumn{1}{c|}{98.6\%} & \multicolumn{1}{c|}{98.8\%} & \multicolumn{1}{c|}{87.4\%} & \multicolumn{1}{c|}{87.3\%} & \multicolumn{1}{c|}{96.3\%} & \multicolumn{1}{c|}{92.7\%} & \multicolumn{1}{c|}{97.4\%} & \multicolumn{1}{c}{99.1\%} \\
\bottomrule
\end{tabular}
\end{adjustbox}
\caption{\textbf{Including self-supervision while training deep learning models combined with TTT improves model robustness to natural anatomy, dataset, modality, and acceleration shifts.} SSIM scores are averaged over 100 slices from the test set of $Q$. Note that TTT is only effective when self-supervision is included during training, and as shown in the second to the last row, it offers marginal improvement when applied to a model that is trained in a supervised fashion.}
\label{tab:results}
\end{table*}

\begin{table*}[th]
\centering
\begin{adjustbox}{width=1\textwidth}
\begin{tabular}{l|l|c|c|c|c|c|c|c|c}
\toprule 
\multicolumn{1}{c|}{} & \multicolumn{1}{c|}{} & \multicolumn{2}{c|}{\bf anatomy shift} & \multicolumn{2}{c|}{\bf dataset shift} & \multicolumn{2}{c|}{\bf modality shift} & \multicolumn{2}{c}{\bf acceleration shift} \\
\multicolumn{1}{c|}{} & \multicolumn{1}{c|}{} & \multicolumn{2}{c|}{\bf P:knee \; Q:brain} & \multicolumn{2}{c|}{\bf P:fastMRI \; Q:Stanford} & \multicolumn{2}{c|}{\bf P:AXT2 \; Q:AXT1PRE} & \multicolumn{2}{c}{\bf P:4x \; Q:2x} \\\cline{3-10}
\multicolumn{1}{c|}{} & \multicolumn{1}{c|}{} & \multicolumn{1}{c|}{U-Net} & \multicolumn{1}{c|}{VarNet} & \multicolumn{1}{c|}{U-Net} & \multicolumn{1}{c|}{VarNet}  & \multicolumn{1}{c|}{U-Net} & \multicolumn{1}{c|}{VarNet} & \multicolumn{1}{c|}{U-Net} & \multicolumn{1}{c}{VarNet}\\
\multicolumn{1}{c|}{training scheme} & \multicolumn{1}{c|}{setup} & \multicolumn{1}{c|}{runtime} &  \multicolumn{1}{c|}{runtime} &  \multicolumn{1}{c|}{runtime} &  \multicolumn{1}{c|}{runtime} &  \multicolumn{1}{c|}{runtime} &  \multicolumn{1}{c|}{runtime} & \multicolumn{1}{c|}{runtime} & \multicolumn{1}{c}{runtime} \\
\multicolumn{1}{c|}{} & \multicolumn{1}{c|}{} & \multicolumn{1}{c|}{(mins/slice)} & \multicolumn{1}{c|}{(mins/slice)} & \multicolumn{1}{c|}{(mins/slice)} & \multicolumn{1}{c|}{(mins/slice)} & \multicolumn{1}{c|}{(mins/slice)} & \multicolumn{1}{c|}{(mins/slice)} & \multicolumn{1}{c|}{(mins/slice)} & \multicolumn{1}{c}{(mins/slice)} \\
\hline
    \multirow{2}{*}{\begin{tabular}{@{}c@{}} self-supervision \\ included \end{tabular}} & train on P test on Q + TTT & 10.1 & 12.9 &  1.2 & 1.5 & 1.6 & 6.5 &  3.6 & 5.2\\ 
    & train on Q test on Q + TTT    & 3.3 &  4.2 & 0.4 & 0.7 & 0.7 & 2.6 & 2.6 & 5.1 \\ 
\bottomrule
\end{tabular}
\end{adjustbox}
\caption{\textbf{Reducing the robustness gap comes at a noticeable computational cost for TTT.} 
Runtimes are averaged over 100 slices from the test set of $Q$ and are reported for a single GPU.}
\label{tab:runtimes}
\end{table*}

\begin{table}[t!]
\setlength{\tabcolsep}{1pt}
\centering
\begin{adjustbox}{width=0.85\textwidth}
\begin{tabular}{ccccc}
  & \begin{tabular}{@{}c@{}} \large train on P\\  \large test on Q \\[5pt]  SSIM: 0.8562 \end{tabular}  & \begin{tabular}{@{}c@{}} \large after our domain \\ \large adaptation method \\[5pt] SSIM: 0.9215 \end{tabular} & \begin{tabular}{@{}c@{}} \large train on Q \\ \large test on Q \\[5pt] SSIM: 0.9185 \end{tabular} & \begin{tabular}{@{}c@{}} \large ground truth \\[25pt] \end{tabular} \\
  \rule{0pt}{3ex}
  \begin{tabular}{@{}c@{}} \large anatomy \\ \large shift \end{tabular} &
  \begin{adjustbox}{valign=m}
    \begin{tikzpicture}[boximg]
      \node[anchor=south] (img) {\includegraphics[width=0.25\textwidth]{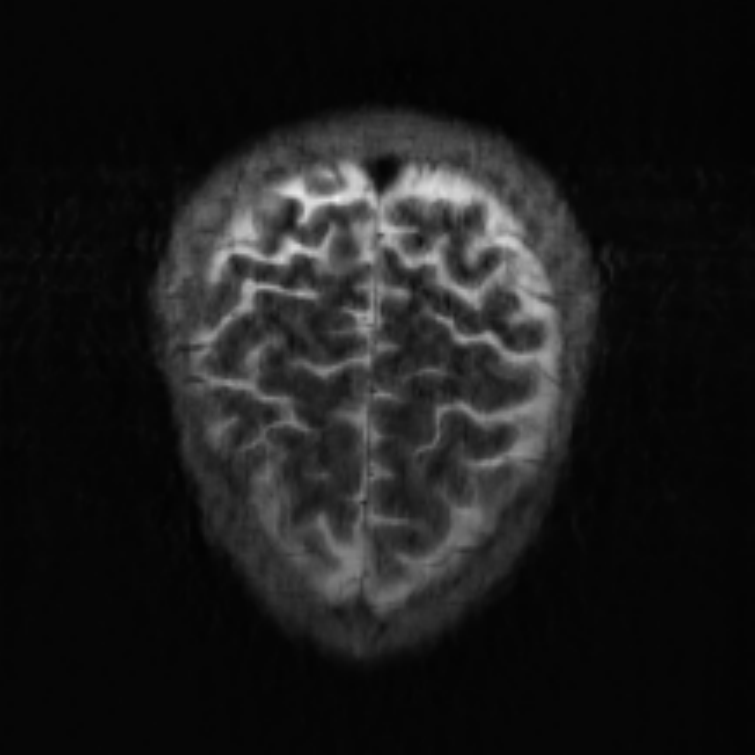}}; 
      \begin{scope}[x=10,y=10]
        \node[draw,minimum height=0.55cm,minimum width=0.65cm] (B1) at (-1.2,6.6) {}; 
        \node (img1) at (-3.95,1.6) {\includegraphics[width=0.08\textwidth,trim={2.3cm 3.8cm 4cm 2.7cm},  clip]{files/anatomy_shift/knee_sup_on_brain.pdf}}; 
        \draw (img1.south west) rectangle (img1.north east);
        \draw (B1) -- (img1);
      \end{scope}
    \end{tikzpicture} \end{adjustbox} &
    \begin{adjustbox}{valign=m}
    \begin{tikzpicture}[boximg]
      \node[anchor=south] (img) {\includegraphics[width=0.25\textwidth]{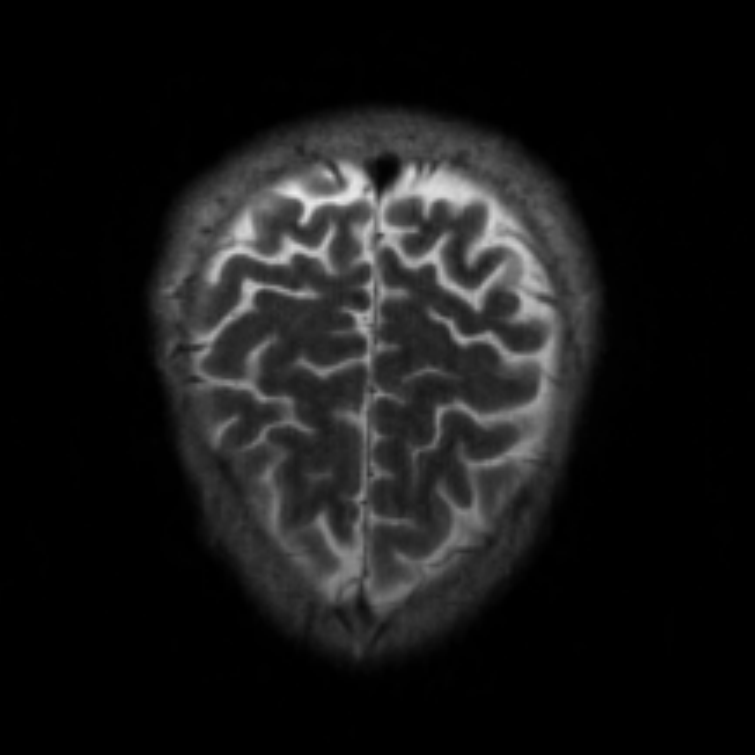}};
      \begin{scope}[x=10,y=10]
        \node[draw,minimum height=0.55cm,minimum width=0.65cm] (B1) at (-1.2,6.6) {};
        \node (img1) at (-3.95,1.6) {\includegraphics[width=0.08\textwidth,trim={2.3cm 3.8cm 4cm 2.7cm}, clip]{files/anatomy_shift/ttt_on_brain.pdf}};
        \draw (img1.south west) rectangle (img1.north east);
        \draw (B1) -- (img1);
      \end{scope}
    \end{tikzpicture} \end{adjustbox} &
    \begin{adjustbox}{valign=m}
    \begin{tikzpicture}[boximg]
      \node[anchor=south] (img) {\includegraphics[width=0.25\textwidth]{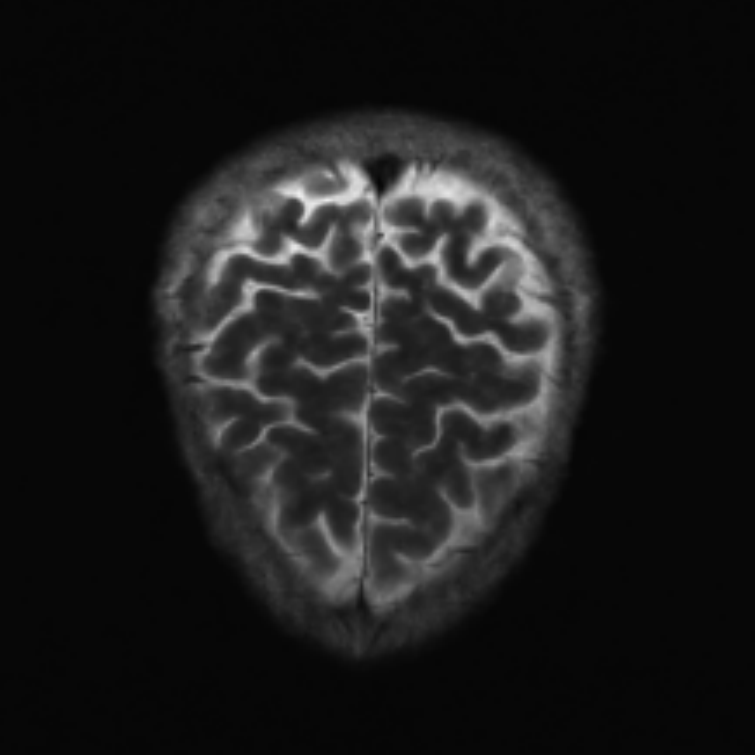}};
      \begin{scope}[x=10,y=10]
        \node[draw,minimum height=0.55cm,minimum width=0.65cm] (B1) at (-1.2,6.6) {};
        \node (img1) at (-3.95,1.6) {\includegraphics[width=0.08\textwidth,trim={2.3cm 3.8cm 4cm 2.7cm}, clip]{files/anatomy_shift/brain_sup_on_brain.pdf}};
        \draw (img1.south west) rectangle (img1.north east);
        \draw (B1) -- (img1);
      \end{scope}
    \end{tikzpicture} \end{adjustbox} &
    \begin{adjustbox}{valign=m}
    \begin{tikzpicture}[boximg]
      \node[anchor=south] (img) {\includegraphics[width=0.25\textwidth]{files/anatomy_shift/orig.pdf}};
      \begin{scope}[x=10,y=10]
        \node[draw,minimum height=0.55cm,minimum width=0.65cm] (B1) at (-1.2,6.6) {};
        \node (img1) at (-3.95,1.6) {\includegraphics[width=0.08\textwidth,trim={2.3cm 3.8cm 4cm 2.7cm}, clip]{files/anatomy_shift/orig.pdf}};
        \draw (img1.south west) rectangle (img1.north east);
        \draw (B1) -- (img1);
      \end{scope}
    \end{tikzpicture} \end{adjustbox}
  \\ \rule{0pt}{3ex}
  & \begin{tabular}{@{}c@{}} SSIM: 0.6709 \end{tabular}  & \begin{tabular}{@{}c@{}} SSIM: 0.6902 \end{tabular} & \begin{tabular}{@{}c@{}}  SSIM: 0.6918 \end{tabular}  \\
  \rule{0pt}{3ex}
  \begin{tabular}{@{}c@{}} \large dataset \\ \large shift \end{tabular} &
  \begin{adjustbox}{valign=m}
    \begin{tikzpicture}[boximg]
      \node[anchor=south] (img) {\scalebox{1}[-1]{\includegraphics[width=0.25\textwidth]{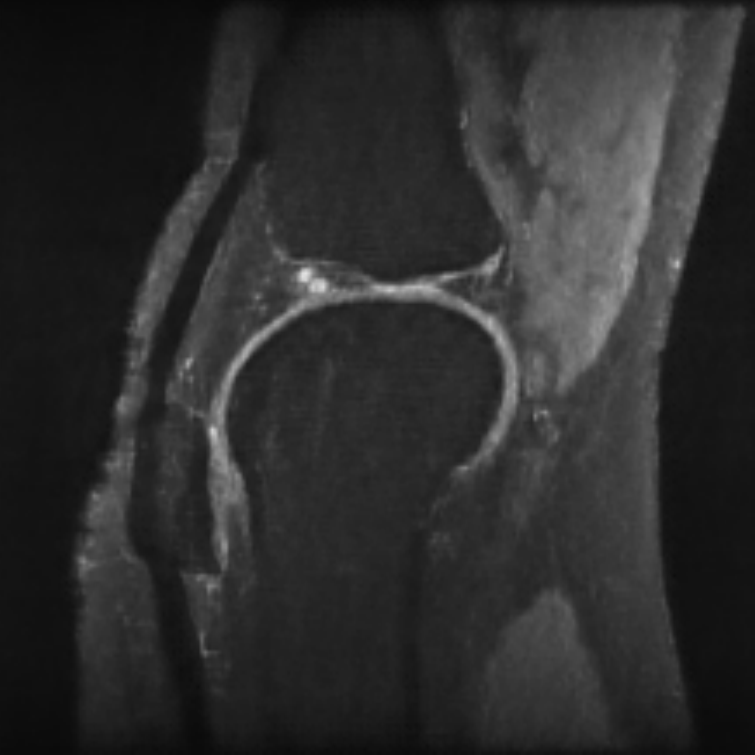}}};
      \begin{scope}[x=10,y=10]
        \node[draw,minimum height=0.55cm,minimum width=0.65cm] (B1) at (2.7,3.5) {}; 
        \node (img1) at (-3.95,1.6) {\scalebox{1}[-1]{\includegraphics[width=0.08\textwidth,trim={4.8cm 4.7cm 1.5cm 1.8cm}, clip]{files/dataset_shift/fs_sup_on_stan.pdf}}}; 
        \draw (img1.south west) rectangle (img1.north east);
        \draw (B1) -- (img1);
      \end{scope}
  \end{tikzpicture} \end{adjustbox} &
  \begin{adjustbox}{valign=m}
    \begin{tikzpicture}[boximg]
      \node[anchor=south] (img) {\scalebox{1}[-1]{\includegraphics[width=0.25\textwidth]{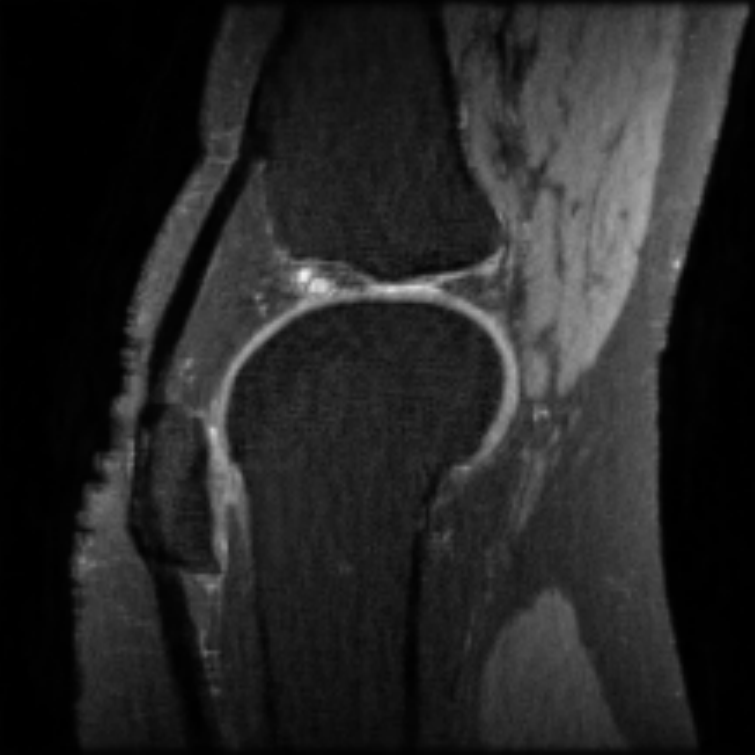}}};
      \begin{scope}[x=10,y=10]
        \node[draw,minimum height=0.55cm,minimum width=0.65cm] (B1) at (2.7,3.5) {};
        \node (img1) at (-3.95,1.6) {\scalebox{1}[-1]{\includegraphics[width=0.08\textwidth,trim={4.8cm 4.7cm 1.5cm 1.8cm}, clip]{files/dataset_shift/ttt_on_stan.pdf}}};
        \draw (img1.south west) rectangle (img1.north east);
        \draw (B1) -- (img1);
      \end{scope}
    \end{tikzpicture} \end{adjustbox} &
  \begin{adjustbox}{valign=m}
    \begin{tikzpicture}[boximg]
      \node[anchor=south] (img) {\scalebox{1}[-1]{\includegraphics[width=0.25\textwidth]{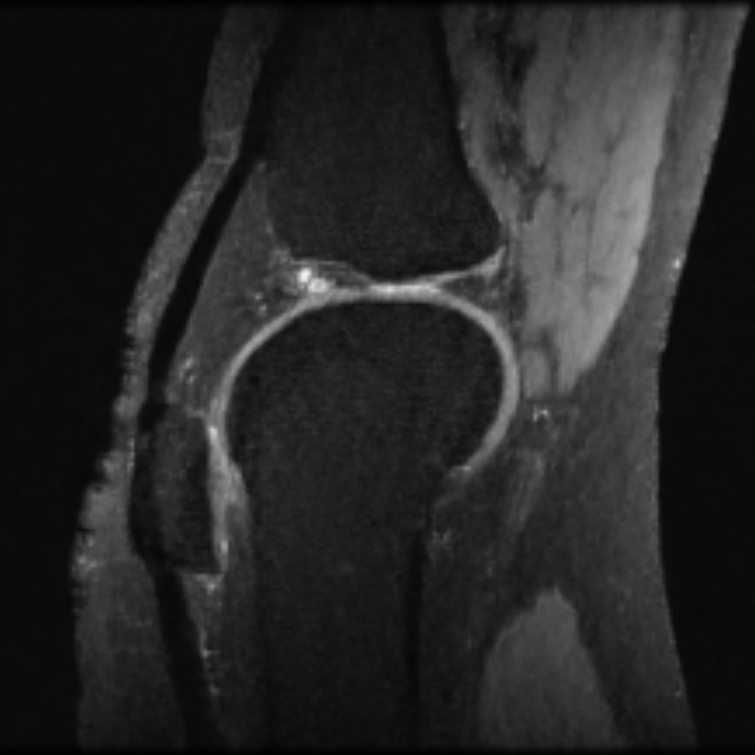}}};
      \begin{scope}[x=10,y=10]
        \node[draw,minimum height=0.55cm,minimum width=0.65cm] (B1) at (2.7,3.5) {};
        \node (img1) at (-3.95,1.6) {\scalebox{1}[-1]{\includegraphics[width=0.08\textwidth,trim={4.8cm 4.7cm 1.5cm 1.8cm}, clip]{files/dataset_shift/stan_sup_on_stan.pdf}}};
        \draw (img1.south west) rectangle (img1.north east);
        \draw (B1) -- (img1);
      \end{scope}
    \end{tikzpicture} \end{adjustbox} &
  \begin{adjustbox}{valign=m}
    \begin{tikzpicture}[boximg]
      \node[anchor=south] (img) {\scalebox{1}[-1]{\includegraphics[width=0.25\textwidth]{files/dataset_shift/orig.pdf}}};
      \begin{scope}[x=10,y=10]
        \node[draw,minimum height=0.55cm,minimum width=0.65cm] (B1) at (2.7,3.5) {};
        \node (img1) at (-3.95,1.6) {\scalebox{1}[-1]{\includegraphics[width=0.08\textwidth,trim={4.8cm 4.7cm 1.5cm 1.8cm}, clip]{files/dataset_shift/orig.pdf}}};
        \draw (img1.south west) rectangle (img1.north east);
        \draw (B1) -- (img1);
      \end{scope}
    \end{tikzpicture} \end{adjustbox}
  \\ \rule{0pt}{3ex}
  & \begin{tabular}{@{}c@{}}  SSIM: 0.8705 \end{tabular}  & \begin{tabular}{@{}c@{}}   SSIM: 0.9317 \end{tabular} & \begin{tabular}{@{}c@{}}  SSIM: 0.9248 \end{tabular}  \\
  \rule{0pt}{3ex}
  \begin{tabular}{@{}c@{}} \large modality \\ \large shift \end{tabular} &
  \begin{adjustbox}{valign=m}
    \begin{tikzpicture}[boximg]
      \node[anchor=south] (img) {\includegraphics[width=0.25\textwidth]{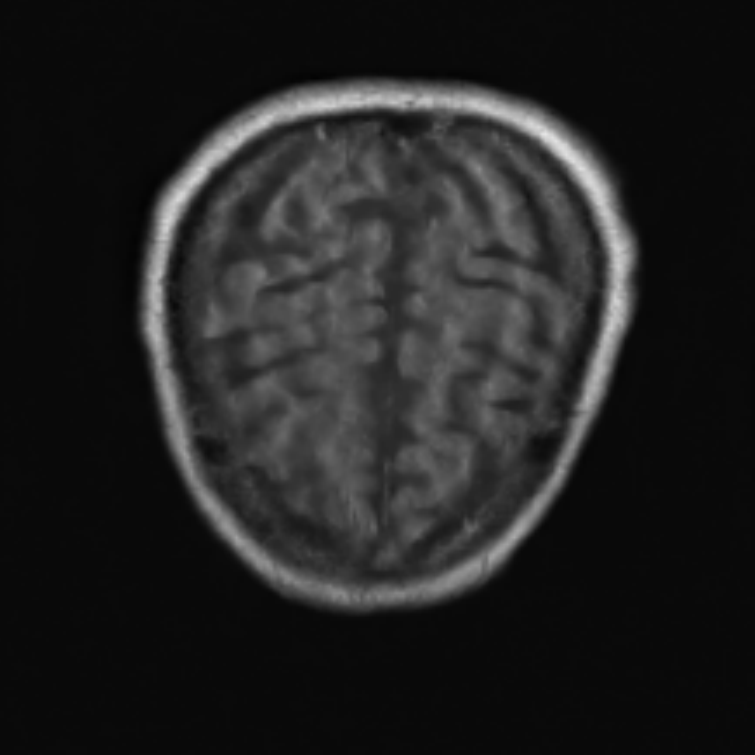}};
      \begin{scope}[x=10,y=10] 
        \node[draw,minimum height=0.55cm,minimum width=0.65cm] (B1) at (-1.2,6.8) {}; 
        \node (img1) at (-3.95,1.6) {\includegraphics[width=0.08\textwidth,trim={2.3cm 3.8cm 4cm 2.7cm}, clip]{files/acquisition_shift/t2_sup_on_t1.pdf}}; 
        \draw (img1.south west) rectangle (img1.north east);
        \draw (B1) -- (img1);
      \end{scope}
    \end{tikzpicture} \end{adjustbox} &
  \begin{adjustbox}{valign=m}
    \begin{tikzpicture}[boximg]
      \node[anchor=south] (img) {\includegraphics[width=0.25\textwidth]{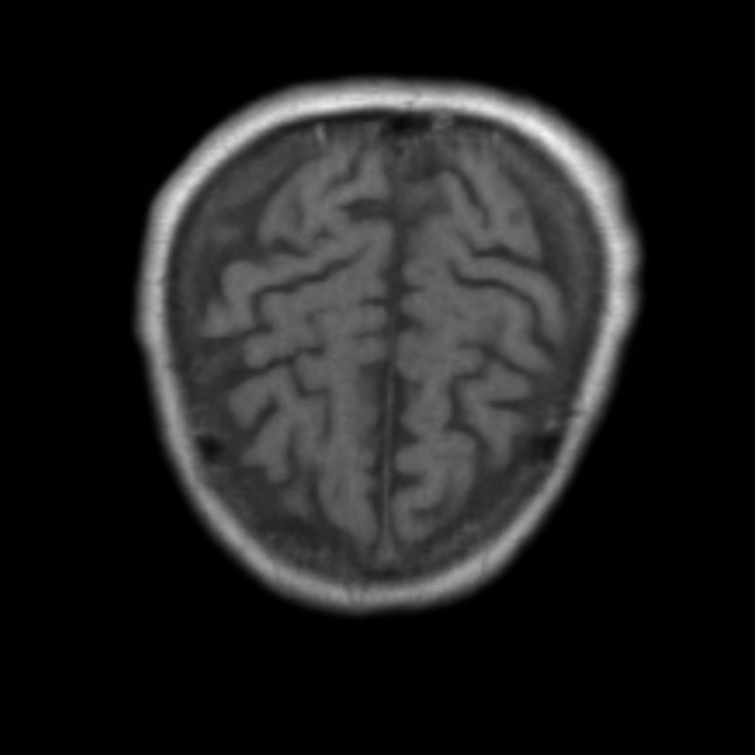}};
      \begin{scope}[x=10,y=10]
        \node[draw,minimum height=0.55cm,minimum width=0.65cm] (B1) at (-1.2,6.8) {};
        \node (img1) at (-3.95,1.6) {\includegraphics[width=0.08\textwidth,trim={2.3cm 3.8cm 4cm 2.7cm}, clip]{files/acquisition_shift/ttt_on_t1.pdf}};
        \draw (img1.south west) rectangle (img1.north east);
        \draw (B1) -- (img1);
      \end{scope}
    \end{tikzpicture} \end{adjustbox} &
  \begin{adjustbox}{valign=m}
    \begin{tikzpicture}[boximg]
      \node[anchor=south] (img) {\includegraphics[width=0.25\textwidth]{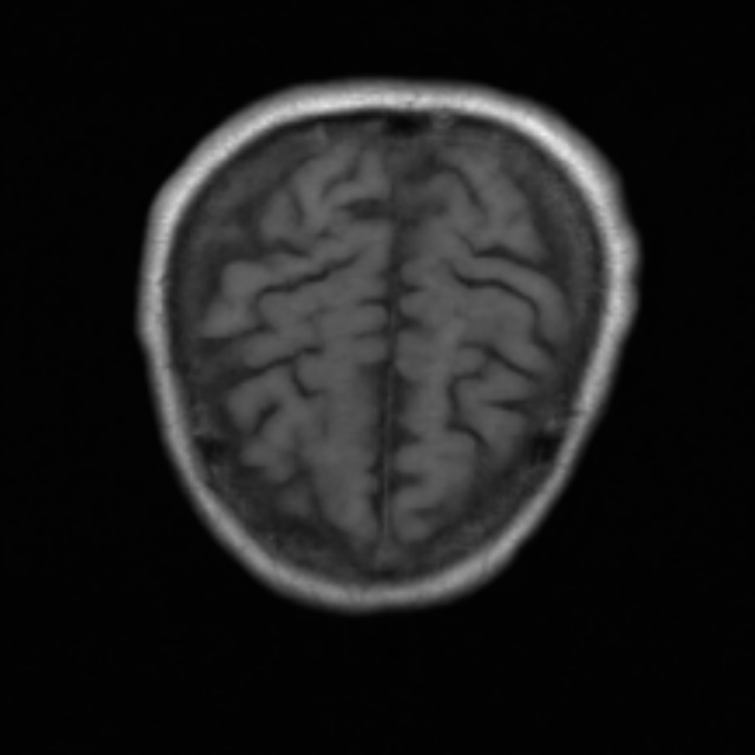}};
      \begin{scope}[x=10,y=10]
        \node[draw,minimum height=0.55cm,minimum width=0.65cm] (B1) at (-1.2,6.8) {};
        \node (img1) at (-3.95,1.6) {\includegraphics[width=0.08\textwidth,trim={2.3cm 3.8cm 4cm 2.7cm}, clip]{files/acquisition_shift/t1_sup_on_t1.pdf}};
        \draw (img1.south west) rectangle (img1.north east);
        \draw (B1) -- (img1);
      \end{scope}
    \end{tikzpicture} \end{adjustbox} &
  \begin{adjustbox}{valign=m}
    \begin{tikzpicture}[boximg]
      \node[anchor=south] (img) {\includegraphics[width=0.25\textwidth]{files/acquisition_shift/orig.pdf}};
      \begin{scope}[x=10,y=10]
        \node[draw,minimum height=0.55cm,minimum width=0.65cm] (B1) at (-1.2,6.8) {};
        \node (img1) at (-3.95,1.6) {\includegraphics[width=0.08\textwidth,trim={2.3cm 3.8cm 4cm 2.7cm}, clip]{files/acquisition_shift/orig.pdf}};
        \draw (img1.south west) rectangle (img1.north east);
        \draw (B1) -- (img1);
      \end{scope}
    \end{tikzpicture} \end{adjustbox} 
  \\ \rule{0pt}{3ex}
  & \begin{tabular}{@{}c@{}}  SSIM: 0.8572 \end{tabular}  & \begin{tabular}{@{}c@{}}   SSIM: 0.9107 \end{tabular} & \begin{tabular}{@{}c@{}}  SSIM: 0.9122 \end{tabular}  \\
  \rule{0pt}{3ex}
  \begin{tabular}{@{}c@{}} \large acceleration \\ \large shift \end{tabular} &
  \begin{adjustbox}{valign=m}
    \begin{tikzpicture}[boximg]
      \node[anchor=south] (img) {\scalebox{1}[-1]{\includegraphics[width=0.25\textwidth]{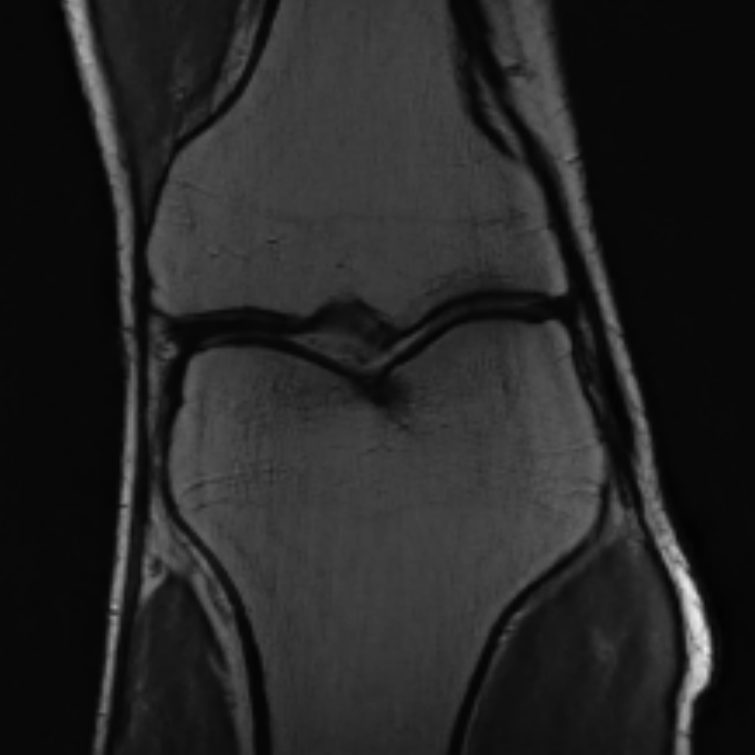}}};
      \begin{scope}[x=10,y=10]
        \node[draw,minimum height=0.55cm,minimum width=0.65cm] (B1) at (-2.4,7) {};
        \node (img1) at (-3.95,1.6) {\scalebox{1}[-1]{\includegraphics[width=0.08\textwidth,trim={1.5cm 2.6cm 4.8cm 3.9cm}, clip]{files/acceleration_shift/4x_sup_on_2x.pdf}}}; 
        \draw (img1.south west) rectangle (img1.north east);
        \draw (B1) -- (img1);
      \end{scope}
    \end{tikzpicture} \end{adjustbox} &
  \begin{adjustbox}{valign=m}
    \begin{tikzpicture}[boximg]
      \node[anchor=south] (img) {\scalebox{1}[-1]{\includegraphics[width=0.25\textwidth]{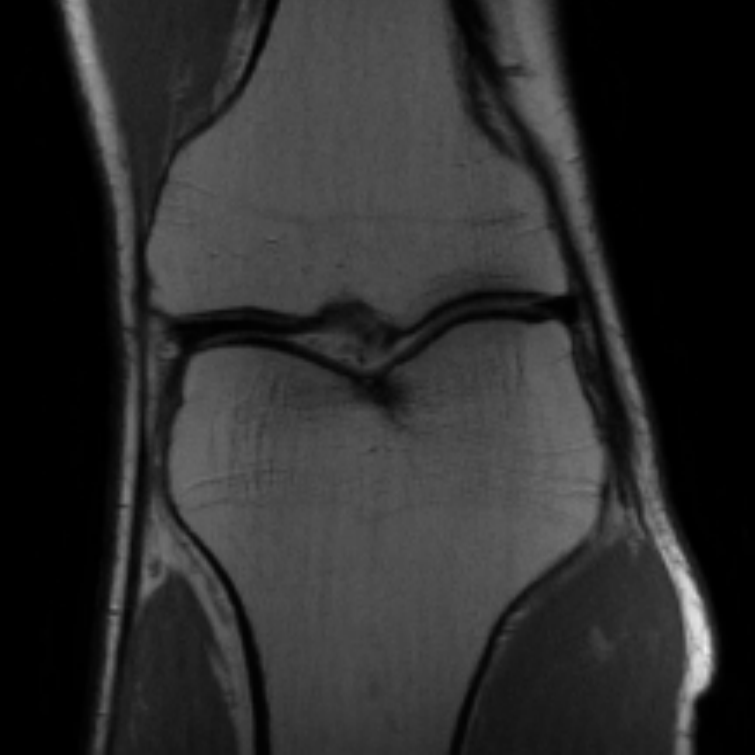}}};
      \begin{scope}[x=10,y=10]
        \node[draw,minimum height=0.55cm,minimum width=0.65cm] (B1) at (-2.4,7) {};
        \node (img1) at (-3.95,1.6) {\scalebox{1}[-1]{\includegraphics[width=0.08\textwidth,trim={1.5cm 2.6cm 4.8cm 3.9cm}, clip]{files/acceleration_shift/4x_self_on_2x_ttt.pdf}}};
        \draw (img1.south west) rectangle (img1.north east);
        \draw (B1) -- (img1);
      \end{scope}
    \end{tikzpicture} \end{adjustbox} &
  \begin{adjustbox}{valign=m}
    \begin{tikzpicture}[boximg]
      \node[anchor=south] (img) {\scalebox{1}[-1]{\includegraphics[width=0.25\textwidth]{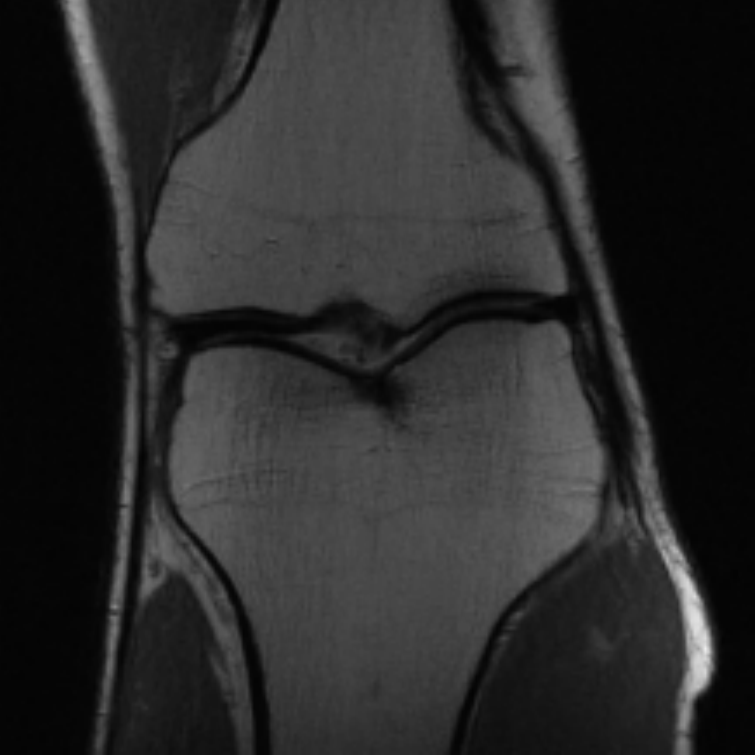}}};
      \begin{scope}[x=10,y=10]
        \node[draw,minimum height=0.55cm,minimum width=0.65cm] (B1) at (-2.4,7) {};
        \node (img1) at (-3.95,1.6) {\scalebox{1}[-1]{\includegraphics[width=0.08\textwidth,trim={1.5cm 2.6cm 4.8cm 3.9cm}, clip]{files/acceleration_shift/2x_sup_on_2x.pdf}}};
        \draw (img1.south west) rectangle (img1.north east);
        \draw (B1) -- (img1);
      \end{scope}
    \end{tikzpicture} \end{adjustbox} &
  \begin{adjustbox}{valign=m}
    \begin{tikzpicture}[boximg]
      \node[anchor=south] (img) {\scalebox{1}[-1]{\includegraphics[width=0.25\textwidth]{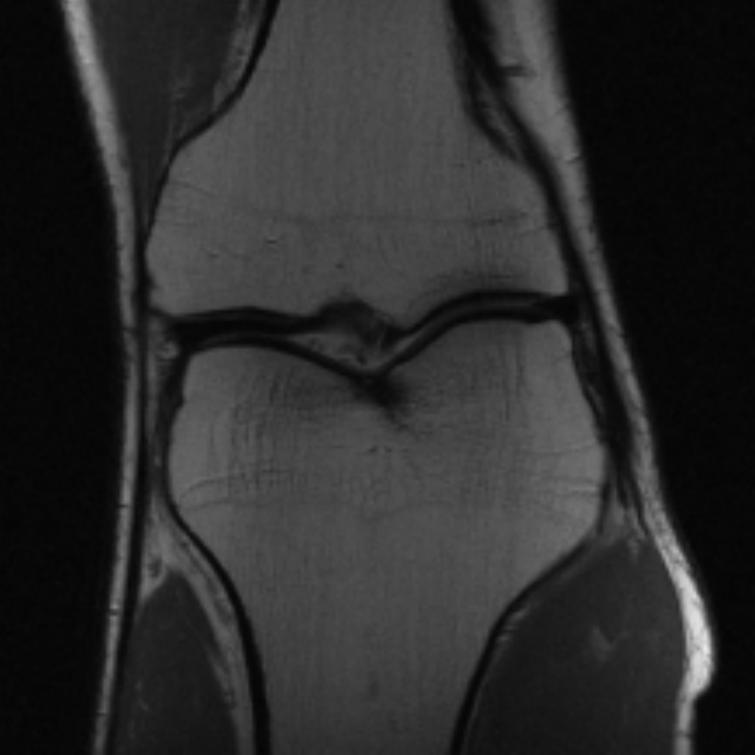}}};
      \begin{scope}[x=10,y=10]
        \node[draw,minimum height=0.55cm,minimum width=0.65cm] (B1) at (-2.4,7) {};
        \node (img1) at (-3.95,1.6) {\scalebox{1}[-1]{\includegraphics[width=0.08\textwidth,trim={1.5cm 2.6cm 4.8cm 3.9cm}, clip]{files/acceleration_shift/orig.pdf}}};
        \draw (img1.south west) rectangle (img1.north east);
        \draw (B1) -- (img1);
      \end{scope}
    \end{tikzpicture} \end{adjustbox} 
  \\
\end{tabular}
\end{adjustbox}
\captionsetup{skip=10pt}
\captionof{figure}{Our domain adaptation method yields a high perceptual quality by closing the distribution shift performance gap. All images are U-Net reconstructions from 4x under-sampled measurements (except the last row which is for 2x acceleration).
}
\vspace{-8pt}
\label{fig:gap}
\end{table}

\paragraph{Anatomy shift:} 
Here, the distribution $P$ are knee images, and $Q$ are brain images.  
The knee training set consists of 376 PD knee slices 
of the fastMRI training dataset and the brain training set consists of
310 AXT2 brain slices of the fastMRI training dataset. We train U-Net and VarNet on both training sets 
separately, and test them on a brain test set consisting of 100 AXT2 brain slices of the fastMRI validation dataset, this yields the numbers for supervised training in Table~\ref{tab:results}. Next, we performed the same experiment, but this time we included self-supervision into the training, and we applied test-time-training as described above. 

\paragraph{Dataset shift:} 
Here, the distribution $P$ is a subset of the fastMRI dataset, and $Q$ is the Stanford dataset. The version of the fastMRI training set consists of 366 PDFS knee slices, and the Stanford training set consists of 308 PDFS knee slices. We then repeat the same experiment explained above for anatomy shift.

\paragraph{Modality shift:} 
Here, the $P$ and $Q$ distributions are the AXT2 and AXT1PRE slices of the fastMRI brain dataset. 
The AXT2 training set consists of 310 slices, and the AXT1PRE consist of 320 slices. 
We then repeat the same experiment explained above. 

\paragraph{Acceleration shift:} 
Here, the distribution $P$ are 4x accelerated knee measurements, and $Q$ are 2x accelerated knee measurements. 
The 4x training set consists of 376 PD knee slices 
of the fastMRI training dataset and the 2x training set consists of the same 376 PD knee slices but accelerated 2 times instead of 4. 

Table~\ref{tab:results} and Table~\ref{tab:runtimes} contain the results of the experiments. Figure~\ref{fig:gap} contains example reconstructions, more examples are in the appendix in Figures~\ref{fig:anat-shift}, \ref{fig:data-shift}, \ref{fig:acq-shift}, and \ref{fig:acc-shift}.


\subsection{Discussion on the results} 
We draw the following conclusions from the results.

%
\paragraph{TTT essentially closes the distribution shift performance gap.}
    For the four considered natural distribution shifts, our method closes the distribution shift performance gap by 87-99\%. 
    Thus, our method closes the gap for practical purposes for those four distribution shifts.
    This is also reflected in a significant increase in perceptual quality (see Figure~\ref{fig:gap} and the appendix). 
    %
    
    \paragraph{TTT also slightly improves in-distribution performance,} which is not surprising as it tunes the parameters on each instance individually as opposed to finding the parameters that work best on average on all slices. That is why we measure the performance gap with and without TTT included.

\paragraph{Increased computation cost.}
    Our approach offers significant reduction in the distribution shift performance gap. However, TTT comes at the cost of more computations at the inference as shown in Table~\ref{tab:runtimes}. 
    %

\subsection{Ablation studies}

\paragraph{Including self-supervision in training is critical for performance.}    
    TTT on a model trained only on the supervised loss (without including the self-supervised loss during training) gives only minimal (for anatomy shift) to no (for dataset shift) performance improvement. 
    %

\begin{figure}[t!]
\centering
  \begin{center}
    \begin{tikzpicture}
    \begin{groupplot}[
    y tick label style={/pgf/number format/.cd,fixed,precision=4},scaled y ticks = false,
    legend style={at={(1,0.6)} , 
    /tikz/every even column/.append style={column sep=-0.1cm}
     },
             group
             style={group size= 1 by 3, xlabels at=edge bottom,
             horizontal sep=0.1cm, vertical sep=0.2cm,
             }, 
             width=0.33\textwidth,height=0.2\textwidth,
             ]
    \nextgroupplot[
        ylabel style={rotate=-90,xshift=-2ex},
        ylabel={\parbox{3.3cm}{\centering test-time-train error \\ $\mc L_{\text{self}}( \vy_{\text{train}}, f_{\vtheta_t}(\vy_{\text{train}}) )$}},
        title={chosen early-stopping point}, title style={xshift=6ex,}, xmode=log,yticklabel pos=left,xticklabels={,,}]
    \addplot +[mark=none,steelblue] table[x=x,y=err]{./files/es.txt};
    %
    \nextgroupplot[
    ylabel style={rotate=-90,xshift=-2ex},
    ylabel={ 
    \parbox{3.3cm}{\centering test-time-val. error \\
    $\norm[1]{ \vy_{\text{val}} - f_{\vtheta_t}(\vy_{\text{val}}) }$ }},xlabel={},title={},xmode=log,xticklabels={,,}]
    	\addplot +[mark=none,steelblue] table[x=x,y=verr]{./files/es.txt};
    %
    \nextgroupplot[
    ylabel style={rotate=-90,xshift=-2ex},
    ylabel={\parbox{3.3cm}{ \centering true accuracy \\ $\text{SSIM}(f_{\vtheta_t}(\vy),\vx)$}},
    xlabel={iteration $t$},title={},xmode=log,yticklabel pos= left,]
    	\addplot +[mark=none,steelblue] table[x=x,y=score]{./files/es.txt};
    \end{groupplot}
    \draw [dashed,help lines,black] (3.25,2)-|(3.25,-3.88);
    \end{tikzpicture}
  \end{center}
\vspace{-3pt}
\caption{
Test-time-training error, validation error, and error with respect to the original (unknown) image $\vx$. 
The given measurement $\vy$ is split into a part for test-time-training ($\vy_{\text{train}}$) and into a part for validation ($\vy_{\text{val}}$), which in turn is used to determine the early-stopping time. Comparing the validation error with the true accuracy shows that i) the validation error is a good proxy for the true accuracy, and ii) early stopping is critical for the performance of test-time training.
}
\vspace{-8pt}
\label{fig:es}
\end{figure}
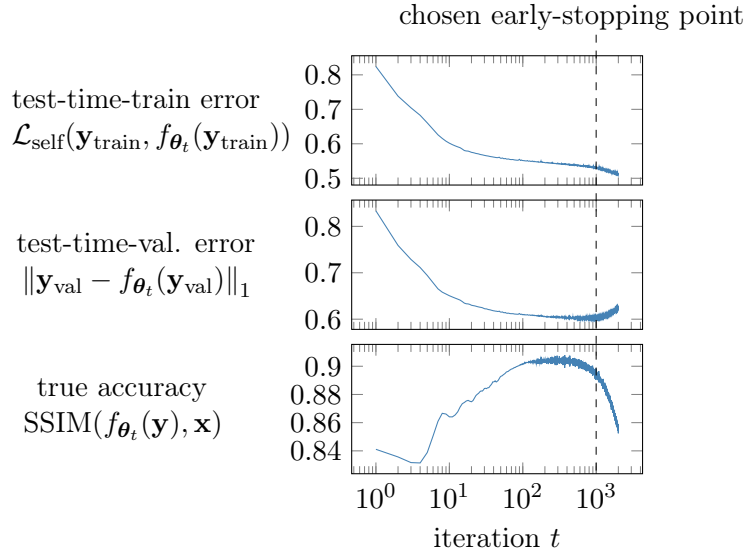

\paragraph{Early-stopping TTT is critical for performance.}
TTT is performed on the under-sampled test measurement and overfits to it without our early-stopping mechanism, which chooses the early-stopping time based on a hold-out set obtained from the under-sampled measurement.  
Figure~\ref{fig:es} illustrates that early stopping is critical for TTT to perform well. As depicted, the self-validation error in the middle panel (which is computed using a fraction of under-sampled test measurement) yields a good early-stopping time, in that the validation loss is inversely correlated with the true accuracy, as expected.  

\paragraph{Variants of TTT.} 
TTT in image classification is typically quite different than the version of TTT proposed here, in that a network with two heads is trained, one part on a supervised loss and one part at a self-supervised loss on an auxiliary task, such as predicting rotations of an image. At inference, TTT is performed on the self-supervised loss~\cite{sun2020test}.
We experimented with the analogous idea for image reconstruction: We took a U-Net with two decoders and a joint encoder, and trained the network on a self-supervised denoising problem and the supervised compressive sensing problem. At inference, we performed TTT on the denoising problem. This approach fails to improve model robustness (see Figure~\ref{fig:semi-ttt} in the supplement). 

We also experimented with another variant of TTT that also does not work well. We exploited the idea behind the CycleGAN~\citep{zhu2017unpaired} to build CycleU-Net. CylcleU-Net comprises two U-Nets in tandem and is trained based on two forward passes: (1) a supervised pass mapping the ground-truth image to itself, and (2) a self-supervised pass mapping the low-quality image to itself (by switching the places of the two U-Nets). This type of training enables TTT w.r.t. the self-supervised pass at inference. Table~\ref{tab:gap-cyc} in the supplement shows that only 29\% of the gap is closed for anatomy shift when this approach is employed.

\section{Test-time training can provably adapt to a distribution shift}

\label{sec:theory}

In this section, we discuss an example illustrating that test-time training with an appropriate loss can optimally adapt to a particular distribution shift.

Consider the problem of denoising a signal that lies in a subspace. The training distribution draws a signal from an unknown $d$-dimensional subspace and  corrupts it by Gaussian noise with noise variance $\sigma^2$: 
\begin{align*}
    P \colon \vy = \vx + \vz, \vx = \mU \vc, \vc \sim \mc N(0,\mI), \vz \sim \mc N(0, \mI \sigma^2).
\end{align*}
The test distribution draws a signal from the same subspace, but corrupts it with noise that has a different noise variance: \begin{align*}
    Q \colon \vy = \vx + \vz, \vx = \mU \vc, \vc \sim \mc N(0,\mI), \vz \sim \mc N(0, \mI \varsigma^2).
\end{align*}
Here, $\mU \in \reals^{n\times d}$ is an orthonormal basis for the signal subspace. 
We consider a linear denoiser of the form $f(\vy) = \alpha \mV \transp{\mV} \vy$, where $\mV \in \reals^{n\times d}$ is an orthonormal basis and $\alpha \in [0,1]$ is a scalar. 
We measure performance in terms of the population risk 
\[
R_Q(\alpha,\mV) = \EX[Q]{ \norm[2]{\vx -  \alpha \mV \transp{\mV}  \vy}^2 }.
\]
The population risk is minimized by the estimator $\alpha = \frac{1}{1+\varsigma^2}$ and $\mV = \mU$. So the optimal estimator for the unknown signal $\vx$ is $\hat \vx = \frac{1}{1+\varsigma^2} \mU \transp{\mU} \vy$ and it projects the observation onto the subspace and then shrinks with a coefficient that is dependent on the noise variance, the larger the noise variance, the more the estimator shrinks.

We train the method on the supervised loss
\begin{align}
    \mc L_P(\alpha,\mV)
    =
    R_P(\alpha,\mV).
\end{align}
For simplicity, we choose as the supervised loss the population loss, because we are not interested in finite-sample effects. This corresponds to a setup where we have abundant training data on the distribution $P$. 

A minimizer of the loss is $\alpha = \frac{1}{1+\sigma^2}, \mV = \mU$, thus minimizing on the loss gives an estimator that is optimal on the distribution $P$. However, the estimator is sub-optimal on the distribution $Q$. 

We next obtain an observation $\vy$ from the distribution $Q$, and our goal is to estimate the corresponding signal $\vx$. Towards this end, we first perform test-time-training on the distribution $Q$, i.e., we minimize the self-supervised loss
\begin{align*}
L_{SS}(\alpha,\mU,\vy) 
&\!=\! 
\norm[2]{\vy \!-\! \alpha \mU \transp{\mU}\vy }^2 \!+\!
\frac{2\alpha d}{n - d} \norm[2]{(\mI \!-\! \mU \transp{\mU}) \vy}^2
\end{align*}
over the scale parameter $\alpha$. Note that we already obtained $\mV = \mU$ from the training on the distribution $P$. For $d$ and $n$ large, the self-supervised loss $L_{SS}(\alpha,\mU,\vy)$ concentrates around the expectation, which can be shown to be $(1-\alpha)^2 d + \alpha^2 d \varsigma^2$. Minimizing this over $\alpha$ gives $\alpha = \frac{1}{1+\varsigma^2}$. See proof in Appendix~\ref{sec:proof-analytical-example}.
Thus, for this example, test-time-training actually yields the optimal estimator for the signal $\vx$ under the distribution $Q$! 


This example demonstrates that in theory, test-time-training 
can yield an optimal estimator under a distribution shift. 
It also illustrates that for TTT to work, the main task and the self-supervised learning task have to be related, and the choice of the self-supervised loss has to take this relation into account. Also note that this particular self-supervised loss will not work for a different distribution shift, it is tailored to the anticipated shift in noise variance. This highlights that different distribution shifts might require different self-supervised losses for test-time-training to work.

\section{Conclusion}
Distribution shifts are a key limiting factor in deep learning based imaging. 
In this paper, we proposed a novel domain adaptation method based on self-supervised training and test-time training (TTT) that 
reduces the distribution shift performance gap for four natural distribution shift by 87-99\%, and thus gives significantly better image quality at inference.

It is perhaps surprising that TTT works so well for natural distribution shifts, in particular since a variety of domain adaptation methods for classification only give a marginal improvement (if at all) for natural distribution shifts~\citet{miller2021accuracy}. 
However, image reconstruction problems are much more amenable for domain adaptation methods, since at test time, we are given an entire measurement of an image, which contains lots of information about the image and a potential distribution shift. We exploit this information through TTT and to perform early stopping during TTT.
Contrary, in an image classification setup, we are not given any information about the label and thus it might be much harder to adapt at test time.

Many important questions remain: 
Due to the TTT, our method has significantly higher computational cost at inference than using a plain neural network without TTT.  Reducing this computational cost is desirable. Moreover, our method is specific to compressive sensing problems, and thus developing TTT methods that are also applicable to inverse problems beyond compressive sensing is an important future direction. Finally, it is important to gain a better theoretical understanding of the mechanisms making TTT work.

\section*{Reproducibility}

Our repository at 
\href{https://github.com/MLI-lab/ttt_for_deep_learning_cs}
{https://github.com/MLI-lab/ttt\_for\_deep\_learning\_cs}
contains the code to reproduce all results in this paper.

\section*{Acknowledgment}\label{sec:ack}

M. Zalbagi Darestani and R. Heckel are (partially) supported by NSF award IIS-1816986, and R. Heckel acknowledges support by the Institute of Advanced Studies at the Technical University of Munich, and the Deutsche Forschungsgemeinschaft (DFG, German Research Foundation) - 456465471, 464123524. J. Liu is supported by the 6G Future Lab Bavaria.

\printbibliography

\newpage

\appendix

\section{Intensity distribution changes under modality shift}

One of the natural distribution shifts we study in our work is a modality shift, where the acquisition mode of train and test domains differ. This difference results in a shift in intensity distributions that is illustrated in Figure~\ref{fig:intensity} for the shift from T2 to T1PRE brain images that we consider.


\begin{table}[h!]
\setlength{\tabcolsep}{-1pt}
\centering
\begin{tabular}{cc}
  \footnotesize T2 intensity distribution & \footnotesize T1PRE intensity distribution \\
    \begin{tikzpicture}
    \pgfplotsset{every x tick scale label/.append style={yshift=1.6em,xshift=1.2em}}
    \begin{axis}[scaled x ticks = true, ymax=1,width=0.25\textwidth,height=0.2\textwidth]
    \addplot[ybar,ybar legend,bar width=10000sp,draw=none, fill=steelblue,opacity=1] table[x index=0,y index=1] {files/acquisition_shift/hist.csv};
    \end{axis}
    \end{tikzpicture} &
    \begin{tikzpicture}
    \pgfplotsset{every x tick scale label/.append style={yshift=1.6em,xshift=1.5em}}
    \begin{axis}[scaled x ticks = true,scaled y ticks = true ,ymax=1,width=0.25\textwidth,height=0.2\textwidth,yticklabels={,,}]
    \addplot[ybar,ybar legend,bar width=10000sp,draw=none, fill=steelblue,opacity=1] table[x index=2,y index=3] {files/acquisition_shift/hist.csv};
    \end{axis}
    \end{tikzpicture} \\
   \raisebox{.5\height}{\includegraphics[width=0.12\textwidth,valign=b]{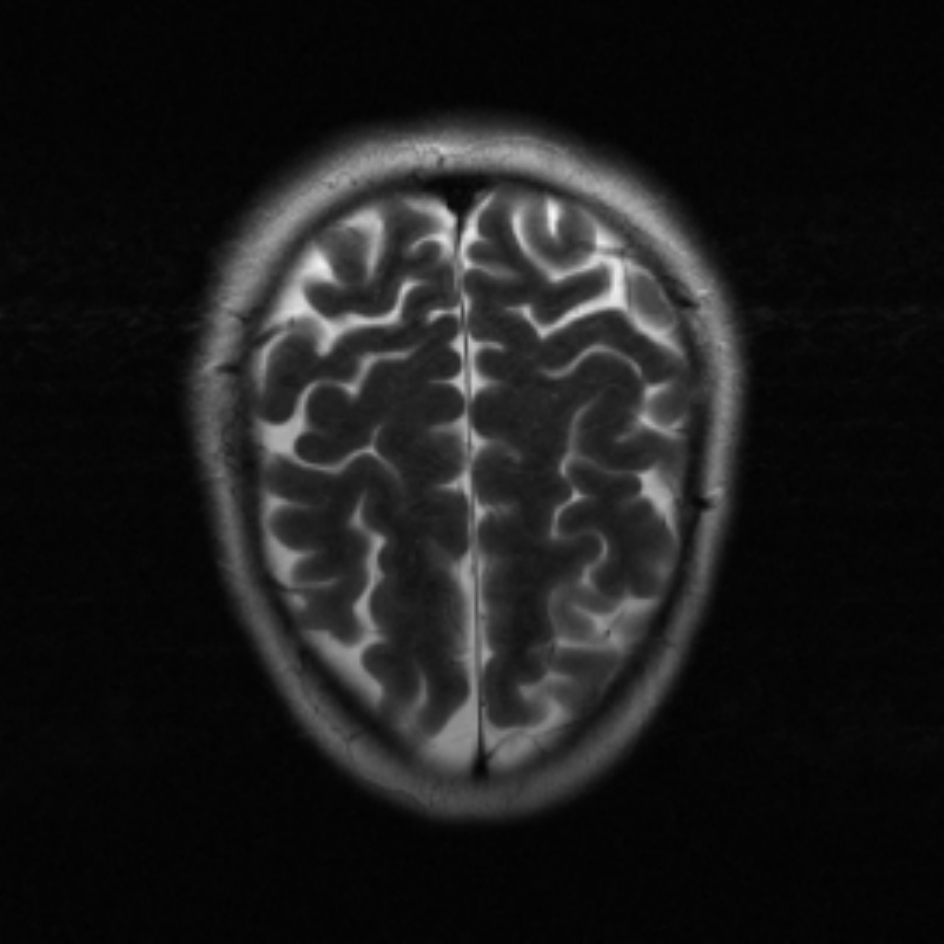}} & \raisebox{.5\height}{\includegraphics[width=0.12\textwidth,valign=b]{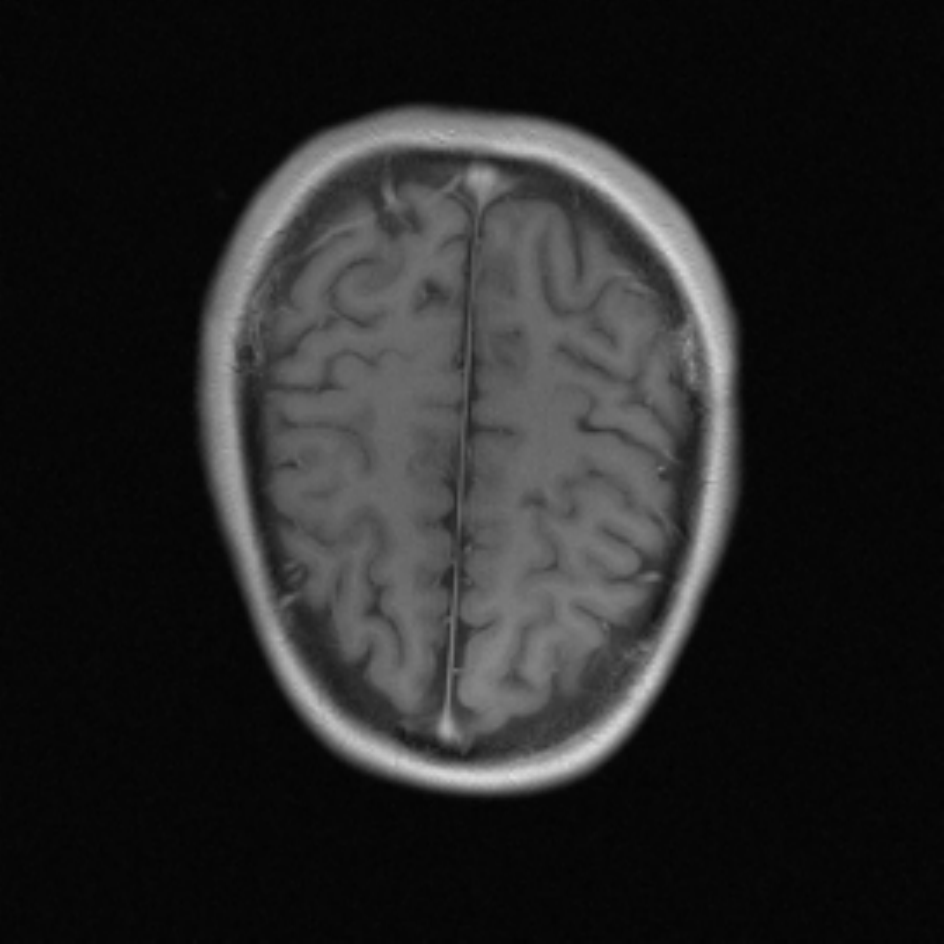}}
\end{tabular}
\vspace{-1cm}
\captionof{figure}{
Modality shifts occurs within an anatomy where the contrast changes via modality change, and thus the intensity distribution of the images changes.
}
\label{fig:intensity}
\end{table}

\section{Comparison to ZS-SSL}\label{sec:zs-ssl}
We mentioned in the introduction that an alternative test-time training (TTT) approach compared to our work is to use the method ZS-SSL~\citep{yaman2021zero}, which was originally introduced as a method for performing reconstruction for a single instance. Given a single under-sampled measurement $\vy$ of an image $\vx$, ZS-SSL creates a dataset $\{(\vy_1,\vy_1',\vy_1''),\ldots,(\vy_K,\vy_K',\vy_K'')\}$ by randomly sampling pixels from $\vy$. Each triplet is a random partition of the under-sampled $k$-space and the proportions are determined by two hyper-parameters $p$ and $p'$.
In each triplet, $\vy_i$ is fed to the network as input, $\vy_i'$ is used in the TTT loss based on the network output, and $\vy_i''$ is used for self-validation to determine when to stop training. Unlike our method that relies on no hyper-parameters, ZS-SSL has three hyper-parameters to tune which are $K$, the number of splits for the synthesized dataset, and $p$ and $p'$ which determine how to split the under-sampled measurement to three partitions

Our TTT approach incorporates a consistency-based self-supervised loss during the supervised pre-training stage, and then TTT is performed w.r.t. that self-supervised loss.
Thus, the main advantage of ZS-SSL over our method is that it does not impose any constraints on the pre-training scheme. Performance-wise, Table~\ref{tab:zs-ssl} compares our method to ZS-SSL for anatomy shift. Both methods are applied to unrolled networks (our approach is applied to VarNet and ZS-SSL is applied to a similar unrolled network explained in~\citep{yaman2021zero}). Table~\ref{tab:zs-ssl} shows that both achieve on-par performance in terms of closing the distribution shift performance gap. However, our approach is computationally cheaper than ZS-SSL which is expected since our method does not rely on dataset synthesis for TTT.

\begin{table*}[th]
\centering
\begin{adjustbox}{width=0.65\textwidth}
\begin{tabular}{l|c|c|c|c}
\toprule 
 \multicolumn{1}{c|}{} & \multicolumn{4}{c}{\bf anatomy shift}  \\
 \multicolumn{1}{c|}{} & \multicolumn{4}{c}{\bf P:knee \; Q:brain}  \\
 \multicolumn{1}{c|}{} & \multicolumn{2}{c}{Our approach} & \multicolumn{2}{c}{ZS-SSL}  \\
 \multicolumn{1}{c|}{setup} & \multicolumn{1}{c}{SSIM} & \multicolumn{1}{c}{TTT runtime} & \multicolumn{1}{c}{SSIM} & \multicolumn{1}{c}{TTT runtime}  \\
 \multicolumn{1}{c|}{} & \multicolumn{1}{c}{} & \multicolumn{1}{c}{(mins/slice)} & \multicolumn{1}{c}{} & \multicolumn{1}{c}{(mins/slice)}  \\
\hline
    train on P test on Q + TTT  & 0.9358 & 12.9 & 0.9343 & 59.4  \\ 
    train on Q test on Q + TTT    & 0.9375 & 4.2 & \textbf{0.9365} & 33.5  \\
    \thickhline
    train on Q test on Q      & \textbf{0.9396}  & - & 0.9331  & -  \\
    train on P test on Q      & 0.8802 & - & 0.9029 & -  \\
\thickhline
    \multicolumn{1}{l|}{fraction of gap closed by TTT} & \multicolumn{2}{c}{97.1\%} & \multicolumn{2}{c}{93.1\%} \\
\bottomrule
\end{tabular}
\end{adjustbox}
\caption{\textbf{Both ZS-SSL and our test-time training (TTT) approach are highly effective for overcoming natural anatomy shift}. SSIM scores are averaged over 30 validation slices of the fastMRI brain dataset. For both methods, TTT is applied with a similar learning rate as the one used during pre-training.}
\vspace{-8pt}
\label{tab:zs-ssl}
\end{table*}

\section{Variants of TTT}
\label{sec:semi-ttt}

As mentioned in the main body, we also experimented with two variants of TTT, which fail to improve robustness. Those variants are discussed here. 

\paragraph{Variant 1.} The first variant is very similar to how TTT is typically performed in image classification. 
Let $E_\vbeta$ be the encoder of the U-Net and let $D_\vtheta$ and $D_\vvartheta$ be two decoders. 
The encoder with the first decoder is trained to solve the main task which is the supervised compressed sensing task, and the encoder with the second decoder is trained to solve an auxiliary task which we take as a self-supervised denoising task. Specifically, we train the method on the loss:
\begin{align}
    \mc L(\vbeta,\vtheta,\vvartheta) &=
    \frac{1}{n}
    \sum_{i=1}^n
    \Bigg(
    \underbrace{\frac{\norm[1]{\vx_i- E_\vbeta ( D_\vtheta(\mA^\dagger \vy_i))}}{\norm[1]{\vx_i}}}_{\mathcal{L}_{\text{sup}}} +
    \underbrace{\frac{\norm[1]{\mA^\dagger \vy_i - 
    \mA 
    E_\vbeta ( D_\vvartheta (\mA^\dagger \vy_i +\vz_i)) }}{\norm[1]{\mA^\dagger \vy_i}}}_{\mathcal{L}_{\text{self}}}
    \Bigg),
    \label{eq:joint-loss-denois}
\end{align}
where $\vz \sim \mathcal{N}(0,\sigma^2)$ is random Gaussian noise that we generated for the self-supervised loss, 
and $\mA^\dagger \vy$ is the input of the U-Net (i.e., the least-squares reconstruction from measurement $\vy$). 
Similar to our domain adaptation method, at the inference, we optimize the self-supervised loss $\mathcal{L}_{\text{self}}(\vbeta,\vvartheta,\vy) $ with respect to the weights $\vbeta,\vvartheta$ for a given undersampled measurement $\vy$. The stopping iteration was set to $10$ heuristically (as opposed to our domain adaptation method where we used a self-validation-based automatic early stopping).

Figure~\ref{fig:semi-ttt} shows the results when training U-Net with loss function~\eqref{eq:joint-loss-denois} in comparison to a U-Net trained in a supervised manner. 
Because we observed no benefit from this type of TTT under a natural anatomical shift (i.e., $0\%$ of the gap is closed), in Figure~\ref{fig:semi-ttt}, we illustrate the result for artificial distribution shifts. We specifically evaluate this variant on the brain validation set of fastMRI under (1) no transformation (regular), (2) horizontal flipping (fliph), (3) vertical flipping (flipv), (4) 90-degree rotation (rot90), (5) Gaussian noise (noisy), (6) blur artifacts (blur), (7) elastic transformation (elastic), (8) spike artifacts (spike), and (9) ghosting artifacts (ghost). Interestingly, this variant of TTT is not even helpful with these artificial shifts according to SSIM scores reported in Figure~\ref{fig:semi-ttt}.

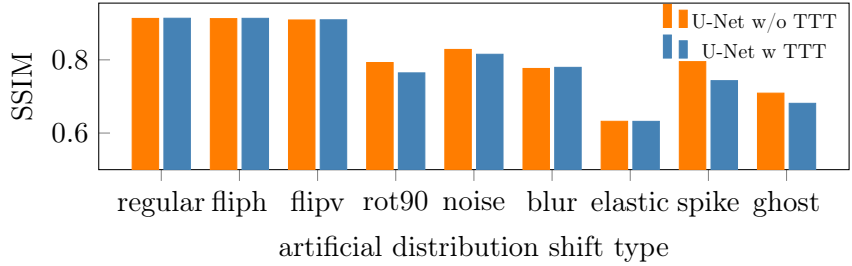
\begin{figure*}[t!]
\centering
  \begin{subfigure}[t]{1\textwidth}
  \begin{center}
    \begin{tikzpicture}
    \begin{groupplot}[
    y tick label style={/pgf/number format/.cd,fixed,precision=4},scaled y ticks = false,
    legend style={at={(1,1.03)} , draw=none, fill=none, nodes={scale=0.7},
    /tikz/every even column/.append style={column sep=-0.1cm}
     }, legend image code/.code={
        \draw [#1] (0cm,-0.1cm) rectangle (0.2cm,0.25cm); },
             group
             style={group size= 1 by 1, xlabels at=edge bottom,
             horizontal sep=1.5cm, vertical sep=0.4cm,
             }, 
             ybar,
             width=0.7\textwidth,height=0.23\textwidth,
             ymin=0.5,
             ]
    \nextgroupplot[ylabel=SSIM,xlabel={artificial distribution shift type},title={},xtick=data,xticklabels from table={files/semi_bar.txt}{x},ytick pos=left,xtick pos=bottom,]
    	\addplot +[mark=none,amber,fill=amber] table[x expr=\coordindex,y=y]{./files/semi_bar.txt};
    	\addlegendentry{U-Net w/o TTT}
    	\addplot +[mark=none,steelblue,fill=steelblue] table[x expr=\coordindex,y=y_semi]{./files/semi_bar.txt};
    	\addlegendentry{U-Net w TTT}
    \end{groupplot}
    \end{tikzpicture}
  \end{center}
  \end{subfigure}\par\medskip 
  %
  %
  %
  %
\caption{\textbf{A variant of test-time training (TTT) with a two-head U-Net inspired from image classification does not improve model robustness.} 
Including denoising as an auxiliary task during supervised training of U-Net even impairs model robustness for some artificial distribution shifts. 
Evaluation scores when each transformation on the $x$ axis is applied to the brain validation set.
}
\label{fig:semi-ttt}
\end{figure*}

\paragraph{Variant 2.} The second variant is a method we propose based on CycleGAN~\citep{zhu2017unpaired}, dubbed CycleU-Net, and works as follows.
Suppose we put two U-Nets $f_\vtheta$ and $g_\vbeta$ in tandem to form a larger model. We then train the resulting model by making two forward passes for each input pair $(\vy_i,\vx_i)$ at each epoch: (1) a self-supervised pass as $g_\vbeta(f_\theta(\mA^\dagger \vy_i))$, and (2) a supervised forward pass as $f_\vtheta(g_\vbeta(\vx_i))$. This is illustrated in Figure~\ref{fig:cycleunet}.

\begin{figure}[h!]

\centering
\begin{tikzpicture}
[node distance = 1cm, auto,font=\footnotesize,
every node/.style={node distance=2cm},
comment/.style={rectangle, inner sep= 3pt, text width=2cm, text badly centered,node distance=0.15cm, font=\footnotesize}, 
force1/.style={rectangle, draw=steelblue, fill=steelblue!10, inner sep=-5pt, text width=2cm, text badly centered, minimum height=1.2cm, font=\bfseries\footnotesize},
force2/.style={rectangle, draw=amber, fill=amber!10, inner sep=-5pt, text width=2cm, text badly centered, minimum height=1.2cm, font=\bfseries\footnotesize}] 

\node [force1] (unet1) {$f_\vtheta$};
\node [force2, right=1cm of unet1] (unet2) {$g_\vbeta$};
\node [draw=none, left=2cm of unet1] (ls1) at (0,0) {\includegraphics[scale=0.3,]{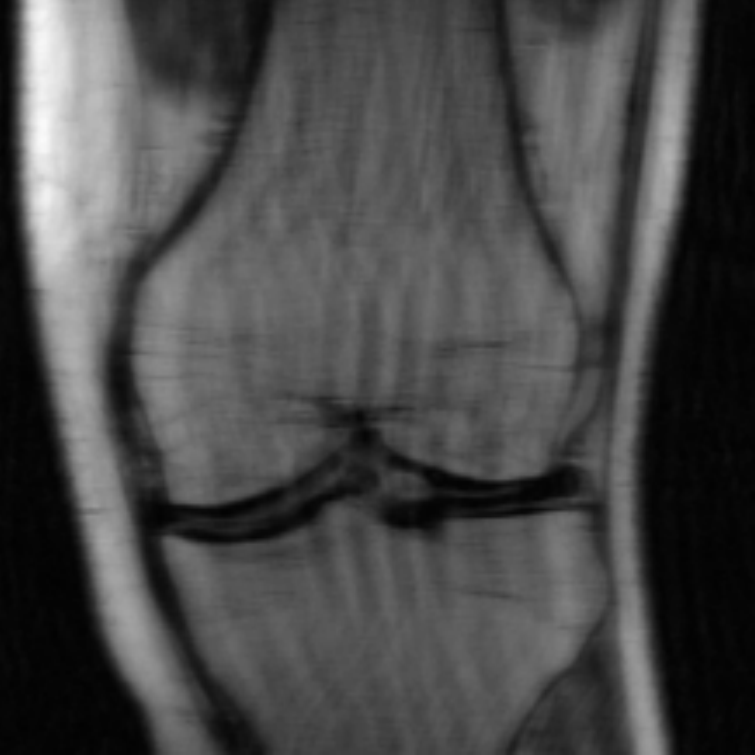}};
\node [draw=none, right=4.7cm of unet2] (ls2) at (0,0) {\includegraphics[scale=0.3,]{files/ls.pdf}};

\node [force2, below=2cm of unet1] (unet3) {$g_\vbeta$};
\node [force1, right=1cm of unet3] (unet4) {$f_\vtheta$};
\node [draw=none, left=2cm of unet3] (gt1) at (0,-3.2) {\includegraphics[scale=0.3,]{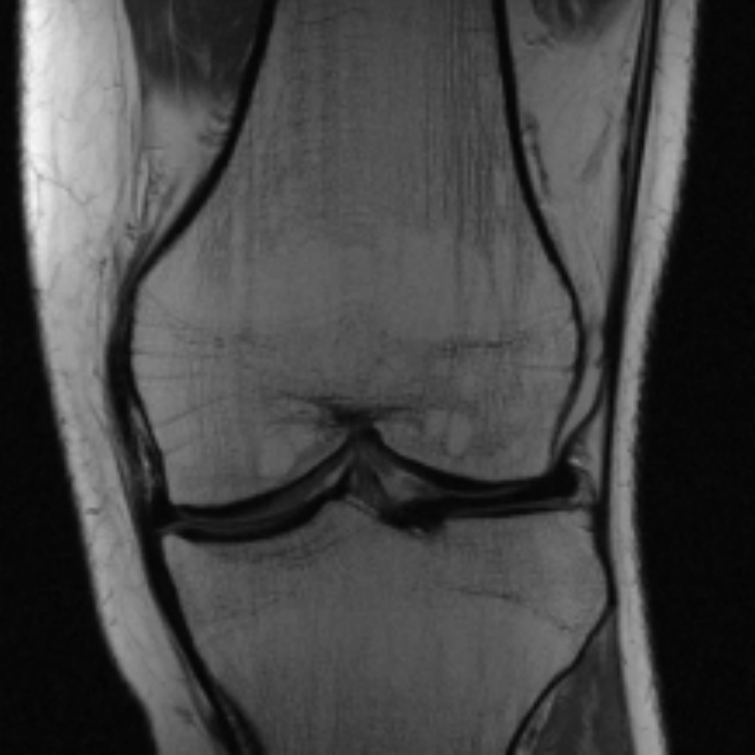}};
\node [draw=none, right=4.7cm of unet4] (gt2) at (0,-3.2) {\includegraphics[scale=0.3,]{files/gt.pdf}};

\node [comment, above=0.01 of ls1] (ls-lab1)
{zero-filled};
\node [comment, left=-0.8 of ls1] (col1)
{\large :};
\node [comment, left=0.3 of ls1] (self)
{self-supervised \\ forward pass};
\node [comment, above=0.01 of ls2] (ls-lab2)
{zero-filled};

\node [comment, above=0.01 of gt1] (gt-lab1)
{ground-truth};
\node [comment, left=-0.8 of gt1] (col2)
{\large :};
\node [comment, left=0.3 of gt1] (sup)
{supervised \\ forward pass};
\node [comment, above=0.01 of gt2] (gt-lab2)
{ground-truth};

\path[->,thick] 
(ls1) edge (unet1)
(unet1) edge (unet2)
(unet2) edge (ls2);
\path[->,thick] 
(gt1) edge (unet3)
(unet3) edge (unet4)
(unet4) edge (gt2);

\end{tikzpicture} 
\caption{\textbf{Forward passes in CycleU-Net.} $f_\vartheta$ and $g_\vbeta$ are two U-Nets with the same number of parameters. Training based on the two forward passes enables test-time training (TTT) w.r.t. the self-supervised pass. Only $f_\vtheta$ is needed for inference after TTT.}
\label{fig:cycleunet}
\end{figure}

By defining those two forward passes, we can build the training loss function as follows:
\begin{align*}
    \mc L(\vbeta,\vtheta) &=
    \frac{1}{n}
    \sum_{i=1}^n
    \Bigg(
    \norm[1]{\vx_i - f_\vtheta(g_\vbeta(\vx_i)))} +
    \norm[1]{\mA^\dagger \vy_i - g_\vbeta(f_\vtheta(\mA^\dagger \vy_i)))}
    \\&\hspace{2cm}+ 
    \norm[1]{\vx_i - f_\vtheta(\mA^\dagger \vy_i)} +
    \norm[1]{\mA^\dagger \vy_i - g_\vbeta(\vx_i)}
    \Bigg).
    \label{eq:joint-loss-denois}
\end{align*}
Here, the first two terms enforce input-output equality for each of the two forward passes. The third term ensures that $f_\vtheta$ learns a mapping from the under-sampled to the ground-truth domain (likewise, the fourth term ensures that $g_\vbeta$ learns a mapping from the ground-truth to the under-sampled domain). Note that without the last two terms, there is no guarantee that $f_\vtheta$ reconstructs the ground-truth image from the under-sampled measurement.

At inference, we perform TTT w.r.t $\norm[1]{\mA^\dagger \vy_i - g_\vbeta(f_\vtheta(\mA^\dagger \vy_i)))}$ (both $\vtheta$ and $\vbeta$ are optimized) which is fully self-supervised, then detach $f_\vtheta$ from the architecture and use it for reconstruction as $f_\vtheta(\mA^\dagger \vy)$.

Table~\ref{tab:gap-cyc} shows the performance of this approach for anatomy shift (the training and test data are the same as Section~\ref{sec:exp}). As shown, the fraction of gap closed by performing TTT on CycleU-Net is 29.4\% which demonstrates that this approach is not effective in closing the gap.

\begin{table}[t!]
\centering
\begin{adjustbox}{width=0.45\textwidth}
\begin{tabular}{l|c}
\toprule 
\multicolumn{1}{c}{setup} & \multicolumn{1}{c}{P: knee} \\
\multicolumn{1}{c}{} & \multicolumn{1}{c}{Q: brain} \\
\hline
    train on Q test on Q                & 0.9187 \\ 
    train on P test on Q                & 0.8521 \\ \hline 
    distribution shift performance gap  & 0.0666 \\ 
\thickhline
    train on Q test on Q + TTT          & 0.9212 \\
    train on P test on Q + TTT          & 0.8742 \\ \hline 
    distribution shift performance gap  & 0.0470 \\
    \thickhline
    fraction of gap closed by TTT       & 29.4\%        \\
\bottomrule
\end{tabular}
\end{adjustbox}
\caption{\textbf{A variant of test-time training (TTT) with a CycleU-Net inspired from CycleGAN improves model robustness only slightly.} The first three rows are SSIM scores for U-Net when trained in a supervised manner. The second three rows are SSIM scores for CycleU-Net when TTT is applied at the inference. This variant closes the gap by $29.4\%$ but is outperformed by our original method which closes the gap by $98.6\%$ which is discussed in the main body.
}
\vspace{-10pt}
\label{tab:gap-cyc}
\end{table}
\section{Relation to imaging with un-trained neural networks}
\label{sec:priors}

Our domain adaptation method consists of training a network with supervised and self-supervised loss and at inference, training again on the self-supervised loss with early stopping.


The inference step is very similar to how an un-trained neural network is used for image recovery. To see this, reconstruction of a signal from an observation with an un-trained network works as follows.
Let $f_\vtheta$ be a convolutional network that is initialized randomly, and optimized on the loss
\begin{align*}
\mc L(\vtheta) =  \norm[1]{\vy - \mA f_\vtheta(\vz)},
\end{align*}
with gradient descent and early-stopping the iterates for regularization. Here, $\vz$ is an input that is typically random, and has observed to be relatively irrelevant. 
\citet{ulyanov2018deep} demonstrated that this method works well for denoising and super-resolution, and \citet{heckel2019deep,arora2020untrained,darestani2020can} have shown that un-trained convolutional networks work well for denoising and accelerated MRI. 
Reconstruction with an un-trained network is very similar to our TTT step, with the difference that the un-trained networks start from a random initialization of the weights (which in this case, for the DIP-based model introduced in \citep{darestani2020can}, they achieve $0.9046$ SSIM on brain images for the anatomy shift setup we consider in Table~\ref{tab:results}).


We could also just initialize the weights of an un-trained network by pre-training the network on a dataset, and we might expect that this improves performance. 
Viewed from that angle, our method might look like an un-trained network in disguise. 

In this section, we argue that our method is not an un-trained network in disguise by demonstrating that TTT improves over the image prior learned by pre-training. 
We perform the following experiment. We consider the anatomy shift and train a U-Net $f_\text{sup}$ in a supervised manner and another U-net $f_\text{joint}$ using our joint loss function~\eqref{eq:joint-loss} that incorporates self-supervision. 
We then perform TTT by applying $f_\text{sup}$ and $f_\text{joint}$ to a brain test sample. 

Figure~\ref{fig:priors} depicts the results. The first row shows the self-validation error whose increase determines the early-stopping point. The second row depicts the SSIM with respect to the unknown ground truth image during TTT. 
Observe that when TTT is applied to $f_\text{joint}$, SSIM improves during TTT. Thus TTT improves over the learned prior from knees to achieve a good reconstruction accuracy on the brain test sample. 

Contrary, when TTT is applied to $f_\text{sup}$, SSIM first decreases dramatically and then starts to rise again. This suggests that TTT applied to a fully self-supervised loss ignores the learned prior and uses the model as an un-trained network.

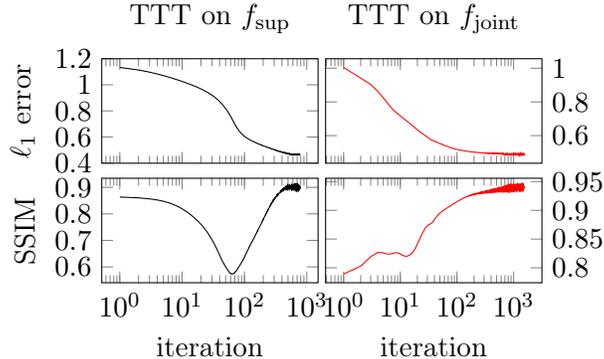
\begin{figure}[t!]
\centering
  \begin{center}
    \begin{tikzpicture}
    \begin{groupplot}[
    y tick label style={/pgf/number format/.cd,fixed,precision=4},scaled y ticks = false,
    legend style={at={(1,0.6)} , 
    /tikz/every even column/.append style={column sep=-0.1cm}
     },
             group
             style={group size= 2 by 2, xlabels at=edge bottom,
             xticklabels at=edge bottom,
             horizontal sep=0.1cm, vertical sep=0.2cm,
             }, 
             width=0.27\textwidth,height=0.18\textwidth,
             ]
    \nextgroupplot[ylabel={$\ell_1$ error},xlabel={},title={TTT on $f_\text{sup}$},xmode=log,]
    	\addplot +[mark=none,black] table[x=x,y=verr]{./files/ttt_on_priors/sup.txt};
    \nextgroupplot[ylabel={},xlabel={},title={TTT on $f_\text{joint}$},xmode=log,yticklabel pos=right,]
    	\addplot +[mark=none,red] table[x=x,y=verr]{./files/ttt_on_priors/self.txt};
    \nextgroupplot[ylabel={SSIM},xlabel={iteration},title={},xmode=log,]
    	\addplot +[mark=none,black] table[x=x,y=score]{./files/ttt_on_priors/sup.txt};
    \nextgroupplot[ylabel={},xlabel={iteration},title={},xmode=log,yticklabel pos= right,]
    	\addplot +[mark=none,red] table[x=x,y=score]{./files/ttt_on_priors/self.txt};
    \end{groupplot}
    \end{tikzpicture}
  \end{center}
\caption{Test-time training (TTT) improves the reconstruction accuracy of a pre-trained model when self-supervision is included during training (right column), but treats a pre-trained model as an un-trained network when pre-training is fully supervised (left column). The first row shows the self-validation error which is used to stop TTT early. The second row shows SSIM w.r.t. the ground truth image during TTT.}
\label{fig:priors}
\end{figure}

\section{Proof of claims in Section~\ref{sec:theory}: Test-time training can provably adapt to a distribution shift}
\label{sec:proof-analytical-example}

\begin{proposition}
The supervised loss 
\begin{align*}
    R_P(\alpha,\mV) = \EX[P]{ \norm[2]{\vx -  \alpha \mV \transp{\mV}  \vy}^2 }
\end{align*}
is minimized by $\alpha = 1/(1+\sigma^2)$ and $\mV = \mU$.
\end{proposition}
\begin{proof}
Let $\mW = \alpha \mV \transp{\mV}$ and consider the following related convex optimization problem
\begin{align*}
    \min_{\mW} \EX[P]{ \norm[2]{\vx -  \mW \vy}^2 }.
\end{align*}
Note that here we optimizer over a larger space, as we do not constrain $\mW$ to be symmetric. We show that a solution to the optimization problem is $\mW = \frac{1}{1+\sigma^2} \mU \transp{\mU}$, therefore $(\alpha,\mU)$ is a minimizer of $R_P(\alpha,\mV)$, which proves the claim. 

The gradient of the objective function above is
\begin{align*}
    \nabla_{\mW} \EX[P]{ \norm[2]{\vx -  \mW  \vy}^2 } 
    &= \nabla_{\mW} \trace\left( \EX[P]{\vx \transp{\vx}} - 2 \transp{\mW} \EX[P]{\vx \transp{\vy}} + \alpha^2 \transp{\mW} \mW \EX[P]{\vy \transp{\vy}} \right) \\
    &= 2 \mW \EX[P]{\vy \transp{\vy}} - 2 \EX[P]{\vx \transp{\vy}}.
\end{align*}
Setting the gradient to zero, the minimizer satisfied
\begin{align*}
    \mW = \EX[P]{\vx \transp{\vy}} \left( \EX[P]{\vy \transp{\vy}} \right)^{-1}.
\end{align*}
Since $\vy = \vx + \vz$ and $\vx = \mU \vc$, where $\vc \sim \mc N(0,\mI)$ and $\vz \sim \mc N(0, \mI \sigma^2)$ are independent, it holds that
\begin{align*}
    \EX[P]{\vx \transp{\vy}} &= \EX[P]{\vx \transp{\vx}} = \mU \transp{\mU},
   \end{align*}
and
\begin{align*}
    \EX[P]{\vy \transp{\vy}} &= \EX[P]{\vx \transp{\vx} + \vz \transp{\vz}} = \mU \transp{\mU} + \sigma^2 \mI.
\end{align*}
Plugging in the two expressions into the expression for $\mW$  we obtain
\begin{align*}
    \mW = \frac{1}{1+\sigma^2} \mU \transp{\mU},
\end{align*}
as desired.
\end{proof}

\begin{proposition}
For fixed $\mU$, the expectation of the self-supervised loss 
\begin{align*}
    \EX[Q]{L_{SS}(\alpha,\mU,\vy)} = \EX[Q]{\norm[2]{\vy - \alpha \mU \transp{\mU}\vy }^2 + \frac{2\alpha d }{n-d} \norm[2]{(\mI - \mU \transp{\mU}) \vy}^2}
\end{align*}
is minimized by $\alpha = 1/(1+\varsigma^2)$.
\end{proposition}
\begin{proof}
Note that
\begin{align*}
    \EX[Q]{\norm[2]{\vy - \alpha \mU \transp{\mU}\vy }^2}
    &= \trace\left( \EX[Q]{\vy \transp{\vy}} - 2\alpha \transp{\mU} \mU \EX[Q]{\vy \transp{\vy}} + \alpha^2 \transp{\mU} \mU \EX[Q]{\vy \transp{\vy}} \right) \\
    &= \trace\left( \left( \mI + (\alpha^2 - 2\alpha) \transp{\mU} \mU \right) \EX[Q]{\vy \transp{\vy}} \right) \\
    &= \trace\left( \left( \mI + (\alpha^2 - 2\alpha) \transp{\mU} \mU \right) \left( \mU \transp{\mU} + \varsigma^2 \mI \right) \right) \\
    &= \trace\left( \varsigma^2 \mI + ((1 - \alpha)^2 + (\alpha^2 - 2\alpha) \varsigma^2) \mU \transp{\mU} \right) \\
    &= \varsigma^2 n + (1 - \alpha)^2 d + (\alpha^2 - 2\alpha) \varsigma^2 d.
\end{align*}
It follows that, for $\alpha=1$, 
\begin{align*}
    \EX[Q]{\norm[2]{\vy - \mU \transp{\mU}\vy }^2}
    &= \varsigma^2 n  - \varsigma^2 d.
\end{align*}
Hence,
\begin{align*}
    \EX[Q]{L_{SS}(\alpha,\mU,\vy)} 
    &= \varsigma^2 n + (1 - \alpha)^2 d + (\alpha^2 - 2\alpha) \varsigma^2 d + \frac{2\alpha d }{n-d} (\varsigma^2 n - \varsigma^2 d) \\
    &= \varsigma^2 n + (1 - \alpha)^2 d + \alpha^2 \varsigma^2 d,
\end{align*}
whose minimum is achieved at $\alpha = 1/(1+\varsigma^2)$. To see this, take the derivative with respect to $\alpha$, set it to zero, and solve for $\alpha$. 
\end{proof}

\section{Test-time training for non-convolutional architectures}\label{sec:vit}

The two networks that we studied throughout are based on the U-Net, a convolutional neural network. 
There are, however, other non-convolutions neural network architectures that perform well for image reconstruction problems. 
In this section we explore how our test-time training (TTT) approach performs with non-convolutional networks.

We consider the Vision Transformer (ViT)~\citep{dosovitskiy2020image} that works very well for signal reconstruction problems. ViT has been tailored to accelerated MRI reconstruction as well~\citep{feng2021accelerated,feng2021task,korkmaz2021deep,korkmaz2022unsupervised}, and has been shown to be computationally faster than U-Net and also slightly more robust than U-Net against anatomy shift~\citep{lin2021vision}.

We repeated the same experiment we performed for U-Net under the anatomy shift. The results, reported in Table~\ref{tab:gap-vit} show that TTT for a ViT is effective, but not as effective as for the U-Net (specifically, the fraction of the gap closed by TTT is only $84.5\%$ for ViT, whereas the gap closed by TTT is $98.6\%$ for U-Net). 
This is also reflected in the visual quality of the images, as illustrated in Figure~\ref{fig:anat-shift-vit}. In Figure~\ref{fig:anat-shift-vit} we see that TTT for a ViT gives reconstruction artifacts. 

The self-supervised loss function used in our TTT approach works better for U-Net compared to ViT, and hence our TTT is more effective for U-Net. To see this, we perform the following experiment. We train U-Net and ViT on the knee training set of fastMRI in a fully self-supervised manner, i.e., we minimize the training loss
\[
\mc L(\vtheta) = \sum_{i=1}^n \frac{\norm[1]{\vy_i - \mA
    f_\vtheta(\mA^\dagger \vy_i)}}{\norm[1]{\vy_i}},
\]
on a set of training measurements $\vy_1,\ldots, \vy_n$, for both U-net and ViT. 
During training, we monitor true accuracy, by comparing the reconstructions generated by U-Net and ViT to the (unknown) ground-truth images $\vx_1,\ldots, \vx_n$ associated with the measurements $\vy_1,\ldots, \vy_n$. 
Figure~\ref{fig:indbias} shows that towards convergence, there is a constant gap between the true accuracy achieved by U-Net and ViT. This demonstrates that our self-supervised loss works better with U-Net than ViT in terms of the quality of the learned prior for the images.

\begin{table}[h!]
\centering
\begin{adjustbox}{width=0.45\textwidth}
\begin{tabular}{l|c}
\toprule 
\multicolumn{1}{c}{setup} & \multicolumn{1}{c}{P: knee} \\
\multicolumn{1}{c}{} & \multicolumn{1}{c}{Q: brain} \\
\hline
    train on Q test on Q                & 0.9041 \\ 
    train on P test on Q                & 0.8429 \\ \hline 
    distribution shift performance gap  & 0.0612 \\ 
\thickhline
    train on Q test on Q + TTT          & 0.8947 \\
    train on P test on Q + TTT          & 0.8859 \\ \hline 
    distribution shift performance gap  & 0.0095 \\
    \thickhline
    fraction of gap closed by TTT       & 84.5\%        \\
\bottomrule
\end{tabular}
\end{adjustbox}
\caption{\textbf{For ViT, using self-supervision with test-time training (TTT) closes 84\% of the distribution shift performance gap for anatomy shift.} The first three rows are SSIM scores for ViT when trained in a supervised manner. The second three rows are SSIM scores for ViT when self-supervision is included during training and then TTT is applied at the inference.
}
\vspace{-10pt}
\label{tab:gap-vit}
\end{table}

\begin{table*}[h!]
\setlength{\tabcolsep}{1pt}
\centering
\begin{tabular}{ccccc}
  & \begin{tabular}{@{}c@{}} trained on knee\\  (supervised) \\[5pt] \footnotesize SSIM: 0.8480 \end{tabular}  & \begin{tabular}{@{}c@{}}  trained on knee + TTT \\  (self-supervision included) \\[5pt] \footnotesize SSIM: 0.8855 \end{tabular} & \begin{tabular}{@{}c@{}} trained on brain \\ (supervised) \\[5pt] \footnotesize SSIM: 0.9040 \end{tabular} & \begin{tabular}{@{}c@{}}  ground truth \\[25pt] \end{tabular} \\
  \rule{0pt}{3ex}
   &
  \begin{adjustbox}{valign=m}
    \begin{tikzpicture}[boximg]
      \node[anchor=south] (img) {\includegraphics[width=0.23\textwidth]{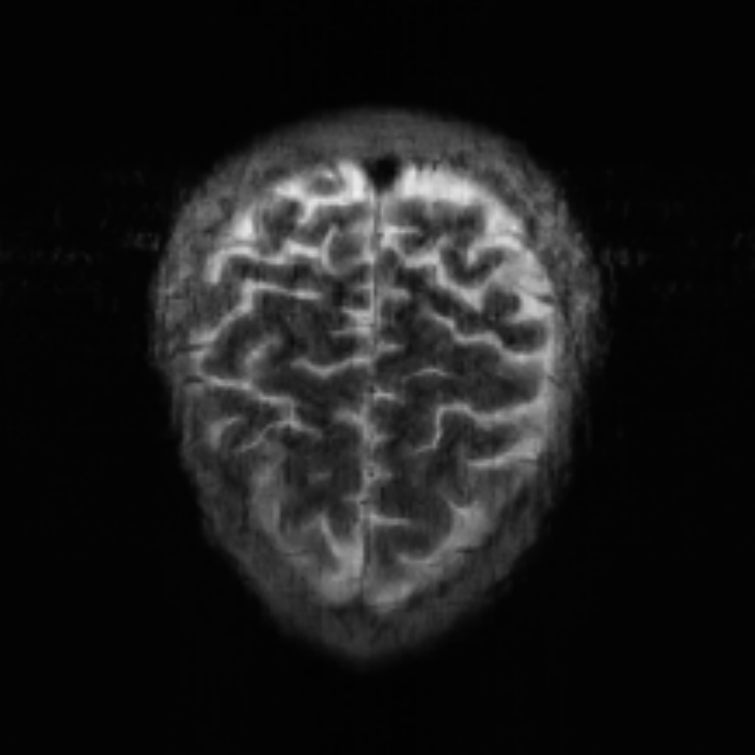}}; 
      \begin{scope}[x=10,y=10]
        \node[draw,minimum height=0.55cm,minimum width=0.65cm] (B1) at (-1.2,6.1) {};  
        \node (img1) at (-3.7,1.4) {\includegraphics[width=0.07\textwidth,trim={2.3cm 3.8cm 4cm 2.7cm},  clip]{files/anatomy_shift/vitknee_sup_on_brain.pdf}}; 
        \draw (img1.south west) rectangle (img1.north east);
        \draw (B1) -- (img1);
      \end{scope}
    \end{tikzpicture} \end{adjustbox} &
  \begin{adjustbox}{valign=m}
    \begin{tikzpicture}[boximg]
      \node[anchor=south] (img) {\includegraphics[width=0.23\textwidth]{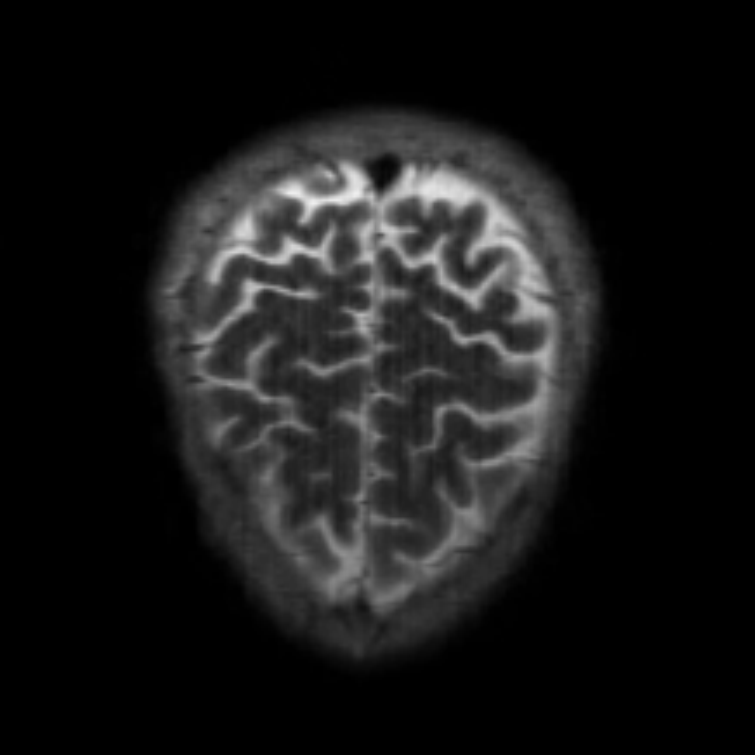}};
      \begin{scope}[x=10,y=10]
        \node[draw,minimum height=0.55cm,minimum width=0.65cm] (B1) at (-1.2,6.1) {};
        \node (img1) at (-3.7,1.4) {\includegraphics[width=0.07\textwidth,trim={2.3cm 3.8cm 4cm 2.7cm}, clip]{files/anatomy_shift/vitttt_on_brain.pdf}};
        \draw (img1.south west) rectangle (img1.north east);
        \draw (B1) -- (img1);
      \end{scope}
    \end{tikzpicture} \end{adjustbox} &
  \begin{adjustbox}{valign=m}
    \begin{tikzpicture}[boximg]
      \node[anchor=south] (img) {\includegraphics[width=0.23\textwidth]{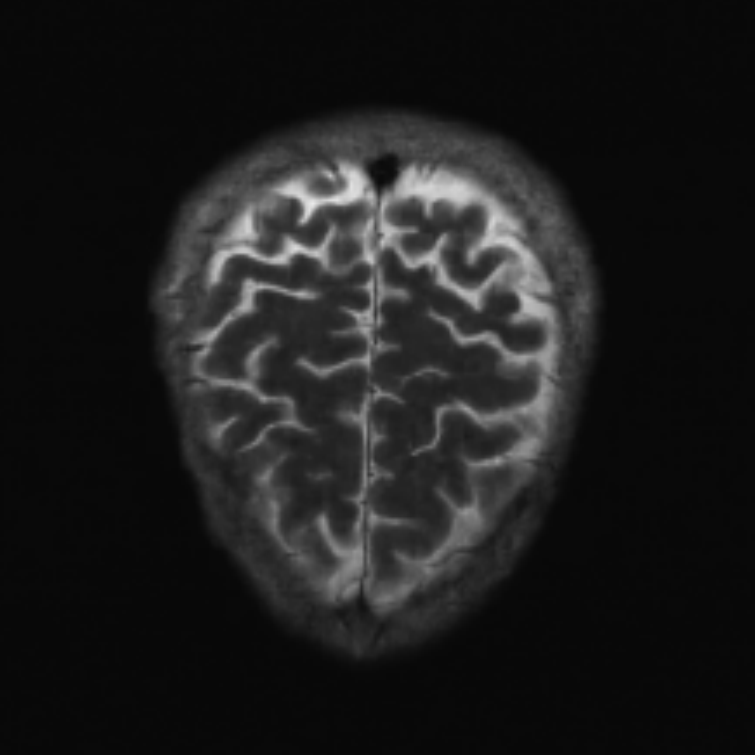}};
      \begin{scope}[x=10,y=10]
        \node[draw,minimum height=0.55cm,minimum width=0.65cm] (B1) at (-1.2,6.1) {};
        \node (img1) at (-3.7,1.4) {\includegraphics[width=0.07\textwidth,trim={2.3cm 3.8cm 4cm 2.7cm}, clip]{files/anatomy_shift/vitbrain_sup_on_brain.pdf}};
        \draw (img1.south west) rectangle (img1.north east);
        \draw (B1) -- (img1);
      \end{scope}
    \end{tikzpicture} \end{adjustbox} &
  \begin{adjustbox}{valign=m}
    \begin{tikzpicture}[boximg]
      \node[anchor=south] (img) {\includegraphics[width=0.23\textwidth]{files/anatomy_shift/orig.pdf}};
      \begin{scope}[x=10,y=10]
        \node[draw,minimum height=0.55cm,minimum width=0.65cm] (B1) at (-1.2,6.1) {};
        \node (img1) at (-3.7,1.4) {\includegraphics[width=0.07\textwidth,trim={2.3cm 3.8cm 4cm 2.7cm}, clip]{files/anatomy_shift/orig.pdf}};
        \draw (img1.south west) rectangle (img1.north east);
        \draw (B1) -- (img1);
      \end{scope}
    \end{tikzpicture} \end{adjustbox} \\
\end{tabular}
\captionsetup{skip=10pt}
\captionof{figure}{Including self-supervision while training DL models combined with TTT improves model robustness to natural anatomy shifts. The sample belongs to the fastMRI brain validation dataset.}
\label{fig:anat-shift-vit}
\end{table*}

\begin{figure}[t!]
\centering
  \begin{center}
    \begin{tikzpicture}
    \begin{groupplot}[
    y tick label style={/pgf/number format/.cd,fixed,precision=4},scaled y ticks = false,
    legend style={at={(1,0.7)}, draw={none}, fill = none,text opacity=1, 
    /tikz/every even column/.append style={column sep=-0.1cm}
     },
             group
             style={group size= 1 by 1, xlabels at=edge bottom,
             xticklabels at=edge bottom,
             horizontal sep=0.1cm, vertical sep=0.2cm,
             }, 
             width=0.4\textwidth,height=0.3\textwidth,
             ]
    \nextgroupplot[ylabel={\parbox{2.9cm}{ \centering true training \\ accuracy}},ylabel style={rotate=-90,xshift=2ex,},xlabel={epoch},title={training dynamics},grid=both]
    	\addplot +[mark=none,amber,thick,] table[x=x,y=unet]{./files/anatomy_shift/indbias.txt};
    	\addplot +[mark=none,steelblue,thick] table[x=x,y=vit]{./files/anatomy_shift/indbias.txt};
    	\legend{U-Net,ViT}
    \end{groupplot}
    \end{tikzpicture}
  \end{center}
\caption{The inductive bias of self-supervised training of U-Net results in a higher true accuracy at convergence than ViT. True accuracy is monitored using the ground-truth data during the self-supervised training.}
\label{fig:indbias}
\end{figure}
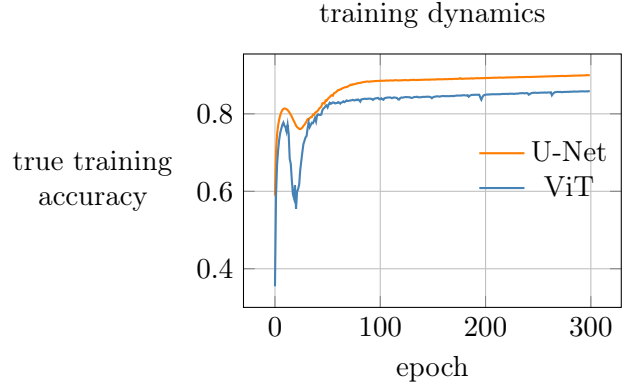

\section{More reconstruction examples}

The results of Table~\ref{tab:results} demonstrate that our domain adaptation method closes the distribution shift performance gap for anatomy, dataset, modality, and acceleration shifts by about 90\%. Figure~\ref{fig:gap} in the main body shows example images reconstructed with U-Net to demonstrate that the perceptual quality is improved after applying our domain adaptation method.


In this section, we provide more detailed illustrations for both U-Net and VarNet under each distribution shift. Figure~\ref{fig:anat-shift}, Figure~\ref{fig:data-shift}, Figure~\ref{fig:acq-shift}, and Figure~\ref{fig:acc-shift} provide reconstruction examples for anatomy, dataset, modality, and acceleration shifts before and after our domain adaptation method. As shown in the figures, the perceptual quality of the reconstructions improve noticeably with test-time training.

\begin{table*}[t!]
\setlength{\tabcolsep}{1pt}
\centering
\begin{adjustbox}{width=0.85\textwidth}
\begin{tabular}{ccccc}
  & \begin{tabular}{@{}c@{}}  trained on knee\\  (supervised) \\[5pt] \footnotesize SSIM: 0.8562 \end{tabular}  & \begin{tabular}{@{}c@{}} \footnotesize trained on knee + TTT \\  (self-supervision included) \\[5pt] \footnotesize SSIM: 0.9215 \end{tabular} & \begin{tabular}{@{}c@{}}  trained on brain \\  (supervised) \\[5pt] \footnotesize SSIM: 0.9185 \end{tabular} & \begin{tabular}{@{}c@{}}  ground truth\\[25pt] \end{tabular} \\
  \rule{0pt}{3ex}
  U-Net &
  \begin{adjustbox}{valign=m}
    \begin{tikzpicture}[boximg]
      \node[anchor=south] (img) {\includegraphics[width=0.25\textwidth]{files/anatomy_shift/knee_sup_on_brain.pdf}};
      \begin{scope}[x=10,y=10]
        \node[draw,minimum height=0.55cm,minimum width=0.65cm] (B1) at (-1.2,6.6) {}; 
        \node (img1) at (-4.2,1.4) {\includegraphics[width=0.07\textwidth,trim={2.3cm 3.8cm 4cm 2.7cm}, clip]{files/anatomy_shift/knee_sup_on_brain.pdf}}; 
        \draw (img1.south west) rectangle (img1.north east);
        \draw (B1) -- (img1);
      \end{scope}
    \end{tikzpicture} \end{adjustbox} &
  \begin{adjustbox}{valign=m}
    \begin{tikzpicture}[boximg]
      \node[anchor=south] (img) {\includegraphics[width=0.25\textwidth]{files/anatomy_shift/ttt_on_brain.pdf}};
      \begin{scope}[x=10,y=10]
        \node[draw,minimum height=0.55cm,minimum width=0.65cm] (B1) at (-1.2,6.6) {}; 
        \node (img1) at (-4.2,1.4) {\includegraphics[width=0.07\textwidth,trim={2.3cm 3.8cm 4cm 2.7cm}, clip]{files/anatomy_shift/ttt_on_brain.pdf}};
        \draw (img1.south west) rectangle (img1.north east);
        \draw (B1) -- (img1);
      \end{scope}
    \end{tikzpicture} \end{adjustbox} &
  \begin{adjustbox}{valign=m}
    \begin{tikzpicture}[boximg]
      \node[anchor=south] (img) {\includegraphics[width=0.25\textwidth]{files/anatomy_shift/brain_sup_on_brain.pdf}};
      \begin{scope}[x=10,y=10]
        \node[draw,minimum height=0.55cm,minimum width=0.65cm] (B1) at (-1.2,6.6) {}; 
        \node (img1) at (-4.2,1.4) {\includegraphics[width=0.07\textwidth,trim={2.3cm 3.8cm 4cm 2.7cm}, clip]{files/anatomy_shift/brain_sup_on_brain.pdf}};
        \draw (img1.south west) rectangle (img1.north east);
        \draw (B1) -- (img1);
      \end{scope}
    \end{tikzpicture} \end{adjustbox} &
  \begin{adjustbox}{valign=m}
    \begin{tikzpicture}[boximg]
      \node[anchor=south] (img) {\includegraphics[width=0.25\textwidth]{files/anatomy_shift/orig.pdf}};
      \begin{scope}[x=10,y=10]
        \node[draw,minimum height=0.55cm,minimum width=0.65cm] (B1) at (-1.2,6.6) {}; 
        \node (img1) at (-4.2,1.4) {\includegraphics[width=0.07\textwidth,trim={2.3cm 3.8cm 4cm 2.7cm}, clip]{files/anatomy_shift/orig.pdf}};
        \draw (img1.south west) rectangle (img1.north east);
        \draw (B1) -- (img1);
      \end{scope}
    \end{tikzpicture} \end{adjustbox} \\
  \rule{0pt}{3ex}
  & \footnotesize SSIM: 0.8782 & \footnotesize SSIM: 0.9292 & \footnotesize SSIM: 0.9316 & \\
  \rule{0pt}{3ex}
  VarNet &
  \begin{adjustbox}{valign=m}
    \begin{tikzpicture}[boximg]
      \node[anchor=south] (img) {\includegraphics[width=0.25\textwidth]{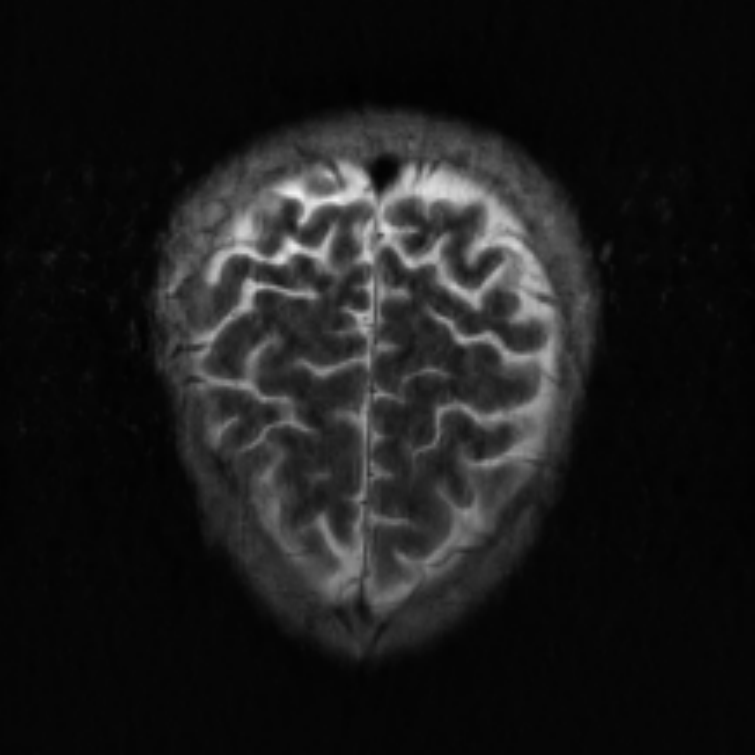}};
      \begin{scope}[x=10,y=10]
        \node[draw,minimum height=0.55cm,minimum width=0.65cm] (B1) at (-1.2,6.6) {};
        \node (img1) at (-4.2,1.4) {\includegraphics[width=0.07\textwidth,trim={2.3cm 3.8cm 4cm 2.7cm}, clip]{files/anatomy_shift/vknee_sup_on_brain.pdf}};
        \draw (img1.south west) rectangle (img1.north east);
        \draw (B1) -- (img1);
      \end{scope}
    \end{tikzpicture} \end{adjustbox} &
  \begin{adjustbox}{valign=m}
    \begin{tikzpicture}[boximg]
      \node[anchor=south] (img) {\includegraphics[width=0.25\textwidth]{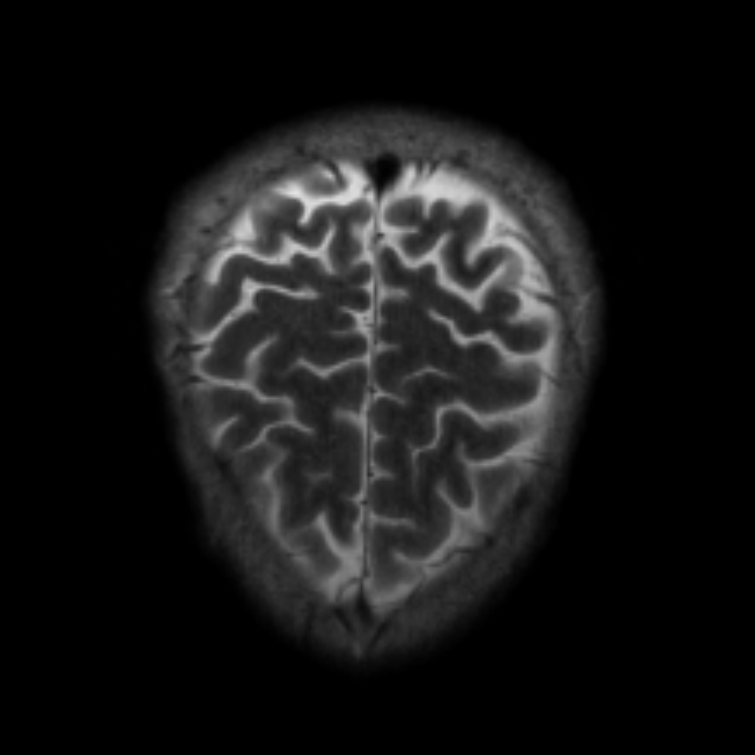}};
      \begin{scope}[x=10,y=10]
        \node[draw,minimum height=0.55cm,minimum width=0.65cm] (B1) at (-1.2,6.6) {};
        \node (img1) at (-4.2,1.4) {\includegraphics[width=0.07\textwidth,trim={2.3cm 3.8cm 4cm 2.7cm}, clip]{files/anatomy_shift/vttt_on_brain.pdf}};
        \draw (img1.south west) rectangle (img1.north east);
        \draw (B1) -- (img1);
      \end{scope}
    \end{tikzpicture} \end{adjustbox} &
  \begin{adjustbox}{valign=m}
    \begin{tikzpicture}[boximg]
      \node[anchor=south] (img) {\includegraphics[width=0.25\textwidth]{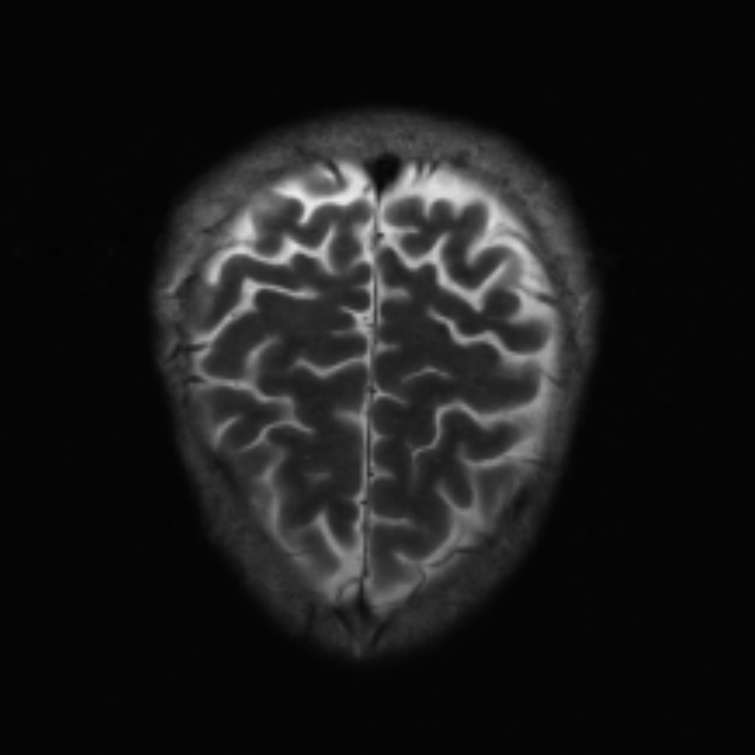}};
      \begin{scope}[x=10,y=10]
        \node[draw,minimum height=0.55cm,minimum width=0.65cm] (B1) at (-1.2,6.6) {};
        \node (img1) at (-4.2,1.4) {\includegraphics[width=0.07\textwidth,trim={2.3cm 3.8cm 4cm 2.7cm}, clip]{files/anatomy_shift/vbrain_sup_on_brain.pdf}};
        \draw (img1.south west) rectangle (img1.north east);
        \draw (B1) -- (img1);
      \end{scope}
    \end{tikzpicture} \end{adjustbox} &
  \begin{adjustbox}{valign=m}
    \begin{tikzpicture}[boximg]
      \node[anchor=south] (img) {\includegraphics[width=0.25\textwidth]{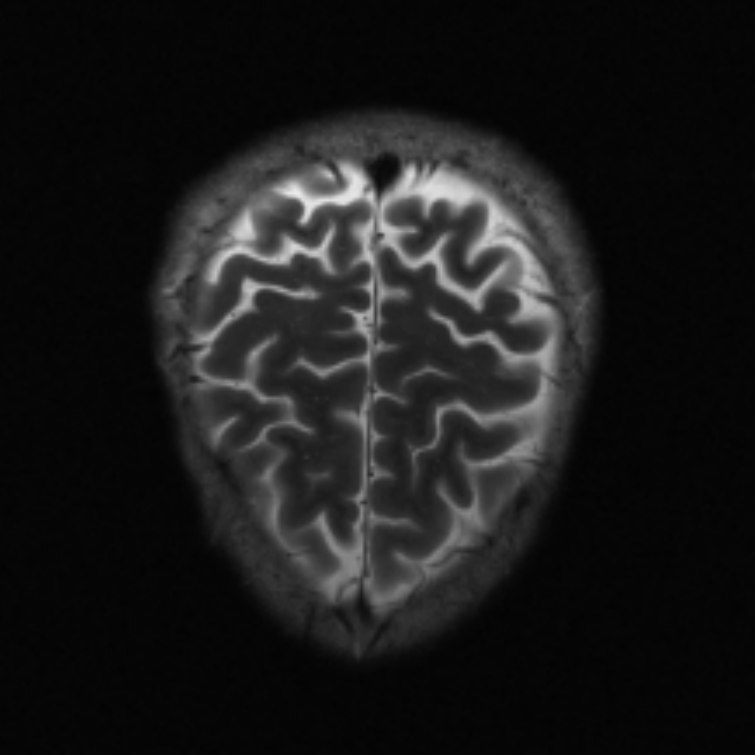}};
      \begin{scope}[x=10,y=10]
        \node[draw,minimum height=0.55cm,minimum width=0.65cm] (B1) at (-1.2,6.6) {};
        \node (img1) at (-4.2,1.4) {\includegraphics[width=0.07\textwidth,trim={2.3cm 3.8cm 4cm 2.7cm}, clip]{files/anatomy_shift/vorig.pdf}};
        \draw (img1.south west) rectangle (img1.north east);
        \draw (B1) -- (img1);
      \end{scope}
    \end{tikzpicture} \end{adjustbox}
\end{tabular}
\end{adjustbox}
\captionsetup{skip=10pt}
\captionof{figure}{Including self-supervision while training DL models combined with TTT improves model robustness to natural anatomy shifts. The sample belongs to the fastMRI brain validation dataset.}
\label{fig:anat-shift}
\end{table*}

\begin{table*}[h!]
\setlength{\tabcolsep}{1pt}
\centering
\begin{adjustbox}{width=0.85\textwidth}
\begin{tabular}{ccccc}
  & \begin{tabular}{@{}c@{}} trained on fastMRI\\  (supervised) \\[5pt] \footnotesize SSIM: 0.6709 \end{tabular}  & \begin{tabular}{@{}c@{}} trained on fastMRI + TTT \\  (self-supervision included) \\[5pt] \footnotesize SSIM: 0.6902 \end{tabular} & \begin{tabular}{@{}c@{}} trained on Stanford \\ (supervised) \\[5pt] \footnotesize SSIM: 0.6918 \end{tabular} & \begin{tabular}{@{}c@{}} ground truth\\[25pt] \end{tabular} \\
  \rule{0pt}{3ex}
  U-Net &
  \begin{adjustbox}{valign=m}
    \begin{tikzpicture}[boximg]
      \node[anchor=south] (img) {\scalebox{1}[-1]{\includegraphics[width=0.25\textwidth]{files/dataset_shift/fs_sup_on_stan.pdf}}};
      \begin{scope}[x=10,y=10]
        \node[draw,minimum height=0.55cm,minimum width=0.65cm] (B1) at (2.6,3.5) {};
        \node (img1) at (-4.2,1.4) {\scalebox{1}[-1]{\includegraphics[width=0.07\textwidth,trim={4.8cm 4.7cm 1.5cm 1.8cm}, clip]{files/dataset_shift/fs_sup_on_stan.pdf}}};
        \draw (img1.south west) rectangle (img1.north east);
        \draw (B1) -- (img1);
      \end{scope}
  \end{tikzpicture} \end{adjustbox} &
  \begin{adjustbox}{valign=m}
    \begin{tikzpicture}[boximg]
      \node[anchor=south] (img) {\scalebox{1}[-1]{\includegraphics[width=0.25\textwidth]{files/dataset_shift/ttt_on_stan.pdf}}};
      \begin{scope}[x=10,y=10]
        \node[draw,minimum height=0.55cm,minimum width=0.65cm] (B1) at (2.6,3.5) {};
        \node (img1) at (-4.2,1.4) {\scalebox{1}[-1]{\includegraphics[width=0.07\textwidth,trim={4.8cm 4.7cm 1.5cm 1.8cm}, clip]{files/dataset_shift/ttt_on_stan.pdf}}};
        \draw (img1.south west) rectangle (img1.north east);
        \draw (B1) -- (img1);
      \end{scope}
  \end{tikzpicture} \end{adjustbox} &
  \begin{adjustbox}{valign=m}
    \begin{tikzpicture}[boximg]
      \node[anchor=south] (img) {\scalebox{1}[-1]{\includegraphics[width=0.25\textwidth]{files/dataset_shift/stan_sup_on_stan.pdf}}};
      \begin{scope}[x=10,y=10]
        \node[draw,minimum height=0.55cm,minimum width=0.65cm] (B1) at (2.6,3.5) {};
        \node (img1) at (-4.2,1.4) {\scalebox{1}[-1]{\includegraphics[width=0.07\textwidth,trim={4.8cm 4.7cm 1.5cm 1.8cm}, clip]{files/dataset_shift/stan_sup_on_stan.pdf}}};
        \draw (img1.south west) rectangle (img1.north east);
        \draw (B1) -- (img1);
      \end{scope}
  \end{tikzpicture} \end{adjustbox} &
  \begin{adjustbox}{valign=m}
    \begin{tikzpicture}[boximg]
      \node[anchor=south] (img) {\scalebox{1}[-1]{\includegraphics[width=0.25\textwidth]{files/dataset_shift/orig.pdf}}};
      \begin{scope}[x=10,y=10]
        \node[draw,minimum height=0.55cm,minimum width=0.65cm] (B1) at (2.6,3.5) {};
        \node (img1) at (-4.2,1.4) {\scalebox{1}[-1]{\includegraphics[width=0.07\textwidth,trim={4.8cm 4.7cm 1.5cm 1.8cm}, clip]{files/dataset_shift/orig.pdf}}};
        \draw (img1.south west) rectangle (img1.north east);
        \draw (B1) -- (img1);
      \end{scope}
  \end{tikzpicture} \end{adjustbox} \\
  \rule{0pt}{3ex}
  & \footnotesize SSIM: 0.6863 & \footnotesize SSIM: 0.7050 & \footnotesize SSIM: 0.7236 & \\
  \rule{0pt}{3ex}
  VarNet &
  \begin{adjustbox}{valign=m}
    \begin{tikzpicture}[boximg]
      \node[anchor=south] (img) {\scalebox{1}[-1]{\includegraphics[width=0.25\textwidth]{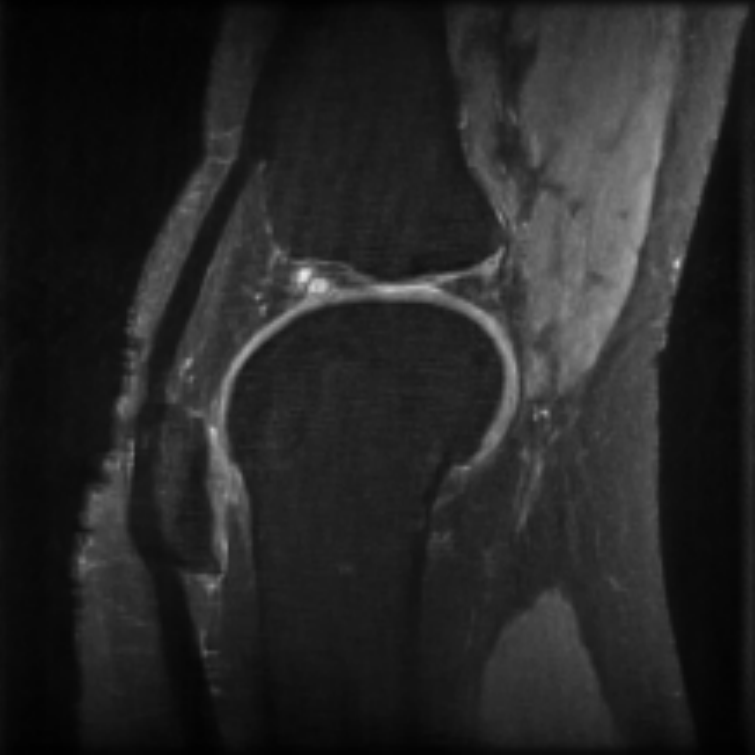}}};
      \begin{scope}[x=10,y=10]
        \node[draw,minimum height=0.55cm,minimum width=0.65cm] (B1) at (2.6,3.5) {};
        \node (img1) at (-4.2,1.4) {\scalebox{1}[-1]{\includegraphics[width=0.07\textwidth,trim={4.8cm 4.7cm 1.5cm 1.8cm}, clip]{files/dataset_shift/vfs_sup_on_stan.pdf}}};
        \draw (img1.south west) rectangle (img1.north east);
        \draw (B1) -- (img1);
      \end{scope}
  \end{tikzpicture} \end{adjustbox} &
  \begin{adjustbox}{valign=m}
    \begin{tikzpicture}[boximg]
      \node[anchor=south] (img) {\scalebox{1}[-1]{\includegraphics[width=0.25\textwidth]{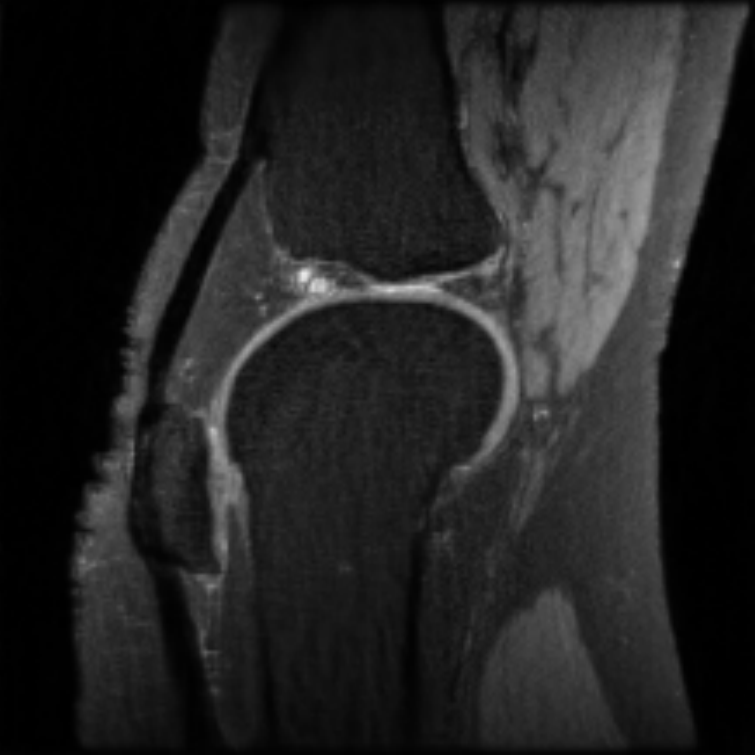}}};
      \begin{scope}[x=10,y=10]
        \node[draw,minimum height=0.55cm,minimum width=0.65cm] (B1) at (2.6,3.5) {};
        \node (img1) at (-4.2,1.4) {\scalebox{1}[-1]{\includegraphics[width=0.07\textwidth,trim={4.8cm 4.7cm 1.5cm 1.8cm}, clip]{files/dataset_shift/vttt_on_stan.pdf}}};
        \draw (img1.south west) rectangle (img1.north east);
        \draw (B1) -- (img1);
      \end{scope}
  \end{tikzpicture} \end{adjustbox} &
  \begin{adjustbox}{valign=m}
    \begin{tikzpicture}[boximg]
      \node[anchor=south] (img) {\scalebox{1}[-1]{\includegraphics[width=0.25\textwidth]{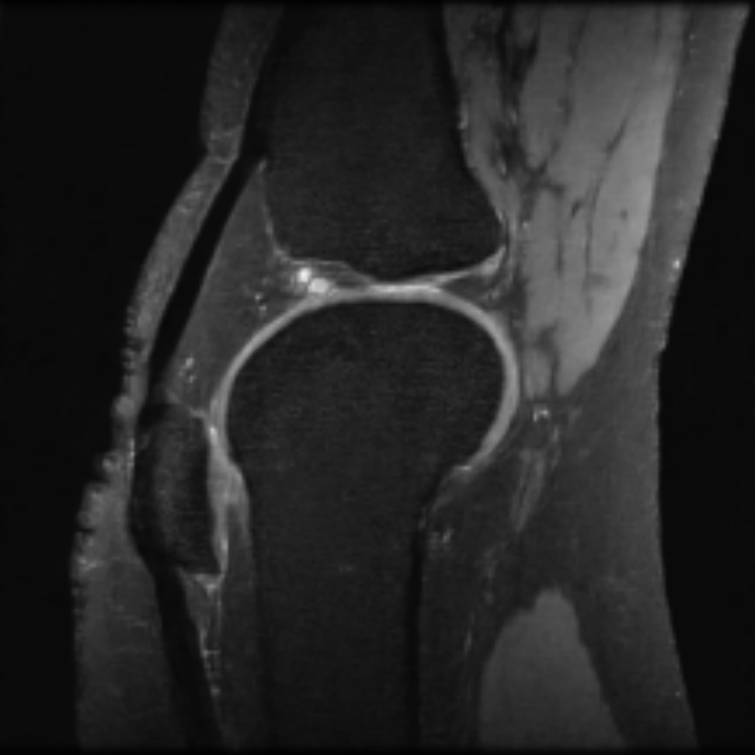}}};
      \begin{scope}[x=10,y=10]
        \node[draw,minimum height=0.55cm,minimum width=0.65cm] (B1) at (2.6,3.5) {};
        \node (img1) at (-4.2,1.4) {\scalebox{1}[-1]{\includegraphics[width=0.07\textwidth,trim={4.8cm 4.7cm 1.5cm 1.8cm}, clip]{files/dataset_shift/vstan_sup_on_stan.pdf}}};
        \draw (img1.south west) rectangle (img1.north east);
        \draw (B1) -- (img1);
      \end{scope}
  \end{tikzpicture} \end{adjustbox} &
  \begin{adjustbox}{valign=m}
    \begin{tikzpicture}[boximg]
      \node[anchor=south] (img) {\scalebox{1}[-1]{\includegraphics[width=0.25\textwidth]{files/dataset_shift/orig.pdf}}};
      \begin{scope}[x=10,y=10]
        \node[draw,minimum height=0.55cm,minimum width=0.65cm] (B1) at (2.6,3.5) {};
        \node (img1) at (-4.2,1.4) {\scalebox{1}[-1]{\includegraphics[width=0.07\textwidth,trim={4.8cm 4.7cm 1.5cm 1.8cm}, clip]{files/dataset_shift/orig.pdf}}};
        \draw (img1.south west) rectangle (img1.north east);
        \draw (B1) -- (img1);
      \end{scope}
  \end{tikzpicture} \end{adjustbox}
\end{tabular}
\end{adjustbox}
\captionsetup{skip=10pt}
\captionof{figure}{Including self-supervision while training DL models combined with TTT improves model robustness to natural dataset shifts. The sample belongs to the Stanford validation dataset and the pointed region reveals how each setup shines or fails at reconstruction.}
\label{fig:data-shift}
\end{table*}

\begin{table*}[t!]
\setlength{\tabcolsep}{1pt}
\centering
\begin{adjustbox}{width=0.85\textwidth}
\begin{tabular}{ccccc}
  & \begin{tabular}{@{}c@{}} trained on AXT2\\  (supervised) \\[5pt] \footnotesize SSIM: 0.8705 \end{tabular}  & \begin{tabular}{@{}c@{}}  trained on AXT2 + TTT \\ (self-supervision included) \\[5pt] \footnotesize SSIM: 0.9317 \end{tabular} & \begin{tabular}{@{}c@{}}   trained on AXT1PRE \\ (supervised) \\[5pt] \footnotesize SSIM: 0.9248 \end{tabular} & \begin{tabular}{@{}c@{}} ground truth\\[25pt] \end{tabular} \\
  \rule{0pt}{3ex}
  U-Net &
  \begin{adjustbox}{valign=m}
    \begin{tikzpicture}[boximg]
      \node[anchor=south] (img) {\includegraphics[width=0.25\textwidth]{files/acquisition_shift/t2_sup_on_t1.pdf}};
      \begin{scope}[x=10,y=10]
        \node[draw,minimum height=0.55cm,minimum width=0.65cm] (B1) at (-1.2,6.7) {}; 
        \node (img1) at (-4.2,1.4) {\includegraphics[width=0.07\textwidth,trim={2.3cm 3.8cm 4cm 2.7cm}, clip]{files/acquisition_shift/t2_sup_on_t1.pdf}};
        \draw (img1.south west) rectangle (img1.north east);
        \draw (B1) -- (img1);
      \end{scope}
    \end{tikzpicture} \end{adjustbox} &
  \begin{adjustbox}{valign=m}
    \begin{tikzpicture}[boximg]
      \node[anchor=south] (img) {\includegraphics[width=0.25\textwidth]{files/acquisition_shift/ttt_on_t1.pdf}};
      \begin{scope}[x=10,y=10]
        \node[draw,minimum height=0.55cm,minimum width=0.65cm] (B1) at (-1.2,6.7) {};
        \node (img1) at (-4.2,1.4) {\includegraphics[width=0.07\textwidth,trim={2.3cm 3.8cm 4cm 2.7cm}, clip]{files/acquisition_shift/ttt_on_t1.pdf}};
        \draw (img1.south west) rectangle (img1.north east);
        \draw (B1) -- (img1);
      \end{scope}
    \end{tikzpicture} \end{adjustbox} &
  \begin{adjustbox}{valign=m}
    \begin{tikzpicture}[boximg]
      \node[anchor=south] (img) {\includegraphics[width=0.25\textwidth]{files/acquisition_shift/t1_sup_on_t1.pdf}};
      \begin{scope}[x=10,y=10]
        \node[draw,minimum height=0.55cm,minimum width=0.65cm] (B1) at (-1.2,6.7) {};
        \node (img1) at (-4.2,1.4) {\includegraphics[width=0.07\textwidth,trim={2.3cm 3.8cm 4cm 2.7cm}, clip]{files/acquisition_shift/t1_sup_on_t1.pdf}};
        \draw (img1.south west) rectangle (img1.north east);
        \draw (B1) -- (img1);
      \end{scope}
    \end{tikzpicture} \end{adjustbox} &
  \begin{adjustbox}{valign=m}
    \begin{tikzpicture}[boximg]
      \node[anchor=south] (img) {\includegraphics[width=0.25\textwidth]{files/acquisition_shift/orig.pdf}};
      \begin{scope}[x=10,y=10]
        \node[draw,minimum height=0.55cm,minimum width=0.65cm] (B1) at (-1.2,6.7) {};
        \node (img1) at (-4.2,1.4) {\includegraphics[width=0.07\textwidth,trim={2.3cm 3.8cm 4cm 2.7cm}, clip]{files/acquisition_shift/orig.pdf}};
        \draw (img1.south west) rectangle (img1.north east);
        \draw (B1) -- (img1);
      \end{scope}
    \end{tikzpicture} \end{adjustbox} \\
  \rule{0pt}{3ex}
  & \footnotesize SSIM: 0.9089 & \footnotesize SSIM: 0.9325 & \footnotesize SSIM: 0.9299 & \\
  \rule{0pt}{3ex}
  VarNet &
  \begin{adjustbox}{valign=m}
    \begin{tikzpicture}[boximg]
      \node[anchor=south] (img) {\includegraphics[width=0.25\textwidth]{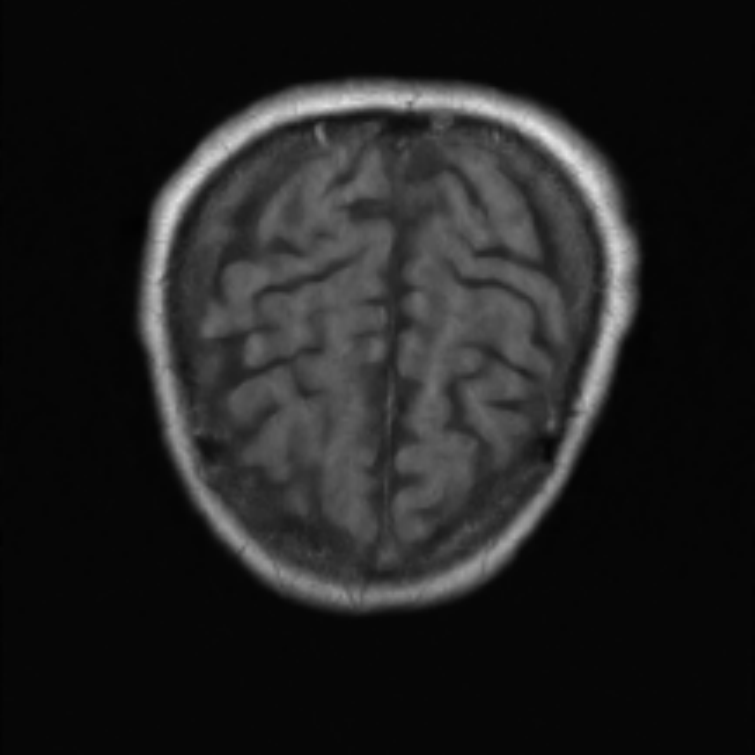}};
      \begin{scope}[x=10,y=10]
        \node[draw,minimum height=0.55cm,minimum width=0.65cm] (B1) at (-1.2,6.7) {};
        \node (img1) at (-4.2,1.4) {\includegraphics[width=0.07\textwidth,trim={2.3cm 3.8cm 4cm 2.7cm}, clip]{files/acquisition_shift/vt2_sup_on_t1.pdf}};
        \draw (img1.south west) rectangle (img1.north east);
        \draw (B1) -- (img1);
      \end{scope}
    \end{tikzpicture} \end{adjustbox} &
  \begin{adjustbox}{valign=m}
    \begin{tikzpicture}[boximg]
      \node[anchor=south] (img) {\includegraphics[width=0.25\textwidth]{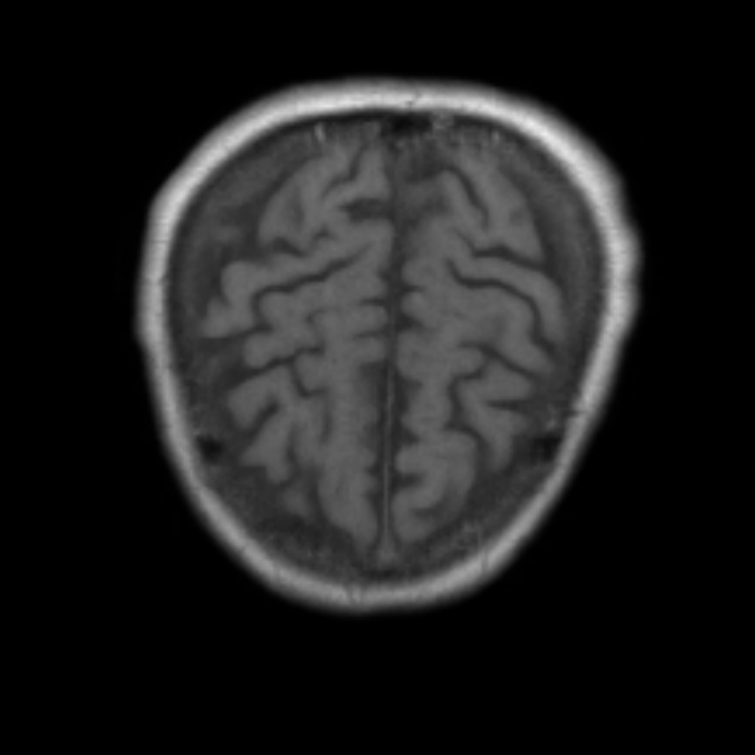}};
      \begin{scope}[x=10,y=10]
        \node[draw,minimum height=0.55cm,minimum width=0.65cm] (B1) at (-1.2,6.7) {};
        \node (img1) at (-4.2,1.4) {\includegraphics[width=0.07\textwidth,trim={2.3cm 3.8cm 4cm 2.7cm}, clip]{files/acquisition_shift/vttt_on_t1.pdf}};
        \draw (img1.south west) rectangle (img1.north east);
        \draw (B1) -- (img1);
      \end{scope}
    \end{tikzpicture} \end{adjustbox} &
  \begin{adjustbox}{valign=m}
    \begin{tikzpicture}[boximg]
      \node[anchor=south] (img) {\includegraphics[width=0.25\textwidth]{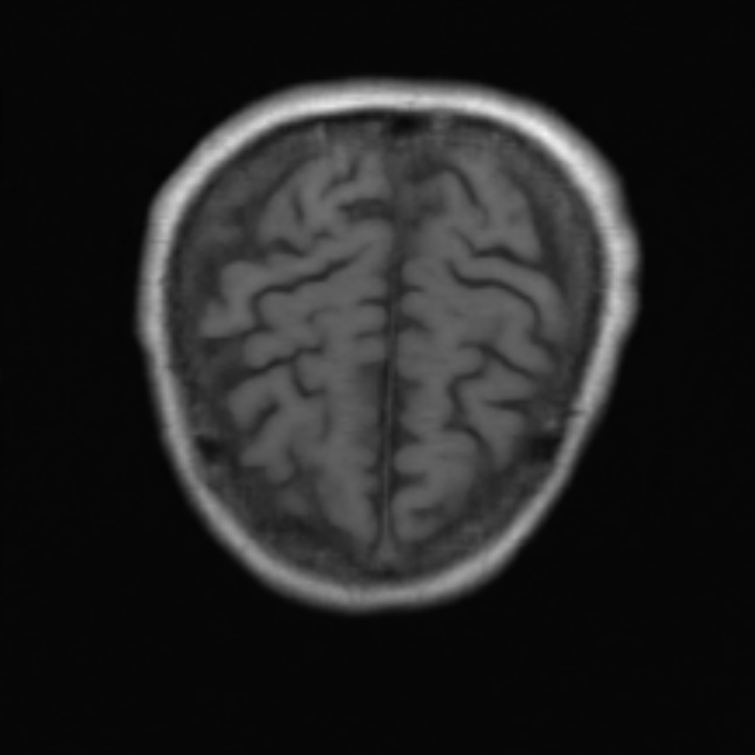}};
      \begin{scope}[x=10,y=10]
        \node[draw,minimum height=0.55cm,minimum width=0.65cm] (B1) at (-1.2,6.7) {};
        \node (img1) at (-4.2,1.4) {\includegraphics[width=0.07\textwidth,trim={2.3cm 3.8cm 4cm 2.7cm}, clip]{files/acquisition_shift/vt1_sup_on_t1.pdf}};
        \draw (img1.south west) rectangle (img1.north east);
        \draw (B1) -- (img1);
      \end{scope}
    \end{tikzpicture} \end{adjustbox} &
  \begin{adjustbox}{valign=m}
    \begin{tikzpicture}[boximg]
      \node[anchor=south] (img) {\includegraphics[width=0.25\textwidth]{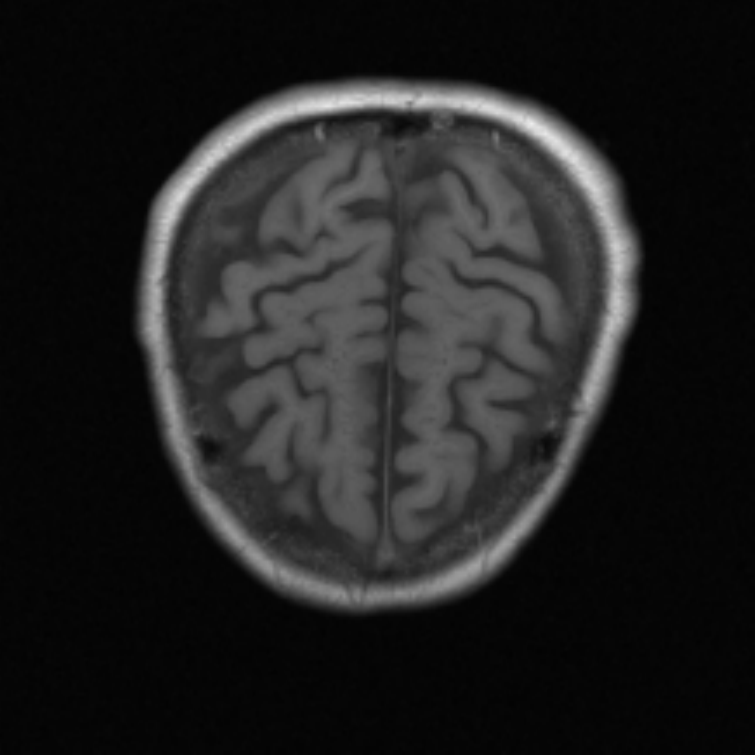}};
      \begin{scope}[x=10,y=10]
        \node[draw,minimum height=0.55cm,minimum width=0.65cm] (B1) at (-1.2,6.7) {};
        \node (img1) at (-4.2,1.4) {\includegraphics[width=0.07\textwidth,trim={2.3cm 3.8cm 4cm 2.7cm}, clip]{files/acquisition_shift/vorig.pdf}};
        \draw (img1.south west) rectangle (img1.north east);
        \draw (B1) -- (img1);
      \end{scope}
    \end{tikzpicture} \end{adjustbox}
\end{tabular}
\end{adjustbox}
\captionsetup{skip=10pt}
\captionof{figure}{Including self-supervision while training DL models combined with TTT improves model robustness to natural modality shifts. The AXT1PRE sample belongs to the fastMRI brain validation dataset.}
\label{fig:acq-shift}
\end{table*}

\begin{table*}[h!]
\setlength{\tabcolsep}{1pt}
\centering
\begin{adjustbox}{width=0.85\textwidth}
\begin{tabular}{ccccc}
  & \begin{tabular}{@{}c@{}} trained on 4x\\  (supervised) \\[5pt]  SSIM: 0.8572 \end{tabular}  & \begin{tabular}{@{}c@{}} trained on 4x + TTT \\  (self-supervision included) \\[5pt] \footnotesize SSIM: 0.9107 \end{tabular} & \begin{tabular}{@{}c@{}}  trained on 2x \\ (supervised) \\[5pt] \footnotesize SSIM: 0.9122 \end{tabular} & \begin{tabular}{@{}c@{}} ground truth\\[25pt] \end{tabular} \\
  \rule{0pt}{3ex}
  U-Net &
  \begin{adjustbox}{valign=m}
    \begin{tikzpicture}[boximg]
      \node[anchor=south] (img) {\scalebox{1}[-1]{\includegraphics[width=0.25\textwidth]{files/acceleration_shift/4x_sup_on_2x.pdf}}};
      \begin{scope}[x=10,y=10]
        \node[draw,minimum height=0.55cm,minimum width=0.65cm] (B1) at (-2.4,6.9) {};
        \node (img1) at (-4.2,1.4) {\scalebox{1}[-1]{\includegraphics[width=0.07\textwidth,trim={1.5cm 2.6cm 4.8cm 3.9cm}, clip]{files/acceleration_shift/4x_sup_on_2x.pdf}}};
        \draw (img1.south west) rectangle (img1.north east);
        \draw (B1) -- (img1);
      \end{scope}
  \end{tikzpicture} \end{adjustbox} &
  \begin{adjustbox}{valign=m}
    \begin{tikzpicture}[boximg]
      \node[anchor=south] (img) {\scalebox{1}[-1]{\includegraphics[width=0.25\textwidth]{files/acceleration_shift/4x_self_on_2x_ttt.pdf}}};
      \begin{scope}[x=10,y=10]
        \node[draw,minimum height=0.55cm,minimum width=0.65cm] (B1) at (-2.4,6.9) {};
        \node (img1) at (-4.2,1.4) {\scalebox{1}[-1]{\includegraphics[width=0.07\textwidth,trim={1.5cm 2.6cm 4.8cm 3.9cm}, clip]{files/acceleration_shift/4x_self_on_2x_ttt.pdf}}};
        \draw (img1.south west) rectangle (img1.north east);
        \draw (B1) -- (img1);
      \end{scope}
  \end{tikzpicture} \end{adjustbox} &
  \begin{adjustbox}{valign=m}
    \begin{tikzpicture}[boximg]
      \node[anchor=south] (img) {\scalebox{1}[-1]{\includegraphics[width=0.25\textwidth]{files/acceleration_shift/2x_sup_on_2x.pdf}}};
      \begin{scope}[x=10,y=10]
        \node[draw,minimum height=0.55cm,minimum width=0.65cm] (B1) at (-2.4,6.9) {};
        \node (img1) at (-4.2,1.4) {\scalebox{1}[-1]{\includegraphics[width=0.07\textwidth,trim={1.5cm 2.6cm 4.8cm 3.9cm}, clip]{files/acceleration_shift/2x_sup_on_2x.pdf}}};
        \draw (img1.south west) rectangle (img1.north east);
        \draw (B1) -- (img1);
      \end{scope}
  \end{tikzpicture} \end{adjustbox} &
  \begin{adjustbox}{valign=m}
    \begin{tikzpicture}[boximg]
      \node[anchor=south] (img) {\scalebox{1}[-1]{\includegraphics[width=0.25\textwidth]{files/acceleration_shift/orig.pdf}}};
      \begin{scope}[x=10,y=10]
        \node[draw,minimum height=0.55cm,minimum width=0.65cm] (B1) at (-2.4,6.9) {};
        \node (img1) at (-4.2,1.4) {\scalebox{1}[-1]{\includegraphics[width=0.07\textwidth,trim={1.5cm 2.6cm 4.8cm 3.9cm}, clip]{files/acceleration_shift/orig.pdf}}};
        \draw (img1.south west) rectangle (img1.north east);
        \draw (B1) -- (img1);
      \end{scope}
  \end{tikzpicture} \end{adjustbox} \\
  \rule{0pt}{3ex}
  & \footnotesize SSIM: 0.8804 & \footnotesize SSIM: 0.9038 & \footnotesize SSIM: 0.9173 & \\
  \rule{0pt}{3ex}
  VarNet &
  \begin{adjustbox}{valign=m}
    \begin{tikzpicture}[boximg]
      \node[anchor=south] (img) {\scalebox{1}[-1]{\includegraphics[width=0.25\textwidth]{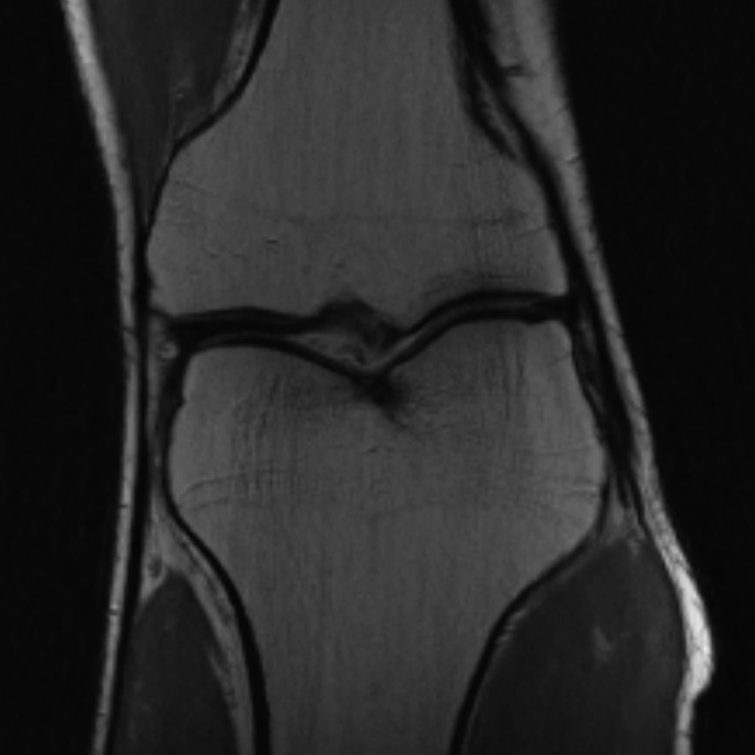}}};
      \begin{scope}[x=10,y=10]
        \node[draw,minimum height=0.55cm,minimum width=0.65cm] (B1) at (-2.4,6.9) {}; 
        \node (img1) at (-4.2,1.4) {\scalebox{1}[-1]{\includegraphics[width=0.07\textwidth,trim={1.5cm 2.6cm 4.8cm 3.9cm}, clip]{files/acceleration_shift/v4x_sup_on_2x.pdf}}};
        \draw (img1.south west) rectangle (img1.north east);
        \draw (B1) -- (img1);
      \end{scope}
  \end{tikzpicture} \end{adjustbox} &
  \begin{adjustbox}{valign=m}
    \begin{tikzpicture}[boximg]
      \node[anchor=south] (img) {\scalebox{1}[-1]{\includegraphics[width=0.25\textwidth]{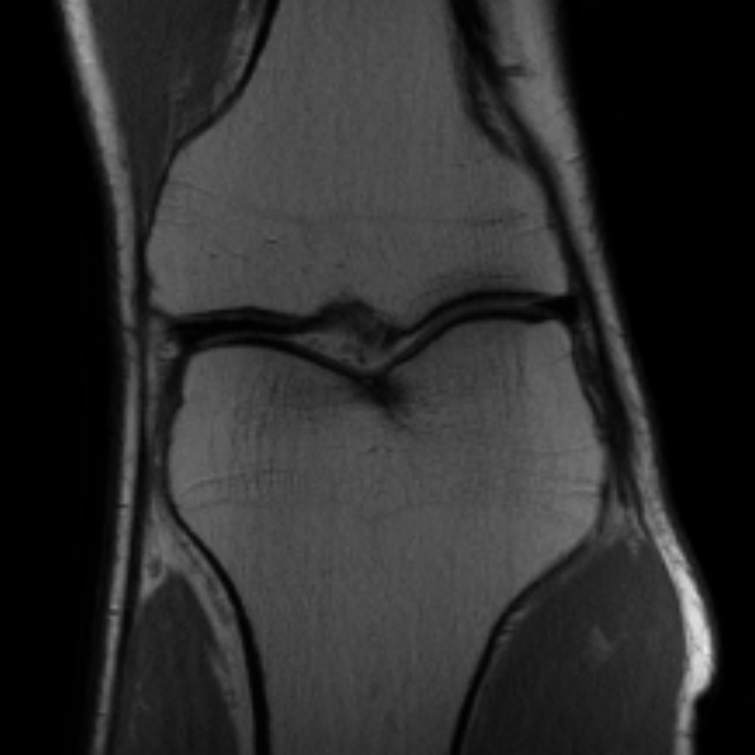}}};
      \begin{scope}[x=10,y=10]
        \node[draw,minimum height=0.55cm,minimum width=0.65cm] (B1) at (-2.4,6.9) {};
        \node (img1) at (-4.2,1.4) {\scalebox{1}[-1]{\includegraphics[width=0.07\textwidth,trim={1.5cm 2.6cm 4.8cm 3.9cm}, clip]{files/acceleration_shift/v4x_self_on_2x_ttt.pdf}}};
        \draw (img1.south west) rectangle (img1.north east);
        \draw (B1) -- (img1);
      \end{scope}
  \end{tikzpicture} \end{adjustbox} &
  \begin{adjustbox}{valign=m}
    \begin{tikzpicture}[boximg]
      \node[anchor=south] (img) {\scalebox{1}[-1]{\includegraphics[width=0.25\textwidth]{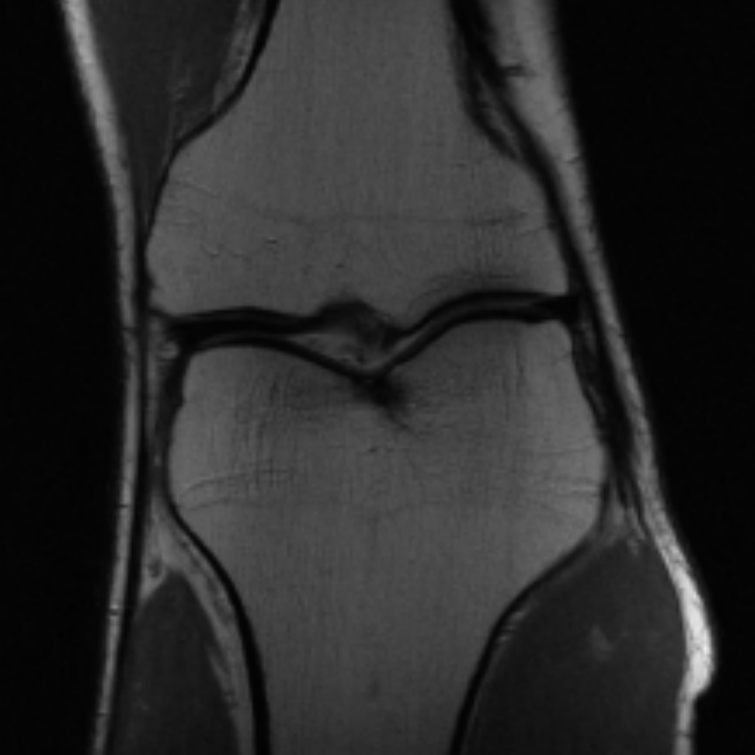}}};
      \begin{scope}[x=10,y=10]
        \node[draw,minimum height=0.55cm,minimum width=0.65cm] (B1) at (-2.4,6.9) {};
        \node (img1) at (-4.2,1.4) {\scalebox{1}[-1]{\includegraphics[width=0.07\textwidth,trim={1.5cm 2.6cm 4.8cm 3.9cm}, clip]{files/acceleration_shift/v2x_sup_on_2x.pdf}}};
        \draw (img1.south west) rectangle (img1.north east);
        \draw (B1) -- (img1);
      \end{scope}
  \end{tikzpicture} \end{adjustbox} &
  \begin{adjustbox}{valign=m}
    \begin{tikzpicture}[boximg]
      \node[anchor=south] (img) {\scalebox{1}[-1]{\includegraphics[width=0.25\textwidth]{files/acceleration_shift/orig.pdf}}};
      \begin{scope}[x=10,y=10]
        \node[draw,minimum height=0.55cm,minimum width=0.65cm] (B1) at (-2.4,6.9) {};
        \node (img1) at (-4.2,1.4) {\scalebox{1}[-1]{\includegraphics[width=0.07\textwidth,trim={1.5cm 2.6cm 4.8cm 3.9cm}, clip]{files/acceleration_shift/orig.pdf}}};
        \draw (img1.south west) rectangle (img1.north east);
        \draw (B1) -- (img1);
      \end{scope}
  \end{tikzpicture} \end{adjustbox}
\end{tabular}
\end{adjustbox}
\captionsetup{skip=10pt}
\captionof{figure}{Including self-supervision while training DL models combined with TTT improves model robustness to natural acceleration shifts. The sample belongs to the fastMRI knee validation dataset and the pointed region reveals how each setup shines or fails at reconstruction. VarNet, unlike U-Net, does not include a region of artifact and the overall contrast of the image has changed under the shift.}
\label{fig:acc-shift}
\end{table*}

\end{document}